\documentclass[a4paper,USenglish,cleveref,nameinlink,autoref,thm-restate]{lipics-v2021}

\nolinenumbers
\hideLIPIcs

\usepackage[inline]{enumitem}
\usepackage{amsfonts,amsthm,amsmath}
\usepackage{amstext}
\usepackage{algorithmic}
\usepackage{graphics}
\usepackage[all]{nowidow}
\usepackage{microtype}
\usepackage[bottom]{footmisc}
\usepackage{hyperref}
\usepackage[forwardlinking=yes]{apxproof} 

\newtheorem{reductionrule}{Rule}
\newtheorem{property}[theorem]{Property}

\Crefname{theorem}{Theorem}{Theorems}
\Crefname{theorem2}{Theorem}{Theorems}
\Crefname{lemma}{Lemma}{Lemmas}
\Crefname{lemma2}{Lemma}{Lemmas}
\Crefname{figure}{Fig.}{Figs.}
\Crefname{algorithm}{Listing}{Listings}
\Crefname{section}{Section}{Sections}
\Crefname{observation}{Observation}{Observations}
\Crefname{property}{Property}{Properties}
\Crefname{lemma}{Lemma}{Lemmas}
\Crefname{claim}{Claim}{Claims}
\Crefname{claimx}{Claim}{Claims}
\Crefname{figure}{Fig.}{Figs.}
\Crefname{subfigure}{Fig.}{Figs.} 
\Crefname{minipage}{Fig.}{Figs.}
\Crefname{enumi}{Condition}{Conditions}
\Crefname{reductionrule}{Rule}{Rules}

\usepackage{xcolor}
\definecolor{realblue}{rgb}{0,0,1}
\definecolor{defblue}{rgb}{0.274,0.392,0.666}
\definecolor{darkerblue}{rgb}{0.094,0.455,0.804}
\definecolor{darkblue}{rgb}{0.063,0.306,0.545}
\definecolor{linkblue}{rgb}{0.098,0.098,0.4392}
\definecolor{red}{rgb}{0.627,0.117,0.156}
\definecolor{green}{RGB}{51, 153, 102}
\definecolor{orange}{rgb}{0.903,0.739,0.382}
\definecolor{realred}{rgb}{1,0,0}
\definecolor{lipicsblue}{rgb}{0.08235294118,0.3098039216,0.537254902}

\usepackage{cite} 
\usepackage{xspace} 
\usepackage{todonotes}

\let\emph\relax
\DeclareTextFontCommand{\emph}{\color{defblue}\em}
\DeclareTextFontCommand{\bl}{\color{lipicsblue}}
\hypersetup{colorlinks=true,
    linkcolor=lipicsblue,
    anchorcolor=lipicsblue,
    citecolor=lipicsblue,
    filecolor=lipicsblue,
    menucolor=lipicsblue,
    urlcolor=lipicsblue,
    bookmarksopen=true,
    bookmarksopenlevel=2,
    bookmarksnumbered=true,
    plainpages=false,
    }

\newcommand{\bridge}{\mathrm{bridge}}

\newcommand{\NP}{\textsf{NP}\xspace}

\newcommand{\wlpF}[1]{{\sc $#1$-Span Weakly leveled planarity}\xspace}

\DeclareMathOperator{\spn}{span}

\newcommand{\att}{\mathrm{att}}
\newcommand{\trim}{\mathrm{trim}}

\author{Michael A. Bekos}{University of Ioannina, Greece}{bekos@uoi.gr}{https://orcid.org/0000-0002-3414-7444}{}
  
\author{Giordano {Da Lozzo}}{Roma Tre University, Italy}{giordano.dalozzo@uniroma3.it}{http://orcid.org/0000-0003-2396-5174}{}

\author{Fabrizio Frati}{Roma Tre University, Italy}{fabrizio.frati@uniroma3.it}{https://orcid.org/0000-0001-5987-8713}{}

\author{Siddharth Gupta}{BITS Pilani, K K Birla Goa Campus, India}{siddharthg@goa.bits-pilani.ac.in}{https://orcid.org/0000-0003-4671-9822}{}

\author{Philipp Kindermann}{Trier University, Germany}{kindermann@uni-trier.de}{https://orcid.org/0000-0001-5764-7719}{}

\author{Giuseppe Liotta}{Perugia University, Italy}{giuseppe.liotta@unipg.it}{https://orcid.org/0000-0002-2886-9694}{}

\author{Ignaz Rutter}{University of Passau, Germany}{rutter@fim.uni-passau.de}{https://orcid.org/0000-0002-3794-4406}{}

\author{Ioannis G. Tollis}{University of Crete, Greece}{tollis@csd.uoc.gr}{https://orcid.org/0000-0002-5507-7692}{}
  
\funding{Research by {Da Lozzo}, {Frati}, and {Liotta} was supported, in part,  by MUR of Italy (PRIN Project no.~2022ME9Z78~-- NextGRAAL and PRIN Project no.~2022TS4Y3N~-- EXPAND). Research by {Liotta} was also supported in part by MUR PON Proj. ARS01\_00540. Research by Rutter was funded by the Deutsche Forschungsgemeinschaft (DFG, German Research Foundation) -- 541433306.}

\authorrunning{Bekos et al.}

\Copyright{Michael A. Bekos, Giordano {Da Lozzo}, Fabrizio Frati, Siddharth Gupta, Philipp Kindermann, Giuseppe Liotta, Ignaz Rutter, Ioannis G. Tollis}



\acknowledgements{This work started at the Graph and Network Visualization Workshop 2023 (GNV '23), 25 - 30 June, Chania, Greece.}

\title{Weakly Leveled Planarity with Bounded Span}

\begin{document}

\maketitle

\begin{abstract}
This paper studies planar drawings of graphs in which each vertex is represented as a point along a sequence of horizontal lines, called levels, and each edge is either a horizontal segment or a strictly $y$-monotone curve. A graph is $s$-span weakly leveled planar if it admits such a drawing where the edges have span at most $s$; the span of an edge is the number of levels it touches minus one. 
We investigate the problem of computing $s$-span weakly leveled planar drawings from both the computational and the combinatorial perspectives. 
We prove the problem to be para-NP-hard with respect to its natural parameter $s$ and investigate its complexity with respect to widely used structural parameters. We show the existence of a polynomial-size kernel with respect to vertex cover number and prove that the problem is FPT when parameterized by treedepth. We also present upper and lower bounds on the span for various graph classes.
%
 Notably, we show that cycle trees, a family of $2$-outerplanar graphs generalizing Halin graphs, are $\Theta(\log n)$-span weakly leveled planar and $4$-span weakly leveled planar when $3$-connected. As a byproduct of these combinatorial results, we obtain improved bounds on the edge-length ratio of the graph families under consideration. 
\end{abstract}

\section{Introduction}
Computing crossing-free drawings of planar graphs is at the heart of Graph Drawing. Indeed, since the seminal papers by Fáry~\cite{Far48} and by Tutte~\cite{tutte1963draw} were published, a rich body of literature has been devoted to the study of crossing-free drawings of planar graphs that satisfy a variety of optimization criteria, including  the area~\cite{DBLP:journals/algorithmica/Kant96,DBLP:journals/combinatorica/FraysseixPP90}, the angular resolution~\cite{DBLP:journals/siamcomp/FormannHHKLSWW93,DBLP:journals/siamdm/MalitzP94}, the face convexity~\cite{DBLP:journals/ijcga/ChrobakK97,DBLP:conf/gd/BonichonFM04,DBLP:journals/algorithmica/BonichonFM07}, the total edge length~\cite{DBLP:journals/siamcomp/Tamassia87}, and the edge-length ratio~\cite{DBLP:conf/gd/BorrazzoF19,DBLP:journals/jocg/BorrazzoF20,DBLP:journals/ijcga/BlazjFL21}; see also~\cite{BattistaETT99,Tamassia:13}. 
 

\begin{figure}[b]
    \centering
    \includegraphics[width=\textwidth]{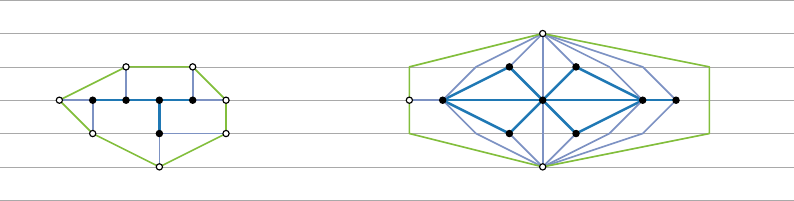}
    \caption{A $1$-span weakly leveled planar drawing of the Frucht graph (left) and a $4$-span weakly leveled planar drawing of the Goldner-Harary graph (right).}
    \label{fig:teaser}
\end{figure}

In this paper, we focus on crossing-free  drawings where the edges are represented as simple Jordan arcs and have the additional constraint of being \emph{$y$-monotone}, that is, traversing each edge from one end-vertex to the other one, the $y$-coordinates never increase or never decrease. 
This leads to a generalization of the well-known \emph{layered drawing} style~\cite{DBLP:journals/algorithmica/DujmovicFKLMNRRWW08,DBLP:conf/gd/BrucknerKM19,DBLP:journals/algorithmica/BrucknerKM22}, where vertices are assigned to horizontal lines, called \emph{levels}, and edges only connect vertices on different levels. We also allow edges between vertices on the same level and seek for drawings of \emph{bounded span}, i.e., in which the edges span few levels. 
In their seminal work~\cite{DBLP:journals/siamcomp/HeathR92}, Heath and Rosenberg study \emph{leveled planar drawings}, i.e., in which edges only connect vertices on consecutive levels and no two edges cross. We also mention the algorithmic framework by Sugiyama et al.~\cite{DBLP:journals/tsmc/SugiyamaTT81}, which yields layered drawings for the so-called hierarchical graphs. In this framework, edges that span more than one level are transformed into paths by inserting a dummy vertex for each level they cross. Hence minimizing the edge span (or equivalently, the number of dummy vertices along the edges) is a relevant~optimization~criterion.

Inspired by these works, we study \emph{$s$-span weakly leveled planar drawings}, which are crossing-free $y$-monotone drawings in which each edge touches at most $s+1$ levels; see \cref{fig:teaser}. We also study \emph{$s$-span leveled planar drawings}, in which the $y$-coordinates (strictly) increase or (strictly) decrease when traversing each edge from one end-vertex to the other one. 

Note that 1-span weakly leveled planar drawings have been studied in different contexts; for example, Bannister et al.~\cite{DBLP:conf/dagstuhl/BastertM99} prove that graphs that admit such drawings have layered pathwidth at most~two\footnote{Bannister et al. use the term weakly leveled planar drawing to mean 1-span weakly leveled planar drawing. We use a different terminology because we allow edges to span more than one level.}. Felsner et al.~\cite{DBLP:conf/gd/FelsnerLW01,DBLP:journals/jgaa/FelsnerLW03} show that every outerplanar graph has a 1-span  weakly leveled planar drawing and use this to compute a 3D drawing of the graph in linear volume; a similar construction by Dujmović et al.~\cite{DBLP:journals/dmtcs/DujmovicPW04} yields a 2-span leveled planar drawing for every outerplanar graph, which can be used to bound the queue number~\cite{DBLP:journals/siamcomp/HeathR92} of these graphs. In general, our work also relates to track layouts~\cite{DBLP:journals/dmtcs/DujmovicPW04} and to the recently-introduced layered decompositions~\cite{DBLP:journals/jacm/DujmovicJMMUW20}, but in contrast to these research works we insist on planarity.


\subparagraph{Our contributions.} We address the problem of computing a weakly leveled planar drawing with bounded span both from the complexity and from the combinatorial perspectives. Specifically, the {\sc $s$-Span (Weakly) leveled planarity} problem asks whether a graph admits a (weakly) leveled planar drawing where the span of every edge is at most $s$.  
The main contributions of this paper can be summarized as follows.

\begin{itemize}
\item 
In \cref{se:hardness}, we show that the {\sc $s$-Span Weakly leveled planarity} problem is \NP-complete for any fixed $s\ge 1$ (\cref{thm:np-hardness}). Our proof technique implies that {\sc $s$-Span Leveled Planarity} is also \NP-complete. This generalizes the \NP-completeness result by Heath and Rosenberg~\cite{DBLP:journals/siamcomp/HeathR92}  which holds for $s=1$. 

\item 
The para-NP-hardness of {\sc $s$-Span Weakly leveled planarity} parameterized by the span~$s$ motivates the study of FPT approaches with respect to structural parameters of the input graph. In \cref{se:fpt}, we show that the {\sc $s$-Span Weakly leveled planarity} problem has a kernel of polynomial size when parameterized by the vertex cover number (\cref{cor:fpt-vc}) and has a (super-polynomial) kernel when parameterized by the treedepth (\cref{thm:treedepth-FPT}). As also pointed out in~\cite{DBLP:journals/csr/Zehavi22}, designing FPT algorithms parameterized by structural parameters bounded by the vertex cover number, such as the 
treedepth, pathwidth, and treewidth is a challenging research direction in the context of graph drawing (see, e.g., ~\cite{BalkoCG00V022,DBLP:conf/gd/BhoreLMN21,BhoreLMN23,DBLP:conf/gd/BhoreGMN20,BhoreGMN22,DBLP:conf/gd/BhoreGMN19,BhoreGMN20,HlinenyS19,DBLP:conf/compgeom/ChaplickGFGRS22}). Again, our algorithms can also be adapted to work for {\sc $s$-Span Leveled Planarity}.

\item 
In \cref{se:combinatorial}, we give combinatorial bounds on the span of weakly leveled planar drawings of various graph classes. It is known that outerplanar graphs admit weakly leveled planar drawings with span~$1$~\cite{DBLP:journals/jgaa/FelsnerLW03}. We extend the investigation by considering both graphs with outerplanarity 2 and  graphs with treewidth 2. We prove that some 2-outerplanar graphs require a linear span (\cref{thm:2-outerplanar}). Since Halin graphs (which have outerplanarity 2) admit weakly leveled planar drawings with span~$1$~\cite{DBLP:journals/algorithmica/BannisterDDEW19,digiacomo2023new}, we consider 3-connected cycle-trees~\cite{DBLP:conf/isaac/LozzoDEJ17,DBLP:conf/wads/ChaplickLGLM21}, which also have outerplanarity 2 and include Halin graphs as a subfamily. Indeed, while the Halin graphs are those graphs of polyhedra containing a face that shares an edge with every other face, the 3-connected cycle-trees are the graphs of polyhedra containing a face that shares a vertex with every other face. We show that 3-connected cycle-trees have weakly leveled planar drawings with span~$4$ (\cref{lem:upper-cycle-tree-3conn}), which is necessary in the worst case (\cref{lem:lower-cycle-tree-3conn}). For general cycle-trees, we prove a $\Theta(\log n)$ tight bound on the span (\cref{lem:upper-general-cycle-trees,lem:lower-cycle-tree-general}); such a difference between the $3$-connected and $2$-connected case was somewhat surprising for us. Concerning graphs of treewidth 2, we prove an upper bound of $O(\sqrt n)$ and a lower bound of $2^{\Omega(\sqrt{\log n})}$ on the span of their weakly leveled planar drawings (\cref{lem:series-parallel-upper-bound,lem:series-parallel-lower-bound}).
\end{itemize}

\subparagraph{Remarks.} Dujmovi\'c et al.~\cite{DBLP:journals/algorithmica/DujmovicFKLMNRRWW08} present an FPT algorithm to minimize the number of levels in a leveled planar graph drawing, where the parameter is the total number of levels. They claim that they can similarly get an FPT algorithm that minimizes the span in a leveled planar graph drawing, where the parameter is the span. Our algorithm differs from the one of Dujmovi\'c et al.~\cite{DBLP:journals/algorithmica/DujmovicFKLMNRRWW08} in three directions: (i) We optimize the span of a weakly level planar drawing, which is not necessarily optimized by minimizing the span of a leveled planar drawing; (ii) we consider structural parameters rather than a parameter of the drawing; one common point of our three algorithms is to derive a bound on the span {\em from} the bound on the structural parameter; and (iii) our algorithms perform conceptually simple kernelizations, while the one in~\cite{DBLP:journals/algorithmica/DujmovicFKLMNRRWW08} exploits a sophisticated dynamic programming on a path decomposition of the input graph.


Concerning the combinatorial contribution, a byproduct of our results implies new bounds on the planar edge-length ratio~\cite{DBLP:conf/gd/LazardLL17,DBLP:journals/tcs/LazardLL19,DBLP:conf/cccg/HoffmannKKR14} of families of planar graphs. The planar edge-length ratio of a planar graph is the minimum edge-length ratio (that is, the ratio of the longest to the shortest edge) over all planar straight-line drawings of the graph. Borrazzo and Frati~\cite{DBLP:journals/jocg/BorrazzoF20} have proven that the planar edge-length ratio of an $n$-vertex 2-tree is $O(n^{0.695})$. \cref{lem:series-parallel-upper-bound}, together with a result relating the span of a weakly leveled planar drawing to its edge-length ratio~\cite[Lemma 4]{digiacomo2023new} lowers the upper bound of~\cite{DBLP:journals/jocg/BorrazzoF20} to $O(\sqrt n)$ (\cref{cor:series-edgelength}). We analogously get an upper bound of $9$ on the edge-length ratio of $3$-connected cycle-trees (\cref{co:edge-length-cycle-tree-3conn}).

\section{Preliminaries} \label{se:preliminaries}
In the paper, we only consider simple connected graphs, unless otherwise specified. We use standard terminology in the context of graph theory~\cite{Diestelbook} and graph drawing~\cite{BattistaETT99}.

A \emph{drawing} of a graph maps each vertex to a distinct point in the plane and each edge to a Jordan arc between its endpoints. A drawing is \emph{planar} if no two edges intersect, except at a common end-point. In a planar drawing, a vertex or an edge is \emph{external} if it is incident to the outer face, and \emph{internal} otherwise. 
A graph is \emph{planar} if it admits a planar drawing. 
A \emph{plane graph} is a planar graph together with a \emph{planar embedding}, which is an equivalence class of planar drawings, where two planar drawings of a connected graph are equivalent if they have the same clockwise order of the edges incident to each vertex and order of the vertices along the outer face. 
A graph drawing is \emph{$y$-monotone} if each edge is drawn as a strictly $y$-monotone curve and \emph{weakly $y$-monotone} if each edge is drawn as a horizontal segment or as a strictly $y$-monotone curve. 
A drawing is \emph{straight-line} (\emph{polyline}) if each edge is represented by a straight-line segment (resp.\ by a sequence of straight-line segments).

For a positive integer $k$, we denote by $[k]$ the set $\{1,\dots,k\}$.
A \emph{leveling} of a graph $G=(V,E)$ is a function~$\ell \colon V \to [k]$.
A leveling $\ell$ of $G$ is \emph{proper} if, for any edge $(u,v) \in E$, it holds $|\ell(u) - \ell(v)| = 1$, and it is \emph{weakly proper} if $|\ell(u) - \ell(v)| \leq 1$.
For each $i \in [k]$, we define $V_i = \ell^{-1}(i)$ and call it the \emph{$i$-th level} of $\ell$. The \emph{height} of~$\ell$ is $h(\ell)=\max_{v \in V} \ell(v) - \min_{v\in V} \ell(v)$. 
A \emph{level graph} is a pair~$(G,\ell)$, where $G$ is a graph and $\ell$ is a leveling of $G$.  A \emph{(weakly) level planar drawing} of a level graph $(G,\ell)$ is a planar (weakly) $y$-monotone drawing of~$G$ where each vertex is drawn with $y$-coordinate $\ell(v)$.
A level graph $(G,\ell)$ is \emph{(weakly) level planar} if it admits a (weakly) level planar drawing. The \emph{span} of an edge $(u,v)$ of a level graph~$(G,\ell)$ is $\spn_{\ell}(u,v)=|\ell(u)-\ell(v)|$. The \emph{span} of a leveling~$\ell$ of $G$ is $\spn(\ell) = \max_{(u,v) \in E} \spn_{\ell}(u,v)$.  



The following observation rephrases a result of Di Battista and Nardelli~\cite[Lemma~1]{BattistaN88} in the weakly-level planar setting.

\begin{observation}\label{obs:gioproperty}
    Let $(G,\ell)$ be a level graph such that $\ell$ is (weakly) proper.
    For each $i\in[k]$, let $\prec_i$ be a linear ordering  on $\ell^{-1}(i)$. Then, there exists a (weakly) level planar drawing of $(G,\ell)$ that respects $\prec_i$ (i.e., in which the left-to-right ordering of the vertices in $\ell^{-1}(i)$ is $\prec_i$) if and only if:
    \begin{enumerate}
    	\renewcommand{\theenumi}{\roman{enumi}} 
    	\renewcommand{\labelenumi}{\textbf{(\theenumi)}} 
        \item if $(u,v)\in E(G)$ with $\ell(u)=\ell(v)=i$, then $u$ and $v$ are consecutive in $\prec_i$; and
        \item if $(u,v)$ and $(w,x)$ are two independent edges (i.e., $\{u,v\}\cap \{w,x\}=\emptyset$) with $\ell(u)=\ell(w)=i$, $\ell(v)=\ell(x)=i+1$, and $u\prec_{i} w$, then $v\prec_{i+1} x$. 
    \end{enumerate}
\end{observation}
      
  

A \emph{(weakly) leveled planar drawing} of a graph $G$ is a (weakly) level planar drawing of a level graph $(G,\ell)$, for some leveling~$\ell$~of~$G$. An \emph{$s$-span (weakly) leveled planar graph} is a graph that admits a (weakly) leveled planar drawing with span $s$. The $1$-span (weakly) leveled planar graphs are also called (weakly) leveled planar graphs. Given a graph $G$, we consider the problem of finding a leveling~$\ell$ that minimizes $\spn(\ell)$ among all levelings where $(G,\ell)$ is weakly level planar. 
Specifically, given a positive integer $s$, we call {\sc $s$-Span (Weakly) leveled planarity} the problem of testing whether a planar graph $G$ is an $s$-span (weakly) leveled planar graph. 
The {\sc $1$-Span Leveled Planarity} problem has been studied under the name of {\sc Leveled Planar} by Heath and Rosenberg~\cite{DBLP:journals/siamcomp/HeathR92}. The assumption that the input graph is planar is not a loss of generality, since obviously a graph is (weakly) leveled planar graph only if it is planar; from an algorithmic point of view, a preliminary linear-time test can ensure that the input graph is indeed planar, and if not it can conclude that the graph is not an $s$-span (weakly) leveled planar graph, for any value of $s$.

A (weakly) $y$-monotone drawing $\Gamma$ of a graph defines a leveling~$\ell$, called the \emph{associated leveling} of $\Gamma$, where vertices with the same~$y$-coordinate are assigned to the same level and the levels are ordered by increasing $y$-coordinates of the vertices they contain. Thus, the span of an edge $(u,v)$ in~$\Gamma$ is~$\spn_\ell(u,v)$, the span of~$\Gamma$ is~$\spn(\ell)$, and the height of $\Gamma$ is $h(\ell)$.

The following lemma appears implicitly in the proof of Lemma 4 in~\cite{digiacomo2023new}.

\begin{lemma}\label{le:weak-nonweak}
Any graph that admits an $s$-span weakly leveled planar drawing with height $h$ has an $(2s+1)$-span leveled planar drawing with height $2h+1$.
\end{lemma}



A planar drawing of a graph is \emph{outerplanar} if all the vertices are external, and \emph{$2$-outerplanar} if removing the external vertices yields an outerplanar drawing. A graph is \emph{outerplanar} (\emph{$2$-outerplanar}) if it admits an outerplanar drawing (resp.\ $2$-outerplanar drawing). A \emph{$2$-outerplane} graph is a $2$-outerplanar graph with an associated planar embedding which corresponds to $2$-outerplanar drawings. A \emph{cycle-tree} is a $2$-outerplane graph such that removing the external vertices yields a tree. A \emph{Halin graph} is a $3$-connected plane graph $G$ such that removing the external edges yields a tree whose leaves are exactly the external vertices of $G$  (and whose internal vertices have degree at least $3$). Note that Halin graphs form a subfamily of the cycle-trees. 





\section{NP-completeness} \label{se:hardness}

This section is devoted to the proof of the following result.

\begin{theorem} \label{thm:np-hardness}
For any fixed $s \geq 1$, {\sc $s$-Span Weakly leveled planarity} is \NP-complete. 
\end{theorem}

\begin{proof}
The \NP-membership is trivial. In fact, given a graph $G$, 
a non-deterministic Turing machine can guess in polynomial time all possible levelings of the vertices of $G$ to up to $s(n-1)+1$ levels.
Moreover, for each such a leveling $\ell$, in deterministic polynomial time, it is possible to test whether $\spn(\ell) \leq s$ and  whether the level graph $(G,\ell)$ is level planar~\cite{JungerLM98}.

Throughout, given a leveling $\ell_Q$ or a graph $Q$ and a (weakly) level planar drawing $\Gamma_Q$ of the level graph $(Q,\ell_Q)$, we denote by $\prec^Q_i$ the left-to-right order of the vertices of $\ell^{-1}(i)$ in $\Gamma_Q$, with $i \in [h(\ell)+1]$.

To prove the \NP-hardness, we distinguish two cases, based on whether $s=1$ or $s>1$. In both cases, we exploit a linear-time reduction from the {\sc Leveled planar} problem, which was proved to be \NP-complete by Heath and Rosenberg~\cite{DBLP:journals/siamcomp/HeathR92}. 
Recall that, given a planar bipartite graph $H$, the {\sc Leveled Planar} problem asks to determine whether $H$ admits $1$-span leveling $\ell_H$ such that $(H,\ell_H)$ is level planar. In other words, the {\sc Level Planar} problem coincides with {\sc $1$-Span Leveled Planarity}.

{\bf (Case $\mathbf{s=1}$).} Note that a $1$-span leveled planar graph must be bipartite. Starting from a bipartite planar graph $H$, we construct a graph $G$ that is a positive instance of {\sc $1$-Span Weakly Leveled Planarity} if and only if $H$ is a positive instance of {\sc $1$-Span Leveled Planarity}. To this aim, we proceed as follows. We initialize $G=H$. Then, for each edge $(u,v)$ of $H$, we remove $(u,v)$ from $G$, introduce a copy $K(u,v)$ of the complete bipartite planar graph $K_{2,4}$, and identify each of $u$ and $v$ with one of the two vertices in the size-$2$ bipartition class of the vertex set of $K(u,v)$. Clearly, $G$ is planar and bipartite. Moreover, the above reduction can clearly be carried out in polynomial, in fact, linear time. 

In the following, for any edge $(u,v) \in E(H)$, we denote the four vertices of the size-$4$ bipartition class of the vertex set of $K(u,v)$ as $x_{uv}$, $y_{uv}$, $z_{uv}$,~and~$w_{uv}$.

$(\Longleftarrow)$ Suppose first that $H$ admits a leveling $\ell_H$ in $k = h(\ell_H)+1$ levels, with $\spn(\ell_H)=1$, such that $(H,\ell_H)$ is level planar, and let $\Gamma_H$ be a level planar drawing of $(H,\ell_H)$. We show that $G$ admits a leveling $\ell_G$ on $2k$ levels, with $\spn(\ell_G) \leq 1$, such that $(G,\ell_G)$ is weakly level planar, in fact even level planar. The leveling $\ell_G$ is computed as follows.
For each vertex $w \in V(H)$, we set $\ell_G(w) = 2\ell_H(w)$. 
For each vertex $w \in V(G) \setminus V(H)$, i.e., 
$w \in \{x_{uv}, y_{uv}, z_{uv}, w_{uv}\}$ with $(u,v) \in E(H)$, we set $\ell_G(w) = \ell_H(u) + \ell_H(v)$; note that $\ell_G(w) = \min \{\ell_G(u),\ell_G(v)\}+1$, i.e., $w$ is assigned to the level that lies strictly between the levels of $u$ and $v$ in $\ell_G$. 
By construction, $\ell_G$ is proper, and thus $\spn(\ell_G) = 1$,  the vertices of $V(H)$ are assigned to even levels, and the vertices in $V(H) \setminus V(G)$ are assigned to odd levels.

We show that $(G,\ell_G)$ is level planar (and thus weakly level planar) by constructing a level planar drawing $\Gamma_G$ of $(G,\ell_G)$. To this aim, we simply need to define orderings $\prec^G_j$ of the vertices of $G$ at each level $j$, with $j \in [2k]$, in such a way that Conditions (i) and (ii) of \cref{obs:gioproperty} are satisfied.
To obtain $\Gamma_G$, for each $j \in [2k]$, we set $\prec^G_j$ as follows.
If $j$ is even, we set $\prec^G_j = \prec^H_{\frac{j}{2}}$, i.e., the order of the vertices of level $j$ of $\ell_H$ (which are vertices of $G$) is the same as the order of the vertices of level $i$ of $\ell_G$. Instead, if $j$ is odd, we define $\prec^G_j$ as follows. Recall that the $j$-th level of $\ell_G$, with $j$ odd, only contains vertices $q_{u,v}$ with $q \in \{x,y,z,w\}$, $(u,v) \in E(H)$, and $j = \min\{\ell_G(u),\ell_G(v)\}+1$. 
To obtain $\prec^G_j$, we first define a total ordering along level $j$ of all vertices $x_{u,v}$, where $(u,v) \in E(H)$ and $j = \min\{\ell_G(u),\ell_G(v)\}+1$, and then impose that $x_{uv}$, $y_{uv}$, $z_{uv}$, and $w_{uv}$ appear consecutively in $\prec^G_j$ in this left-to-right order. 
The above total ordering is computed as follows:
For every two independent edges $(u,v)$ and $(s,t)$ in $E(H)$ such that $\ell_G(u)=\ell_G(s)=j-1$, $\ell_G(v)=\ell_G(t)=j+1$, and 
$(u,v)$ immediately precedes $(s,t)$ in $\Gamma_H$ in the left-to-right order of the edges between the $(\frac{j-1}{2})$-th and the $(\frac{j+1}{2})$-th level of $\ell_H$,
we impose that $x_{u,v} \prec^G_j x_{s,t}$. 
Next, we show that the leveled drawing defined by $\prec^G_j$, with $j \in [2k]$, is leveled planar. 

\begin{claim}
The drawing $\Gamma_G$ of $(G,\ell_G)$ is level planar.
\end{claim}

\begin{nestedproof}
It suffices to prove that the linear orderings $\prec^G_j$ satisfy  Conditions (i) and (ii) of \cref{obs:gioproperty}.
By the construction of $\ell_G$, we have that no two adjacent vertices of $G$ are assigned to the same level; thus, Condition (i) vacuously holds for $\Gamma_G$. Further, let $(u,a)$ and $(w,b)$ be two independent edges such that $j=\ell_G(u)=\ell_G(w)$ and $\ell_G(a)=\ell_G(b)$. By construction, and up to a simultaneous renaming of $u$ with $a$ and $w$ with $b$, we have that $a \in \{x_{uv}, y_{uv},  z_{uv}, w_{uv}\}$, for some edge $(u,v) \in E(H)$, and $b \in \{x_{wz}, y_{wz},  z_{wz}, w_{wz} \}$, for some edge $(w,z)\in E(H)$. Assume that $\ell_G(a)=j+1$, since the case $\ell_G(a)=j-1$ is symmetric, and assume, without loss of generality, that $u \prec^G_j w$. Then, by the construction of $\Gamma_G$,  we have that $u \prec^H_{\frac{j}{2}} w$. Also, by Condition (ii)  for $\Gamma_H$, we have that $v \prec^H_{\frac{j}{2}+1 } z$. Therefore, by construction of $\Gamma_G$, we have that $a \prec^G_{j+1} b$, which proves Condition (ii) for $\Gamma_G$, and thus concludes the proof.
\end{nestedproof}


$(\Longrightarrow)$ Suppose now that $G$ admits a leveling $\ell_G$, with $\spn(\ell_G) \leq 1$, such that $(G,\ell_G)$ is weakly level planar. We show that $H$ admits a leveling $\ell_H$, with $\spn(H)=1$, such that $(H,\ell_H)$ is level planar. The following property will be useful.

\begin{property}\label{prop:drawingOfK24}
Consider the complete bipartite graph $K_{2,4}$, let $u$ and $v$ be the vertices of the size-$2$ bipartition class of its vertex set, and let $x$, $y$, $z$, and $w$ be the other four vertices. 
Then the levelings $\ell$ with $\spn(\ell)\leq 1$ for which  $(K_{2,4},\ell)$ is weakly level planar are such that $\ell(u)=\ell(v) \pm 2$ and $\ell(x)=\ell(y)=\ell(z)=\ell(w)=\frac{\ell(u)+\ell(v)}{2}$.
\end{property}

\begin{nestedproof}
%
%
Let $\ell$ be a leveling with $\spn(\ell)\leq 1$ such that $(K_{2,4},\ell)$ is weakly level planar and let $\Gamma_{2,4}$ be a weakly level planar drawing of $(K_{2,4},\ell)$. 
We show that (a) $|\ell(u)-\ell(v)| \leq 2$, (b) $|\ell(u)-\ell(v)| \neq 0$, and (c) $|\ell(u)-\ell(v)| \neq 1$, which imply that $\ell(u)=\ell(v) \pm 2$. 
\begin{itemize}
\item If (a) does not hold, then the span of $(u,x)$, of $(x,v)$, or both must be larger than $1$. 
\item Suppose next that (b) does not hold, i.e., $u$ and $v$ are assigned to the same level. Then, by the pigeonhole principle and since $\spn(\ell)\leq 1$, we have that one of the levels $\ell(u)$, $\ell(u)+1$, and $\ell(u)-1$ must contain two vertices in $\{x,y,z,w\}$, say $x$ and $y$. If $\ell(x)=\ell(y)=\ell(u)+1$ or $\ell(x)=\ell(y)=\ell(u)-1$, a contradiction is achieved by the observation that the subgraph of $K_{2,4}$ induced by $u,v,x,y$ is a $4$-cycle for which the leveling $\ell$ is proper, however the graphs that admit a leveled planar drawing on two levels are the forests of caterpillars~\cite{DBLP:conf/acsc/EadesMW86}. If $\ell(x)=\ell(y)=\ell(u)$, then both $x$ and $y$ lie between $u$ and $v$, as otherwise $(u,x)$ would overlap $v$, $(u,y)$ would overlap $v$, $(v,x)$ would overlap $u$, or $(v,y)$ would overlap $u$ in $\Gamma_{2,4}$. If both $x$ and $y$ lie between $u$ and $v$, then $(u,x)$ overlaps $(u,y)$  in $\Gamma_{2,4}$, a contradiction.   
\item Finally, suppose that (c) does not hold, i.e., $u$ and $v$ are assigned to consecutive levels. Since $\spn(\ell)\leq 1$, we have that each of $\{x,y,z,w\}$ is assigned to $\ell(u)$ or $\ell(v)$. If two of $\{x,y,z,w\}$ lie on $\ell(u)$, both to the left (or both to the right) of $u$, then the edges between such vertices and $u$ overlap each other. Similarly, if two of $\{x,y,z,w\}$ lie on $\ell(v)$, both to the left (or both to the right) of $v$, then the edges between such vertices and $v$ overlap each other. It follows that one of $\{x,y,z,w\}$ lies on $\ell(u)$ to the left of $u$, one lies on $\ell(u)$ to the right of $u$, one lies on $\ell(v)$ to the left of $v$, and one lies on $\ell(v)$ to the right of $v$. However, this implies that the edge between $v$ and the vertex on $\ell(u)$ to the left of $u$ and the edge between $u$ and the vertex on $\ell(v)$ to the left of $v$ cross each other. 
\end{itemize}
The proof is concluded by observing that all the vertices $x$, $y$, $w$, and $x$ must be assigned in~$\ell$ to the level $\frac{\ell(u)+\ell(v)}{2}$, given that $\ell(u)=\ell(v) \pm 2$ and that $\spn(\ell)\leq 1$.
\end{nestedproof}

Observe that, by \cref{prop:drawingOfK24}, the levelings $\ell$ for which $(K_{2,4}, \ell)$ is weakly level planar are proper. Thus, any $1$-span weakly leveled planar drawing of $K_{2,4}$ is leveled planar. Also, 
note that every edge of $G$ belongs to $K(u,v)$ for some edge $(u,v) \in E(H)$. Thus, we immediately get the following.


\begin{observation}\label{obs:weakly-and-leveled}
Any $1$-span weakly leveled planar drawing of $G$ is leveled planar.
\end{observation}

By \cref{obs:weakly-and-leveled}, we have that $\ell_G$ is proper. Also by \cref{obs:weakly-and-leveled} and since $G$ is connected, we have that all and only the vertices of $V(H)$ lie in even levels of $\ell_G$. For every vertex $v$ of $H$, we set $\ell_H(v) = \frac{\ell_G(v)}{2}$. By construction, $\spn(\ell_H)=1$. Let now $\Gamma_G$ be a (weakly) level planar drawing of $(G,\ell_G)$. We show how to construct a level planar drawing $\Gamma_H$ of $(H,\ell_H)$. To construct $\Gamma_H$, we simply set $\prec^H_i = \prec^G_{2i}$, for each $i \in [k]$, where $k$ is the number of levels of $\ell_G$.

\begin{claim}
The drawing $\Gamma_H$ of $(H,\ell_L)$ is level planar.
\end{claim}

\begin{nestedproof}
We prove that the orderings $\prec^H_i$ satisfy Conditions (i) and (ii) of \cref{obs:gioproperty}. By 
\cref{prop:drawingOfK24} and by the construction of $\ell_H$, we have that no two adjacent vertices of $H$ are assigned to the same level; thus, Condition (i) vacuously holds for $\Gamma_H$. Further, let $(u,v)$ and $(w,z)$ be two independent edges such that $i=\ell_H(u)=\ell_H(w)$ and $\ell_H(v)=\ell_H(z)$. Assume that $\ell_H(v) = \ell_H(u)+1$, as the case in which $\ell_H(v) = \ell_H(u)-1$ is symmetric, and assume, without loss of generality, that $u \prec^H_i w$. Then, by construction of $\Gamma_H$, we have that $u \prec^G_{2i} w$. In turn, by Condition (ii) for $\Gamma_G$, we have that $x_{uv} \prec^G_{2i+1} x_{wz}$, which, again by Condition (ii) for $\Gamma_G$, implies that $v \prec^G_{2(i+1)} z$ in $\Gamma_G$. The latter and the construction of $\Gamma_H$ imply that $v \prec^H_{i+1} z$. Thus, Condition (ii) holds for $\Gamma_H$, which concludes the proof.
\end{nestedproof}

\begin{figure}[tb!]
\centering
    \begin{subfigure}{\linewidth}
    \centering
    \includegraphics[]{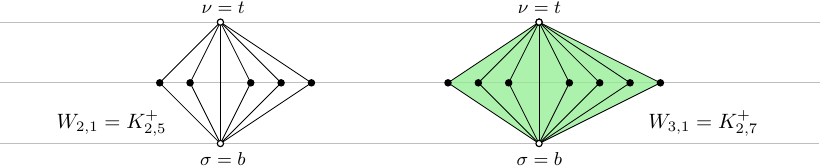}
    \subcaption{The graphs $W_{2,1}$ (left) and $W_{3,1}$ (right).}
    \label{fig:W-base}
    \end{subfigure}
    \begin{subfigure}{\linewidth}
    \centering
    \includegraphics[page=2]{Figures/W.pdf}
    \subcaption{The graph $W_{3,2}$. The seven green shaded regions are copies of $W_{3,1}$.}
    \label{fig:W-inductive}
    \end{subfigure}
    \caption{Illustration for the construction of graphs $W_{i,h}$. Pole vertices are white filled.}
    \label{fig:W-construction}
\end{figure}

{\bf (Case $\mathbf{s>1}$).} In our proof, we will exploit special graphs $W_{i,h}$, where $i$ and $h$ are positive integers such that $h < i$, having two designated vertices $\nu$ and $\sigma$, called \emph{poles}; specifically, $\nu$ is the \emph{north pole} and $\sigma$ is the \emph{south pole} of $W_{i,h}$, see \cref{fig:W-construction}.
In the following, we denote by $K^+_{2,\alpha}$ the graph obtained from the complete bipartite graph $K_{2,\alpha}$ by adding an edge between the two vertices $t$ and $b$ of the size-$2$ bipartition class of the vertex set of $K_{2,\alpha}$. We say that $t$ and $b$ are the \emph{extremes} of $K^+_{2,\alpha}$, and we refer to $t$ and to $b$ as to the \emph{top} and the \emph{bottom extreme}, respectively.
For any $i\geq 2$, the graphs $W_{i,h}$ are defined as follows. 
If $h=1$, the graph $W_{i,1}$ coincides with $K^+_{2,2i+1}$; see \cref{fig:W-base}.
If $h>1$, 
the graph $W_{i,h}$ is obtained from $K^+_{2,2i+1}$ by 
removing each edge $(t,x)$, with $x \neq b$, and by identifying $t$ and $b$ with the north and south pole of a copy of $W_{i,h-1}$, respectively; see \cref{fig:W-inductive}.
We have the following claims.

\begin{claim}\label{claim:drawing-any-s}
For any $i \geq 2$, the graph $W_{i,h}$ admits a leveled planar drawing with span $h+1$ in which the north pole of $W_{i,h}$ lies strictly above all the other vertices of $W_{i,h}$ and the south pole of $W_{i,h}$ lies strictly below all the other vertices of $W_{i,h}$.
\end{claim}

\begin{nestedproof}
\begin{figure}[tb!]
\centering
    \includegraphics[page=3]{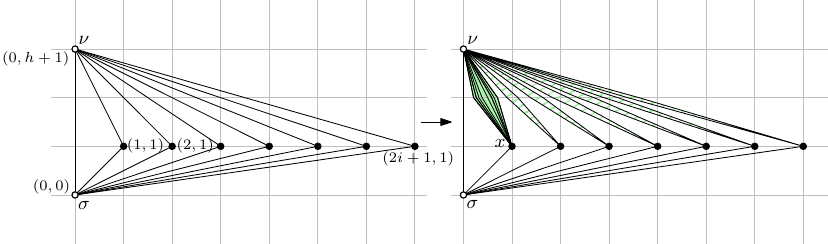}
    \caption{Example for the construction of a drawing of $W_{i,h}$ that satisfies \cref{claim:drawing-any-s}, with $i=3$ and $h=2$. (Left) Initialization of the drawing of $W_{3,2}$ to a straight-line planar drawing of $K^+_{2,7}$. (Right) Replacement of the edges $(\nu,x)$, with $x \neq \sigma$, with a drawing of $W_{3,1}$. For readability purposes, only one edge has been replaced with a drawing of $W_{3,1}$; the other replacements are represented by tailed regions.}
    \label{fig:W-drawings-claim}
\end{figure}
We prove the statement by induction on $h$; refer to \cref{fig:W-drawings-claim}.
In the base case $h=1$. Then $W_{i,1} = K^+_{2,2i+1}$, which admits a level planar drawing with span $h+1=2$ with the desired properties; see \cref{fig:W-base}.
In the inductive case $h>1$. Then, a drawing $\Gamma$ of $W_{i,h}$ with the desired properties can be constructed as follows. We initialize $\Gamma$ to a straight-line planar drawing of $K^+_{2,2i+1}$ in which $\sigma$ is placed on point $(0,0)$, $\nu$ is placed on point $(0,h+1)$, and the remaining $2i+1$ vertices of $K^+_{2,2i+1}$ are placed on points $(j,1)$, for $j = 1,\dots,2i+1$; see \cref{fig:W-drawings-claim}(left).
Observe that, in $\Gamma$, the edge $(\nu,\sigma)$ has span $h+1$, the edges incident to $\sigma$ and not to $\nu$ have span $1$, and the edges incident to $\nu$ and not to $\sigma$ have span $h$. We call the latter edges \emph{long}. 
By induction, the graph $W_{i,h-1}$ admits a leveled planar drawing $\Gamma'$ with span $h$ in which the north pole of $W_{i,h-1}$ lies strictly above all the other vertices of $W_{i,h-1}$ and the south pole of $W_{i,h-1}$ lies strictly below all the other vertices of $W_{i,h-1}$. 
Consider the leveling $\ell'$ of $W_{i,h-1}$ determined by the $y$-coordinates of the vertices in $\Gamma'$.
Without loss of generality, we assume that $\ell'$ assigns the north and the south pole of  $W_{i,h-1}$ to levels $h+1$ and $1$, respectively.
This allows us to replace the drawing of each long edge $(\nu,x)$ (where $x \neq \sigma$ is the south pole of some copy of $W_{i,h-1}$) with a drawing of $W_{i,h-1}$ in which (i) the north pole of $W_{i,h-1}$ lies upon $\nu$ and the south pole of $W_{i,h-1}$ lies upon $x$, and (ii) the vertices of $W_{i,h-1}$ in level $i$ lie arbitrarily close to the intersection point of $(\nu,x)$ and the level $i$, and so that these vertices are consecutive along level $i$ and have the same left-to-right ordering in the resulting drawing as in $\Gamma'$.
This yields a leveled planar drawing $\Gamma$ of $W_{i,h}$ with the desired properties. In particular, the span of $\Gamma$ coincides with the span of the edge $(\sigma,\nu)$, which is $h+1$.
\end{nestedproof}

\begin{claim}\label{cl:uv-long-span}
For any $i \geq 2$, in any weakly leveled planar drawing of $W_{i,i-1}$ with span at most~$i$, the edge connecting the poles of $W_{i,i-1}$ has span $i$.
\end{claim}

\begin{nestedproof}
We start by establishing the following two useful properties of any weakly leveled planar drawing $\Gamma^*$ of $K^+_{2,2i+1}$ with span at most $i$. 

{\bf Property P1:} In $\Gamma^*$, the extremes $b$ and $t$ of $K^+_{2,2i+1}$ lie on different levels. 

{\bf Property P2:} In $\Gamma^*$, there exists a vertex $x$ of $K^+_{2,2i+1}$ that lies strictly between the levels of the extremes $b$ and $t$ of $K^+_{2,2i+1}$.

We prove {\bf property P1}. Suppose, for a contradiction, that $b$ and $t$ are assigned to the same level in $\Gamma^*$. Since  $K^+_{2,2i+1}$ contains the edge $(b,t)$, no vertex of $K^+_{2,2i+1}$ different from $b$ and $t$ is assigned to the same level as $b$ and $t$, as otherwise an edge incident to such a vertex would overlap the edge $(b,t)$. By the pigeonhole principle and since $\spn(\ell)\leq i$, there exist two vertices $a$ and $c$ of $K^+_{2,2i+1}$ that are assigned to a same level, different from the one of $b$~and~$t$. Similarly as in the proof of \cref{prop:drawingOfK24}, this implies that the $4$-cycle $(a,b,c,t)$ crosses itself in $\Gamma^*$, a contradiction. 

\begin{figure}
    \centering
    \begin{subfigure}[t]{0.48\textwidth}
    \centering
        \includegraphics[page=2,scale=.7]{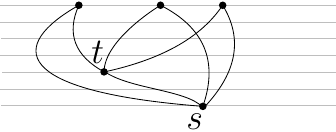}
        \subcaption{}
        \label{fig:more-than-two-level}
    \end{subfigure}
    \begin{subfigure}[t]{0.48\textwidth}
    \centering
        \includegraphics[page=4,scale=.7]{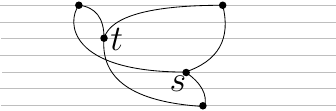}
        \subcaption{}
        \label{fig:two-above-one-below}
    \end{subfigure}
    \caption{Illustration for the proof of property P2 in \cref{cl:uv-long-span}. (a) No more than two vertices lie on the same level in $\Gamma^*$. (b) If two vertices lie on the same level in $\Gamma^*$ and such a level is higher than $\ell^*(t)$, then no vertex lies on a level lower than $\ell^*(b)$.}
    \label{fig:p2}
\end{figure}

We prove {\bf property P2}. Suppose, for the sake of a contradiction, that in $\Gamma^*$ no vertex of $K^+_{2,2i+1}$ lies strictly between $b$ and $t$. We call \emph{non-extremal} the vertices of  $K^+_{2,2i+1}$ different from~$b$ and~$t$.  Let $\ell^*$ be the leveling of $K^+_{2,2i+1}$ determined by $\Gamma^*$, and assume, without loss of generality, that $\ell^*(t)>\ell^*(b)$. 

First, by hypothesis, $\ell^*$ assigns each non-extremal vertex $v$ of $K^+_{2,2i+1}$ to a level $\ell^*(v)$ not lower than $\ell^*(t)$ or not higher than $\ell^*(b)$. Second, at most two non-extremal vertices of $K^+_{2,2i+1}$ are assigned to the same level by $\ell^*$, as edges from $b$ and $t$ to three vertices on the same level would define a crossing (see \cref{fig:more-than-two-level}). Third, if two non-extremal vertices of $K^+_{2,2i+1}$ lie on a level not lower than $\ell^*(t)$, then no non-extremal vertex of $K^+_{2,2i+1}$ lies on a level not higher than $\ell^*(b)$, as edges from $b$ and $t$ to such vertices would define a crossing (see \cref{fig:two-above-one-below}). Symmetrically, if two non-extremal vertices of $K^+_{2,2i+1}$ lie on a level not higher than $\ell^*(b)$, then no non-extremal vertex of $K^+_{2,2i+1}$ lies on a level not lower than $\ell^*(t)$. Fourth, since $\ell^*(t)-\ell^*(b) \geq 1$ and since $\Gamma^*$ has span at most~$i$, we have that, for every  non-extremal vertex~$v$ of $K^+_{2,2i+1}$, it holds true that $\ell^*(v)-\ell^*(t) \leq i$ and $\ell^*(b)-\ell^*(v) \leq i$. Hence, at most~$i$ levels not lower than $\ell^*(t)$ and  at most $i$ levels not higher than $\ell^*(b)$ contain non-extremal vertices of $K^+_{2,2i+1}$. 

We are now ready to prove that $\Gamma^*$ contains at most $2i$ non-extremal vertices of $K^+_{2,2i+1}$, which contradicts the fact that there are $2i+1$ such vertices. If every level contains at most one  non-extremal vertex of $K^+_{2,2i+1}$, then the $2i$ upper bound comes from the fact that at most $i$ levels that are not lower than $\ell^*(t)$ and at most $i$ levels that are not higher than $\ell^*(b)$ might contain non-extremal vertices of $K^+_{2,2i+1}$. Also, if a level contains two non-extremal vertices of $K^+_{2,2i+1}$, assume it is not lower than $\ell^*(t)$, as the case in which it is  not higher than $\ell^*(b)$ is analogous. Then at most $i$ levels not lower than $\ell^*(t)$ contain non-extremal vertices of $K^+_{2,2i+1}$ and at most two each. Further, no level not higher than $\ell^*(b)$ might contain non-extremal vertices of $K^+_{2,2i+1}$. This contradiction concludes the proof of property~P2.

Consider a  weakly leveled planar drawing $\Gamma$ of $W_{i,i-1}$ with span at most $i$ and let $\ell$ be the corresponding leveling. 
Since  $\nu$ and $\sigma$ are the extremes of a $K^+_{2,2i+1}$ subgraph of $W_{i,i-1}$, we derive the following. By property P2, we have that in $\Gamma$
there exists a vertex $x$ of $W_{i,i-1}$ that lies strictly between the levels of $\nu$ and $\sigma$. In particular, it holds that $|\ell(\nu)-\ell(x)| \geq 1$, that $|\ell(\sigma)-\ell(x)| \geq 1$, and that $|\ell(\nu)-\ell(\sigma)| = |\ell(\nu)-\ell(x)| + |\ell(\sigma)-\ell(x)|$. This implies that  $|\ell(\nu)-\ell(\sigma)| \geq  |\ell(\nu)-\ell(x)| + 1$. 
Observe now that $\nu$ and $x$ are the poles of a subgraph $W':=W_{i,i-2}$ of $W_{i,i-1}$, and thus they are also the extremes of a $K^+_{2,2i+1}$ subgraph of $W'$. Therefore, again by property P2, we have that there exists a vertex $x'$ of $W'$ that lies strictly between the levels of $\nu$ and $x$. Similarly to before, we get that $|\ell(\nu)-\ell(x)| \geq  |\ell(\nu)-\ell(x')| + 1$, which implies that $|\ell(\nu)-\ell(\sigma)| \geq  |\ell(\nu)-\ell(x')| + 2$. The repetition of this argument and the construction of $W_{i,i-1}$ imply that $|\ell(\nu)-\ell(\sigma)| \geq  |\ell(\nu)-\ell(x^*)| + (i-1)$, where $x^*$ is a vertex of the ``innermost'' copy of $K^+_{2,2i+1}$ incident to $\nu$. The fact that $|\ell(\nu)-\ell(x^*)|\geq 1$ implies the statement.
\end{nestedproof}

The reduction for $s > 1$ is similar to the one for $s=1$, but the role of $K_{2,4}$ is now played by $W_{s,s-1}$.
Starting from a bipartite planar graph $H$, we construct a graph $G$ that is a positive instance of {\sc $s$-Span Weakly Leveled Planarity} if and only if $H$ is a positive instance of {\sc $1$-Span Leveled Planarity}. To this aim, we proceed as follows. We initialize $G=H$. Then, for each edge $(u,v)$ of $H$, we remove $(u,v)$ from $G$, introduce a copy of $W_{s,s-1}$, which we denote by $W_s(u,v)$, and identify $u$ and $v$ with the poles $\nu$ and $\sigma$ of $W_{s,s-1}$, respectively. Clearly, the above reduction can be carried out in polynomial time, since the size of $W_{s,s-1}$ depends only on $s$. 

Suppose that $H$ admits a leveling $\ell_H$, with $\spn(\ell_H) = 1$, such that $(H,\ell_H)$ is level planar. We prove that $G$ admits a leveling $\ell_G$ with $\spn(\ell_G) \leq s$. Let $\Gamma_H$ be a level planar drawing of $(H,\ell_H)$; then the leveling $\ell'_H$ of $H$ where $\ell'_H(v)=s\cdot \ell_H(v)$ is clearly such that $(H,\ell'_H)$ is level planar and $\spn(\ell'_H) = s$. In fact, a level planar drawing $\Gamma'_H$ of $(H,\ell'_H)$ can be obtained by simply setting the $y$-coordinate of each vertex $v$ in $\Gamma'_H$ to be $s$ times the $y$-coordinate of $v$ in $\Gamma_H$. Note that every edge of $H$ has span $s$ in $\Gamma'_H$.

By \cref{claim:drawing-any-s}, for each edge $(u,v)$ of $H$, the graph $W_s(u,v)$ has a leveled planar drawing $\Gamma_{(u,v)}$ with span $s$ in which $u$ lies strictly above all the other vertices of $W_s(u,v)$ and $v$ lies strictly below all the other vertices of $W_s(u,v)$. 
By \cref{cl:uv-long-span}, the span of the edge $(u,v)$ in $\Gamma_{(u,v)}$ is exactly $s$. Thus, we can replace the drawing of each edge $(u,v)$ in $\Gamma'_H$ with a drawing $\Gamma_{(u,v)}$ of $W_s(u,v)$ in which (i) the placement of $u$ and $v$ is the same as in $\Gamma'_H$, and (ii) the vertices of $W_s(u,v)$ in level $i$ lie arbitrarily close to the intersection point in $\Gamma'_H$ of the edge $(u,v)$ and level $i$, and so that these vertices are consecutive along level $i$ and have the same left-to-right ordering in the resulting drawing as in $\Gamma_{(u,v)}$. This yields a leveled planar drawing $\Gamma_G$ of $G$ with span $s$. 



Suppose now that $G$ admits a leveling $\ell_G$ with $\spn(\ell_G) \leq s$ such that $(G,\ell_G)$ is weakly level planar. We show that $H$ admits a proper leveling $\ell_H$ (hence $\spn(\ell_H)= 1$) such that $(H,\ell_H)$ is level planar. Let $\Gamma_G$ be a level planar drawing of $(G,\ell_G)$. By construction, $H$ is a subgraph of $G$. 
Let $\ell'_H$ be the restriction of $\ell_G$ to the vertices of $H$. Recall that each edge $(u,v)$ of $H$ appears in $G$ as an intra-pole edge of the graph $W_s(u,v)$. By \cref{cl:uv-long-span}, we have that  $\spn_{\ell_G}(u,v)=s$, for any edge $(u,v)$ of $H$. Therefore, the drawing $\Gamma'_H$ of $H$ contained in~$\Gamma_G$ is a level planar drawing of $(H,\ell'_H)$ in which each edge has span exactly $s$. It follows that the leveling $\ell_H$ of $H$ such that $\ell_H(v) = \frac{\ell_G(v)}{s}$ is such that $\ell_H$ is proper and $(H,\ell_H)$ is level planar. In fact, a level planar drawing $\Gamma_H$ of $(H,\ell_H)$ can be obtained from $\Gamma'_H$ by simply rescaling the $y$-coordinates so that the $y$-coordinate of any vertex $v$ in $\Gamma_H$ is equal to the $y$-coordinate of $v$ in $\Gamma'_H$ divided by $s$. 
\end{proof}


The proof of \cref{thm:np-hardness} also shows that, for any fixed $s \geq 1$, deciding whether a graph admits a  (non-weakly) leveled planar drawing with span at most $s$ is \NP-complete, which generalizes the \NP-completeness result by Heath and Rosenberg~\cite{DBLP:journals/siamcomp/HeathR92}, \mbox{which is limited to $s=1$.}



\section{Parameterized Complexity} \label{se:fpt}

Motivated by the \NP-hardness of the {\sc $s$-Span Weakly leveled planarity} problem (\cref{thm:np-hardness}), we consider the parameterized complexity of the problem. 
Recall that a problem~$\mathcal P$ whose input is an $n$-vertex graph $G$ is \emph{fixed-parameter tractable} (for short, \emph{FPT}) with respect to some parameter $k$ if it can be solved via an algorithm with running time $O(f(k) \cdot p(n))$, where $f$ is a computable function and $p$ is a polynomial function. A \emph{kernelization} for~$\mathcal P$  is an algorithm that constructs in polynomial time (in~$n$) an instance~$(G',k')$, called \emph{kernel}, such that: (i) the \emph{size} of the kernel, i.e., the number of vertices in $G'$, is some computable function of $k$; (ii) $(G',k')$ and $(G,k)$ are equivalent instances; and (iii) $k'$ is some computable function of $k$. If $\mathcal P$ admits a kernel with respect to some parameter $k$, then it is FPT with respect to $k$. 

Throughout the section, we are going to use the following combinatorial observation.

\begin{lemma} \label{le:degree-3-vc}
Let $X\subseteq V$ be a set of vertices in a planar graph $G=(V,E)$. The number of vertices in $V\setminus X$ that are connected to at least three vertices in $X$ is at most $2|X|$. Further, the number of pairs $(x,y)$ of vertices in $X$ such that $x$ and $y$ are the neighbors of a degree-$2$ vertex in $V-X$ is at most $3|X|$. 
\end{lemma}

\begin{proof}
From~\cite[Lemma 13.3]{flsz-ktpp-19}, we have that the number of vertices in $V\setminus X$ that are connected to at least three vertices in $X$ is at most $\max\{2|X|-4,0\}\leq 2|X|$. We prove the second statement. Consider the subgraph $H$ of $G$ induced by the vertices in $X$ and by the degree-$2$ vertices in $V-X$ that are neighbors of two vertices in $X$. For each pair $(x,y)$ of vertices in $X$ such that $x$ and $y$ are the neighbors of a degree-$2$ vertex $v_{xy}$ in $V-X$, replace in $H$ the path $(x,v_{xy},y)$ with an edge $(x,y)$ and remove the vertex $v_{xy}$. Further modify $H$ by replacing multiple edges connecting two vertices $(x,y)$ with a single edge $(x,y)$. Now $H$ is a planar graph with $X$ as vertex set and with an edge between two vertices if, in $G$, they are the neighbors of a degree-$2$ vertex in $V-X$. The second statement follows.
\end{proof}

\subsection{Parameterization by Vertex Cover}\label{subse:fpt-vc}

In this section, we show that {\sc $s$-Span Weakly leveled planarity} has a kernel whose size is polynomial in the size of a vertex cover. Hence, the problem is FPT with respect to this size.
Consider an instance $(G,s)$ of {\sc $s$-Span Weakly leveled planarity}, where $G=(V,E)$ is a planar graph and $s\geq 1$ is an integer; recall that the problem asks whether $G$ admits a weakly leveled planar drawing whose span is at most $s$. Let $C$ be a vertex cover, i.e., a set of vertices such that every edge has at least one end-vertex in $C$ and let $k:=|C|$. We start by showing that, if $s$ is sufficiently large (namely, larger than or equal to $6k$), then $(G,s)$ is a positive instance, hence we can assume that $s$ is bounded with respect to $k$.

\begin{lemma} \label{le:drawing-vc}
Suppose that $s\geq 6k$. Then $(G,s)$ is a positive instance of {\sc $s$-Span Weakly leveled planarity}.
\end{lemma}

\begin{proof}
We construct a \emph{trimmed graph} $\trim(G)$, by removing all the degree-$1$ vertices of $G$ in~$V \setminus C$ and ``smoothing'' all degree-$2$ vertices  of $G$ in $V \setminus C$ (that is, removing each degree-$2$ vertex and connecting its neighbors with an edge). Note that $\trim(G)$ has at most $3k$ vertices. Indeed, by \Cref{le:degree-3-vc}, we have that $G$ (and hence $\trim(G)$) has at most $2k$ vertices in $V\setminus C$ whose degree is at least three. Further, by construction, $\trim(G)$ has no degree-$1$ and degree-$2$ vertex in $V\setminus C$. Finally, $\trim(G)$ has at most $k$ vertices in~$C$. An arbitrary leveled planar drawing $\Gamma_t$ of $\trim(G)$ without empty levels has height, and hence span, at most $3k$.
A leveled planar drawing $\Gamma_G$ of $G$ can then be constructed from $\Gamma_t$ by inserting a new level between any two consecutive levels of $\Gamma_t$. This at most doubles the span of the edges which is now at most $6k$. The new levels can be used for reinserting the removed degree-$1$ vertices of $G$ (on a level next to the one of their neighbor) and the smoothed degree-$2$ vertices (at the intersection between the edge they have to lie on and the new level cutting that edge). 
\end{proof}

We now proceed by showing how to obtain a kernel for the given instance $(G,s)$ of {\sc $s$-Span Weakly leveled planarity}. This is done by applying the following two reduction rules.

\begin{reductionrule} \label{rule:vc-deg1}
For every vertex $c \in C$, let $V_c$ be the set of degree-$1$ neighbors of $c$ in $V\setminus C$. If $|V_c|>3$, remove $|V_c|-3$ vertices in $V_c$ from $G$, as well as their incident edges.  
\end{reductionrule}

\begin{reductionrule} \label{rule:vc-deg2}
For every pair of vertices $\{c,d\}\in C$, let $V_{cd}$ be the set of degree-$2$ vertices in $V\setminus C$ with
  neighborhood~$\{c,d\}$. If~$|V_{cd}| > 4s+5$, remove $|V_{cd}|-4s-5$ vertices in $V_{cd}$ from $G$, as well as their incident edges.  
\end{reductionrule}

\begin{lemma}
  \label{lem:vc-rules-safe}
The instance obtained by applying \Cref{rule:vc-deg1,rule:vc-deg2} is equivalent to $(G,s)$.
\end{lemma}

\begin{proof}
Let $G'$ be the instance obtained by applying \Cref{rule:vc-deg1,rule:vc-deg2} to $G$. Clearly, if $(G,s)$ is a positive instance of {\sc $s$-Span Weakly leveled planarity}, then $(G',s)$ is a positive instance too, as $G'$ is a subgraph of $G$. In order to show that, if $(G',s)$ is a positive instance, then $(G,s)$ is a positive instance too, we need to prove that the vertices (and their incident edges) removed by applying \Cref{rule:vc-deg1,rule:vc-deg2} can be inserted in a weakly leveled planar drawing~$\Gamma'$ of~$G'$ with $\spn(\Gamma') \le s$ so that in the resulting drawing $\Gamma$ the span remains bounded by $s$. Throughout the insertion process, we always call $\Gamma'$ the drawing into which we are inserting vertices, although several re-insertions might have been performed already. 

First, consider each vertex $c \in C$ such that $|V_c|>3$. According to \Cref{rule:vc-deg1}, we have that $|V_c|-3$ vertices in $V_c$ were removed from $G$ in order to obtain $G'$. Since at least three vertices from $V_c$ belong to $G'$ and since at most two vertices from $V_c$ lie on the same level as $c$, one to its left and one to its right,  a degree-$1$ neighbor $v_c$ of $c$ lies on a level different from the one of~$c$ in $\Gamma'$. All the removed degree-$1$ neighbors of $c$ can then be placed next to $v_c$.

Second, consider each pair of vertices $\{c,d\}\in C$ such that $|V_{cd}| > 4s+5$. According to \Cref{rule:vc-deg2}, we have that $|V_{cd}|-4s-5$ vertices from $V_{cd}$ were removed from $G$ in order to obtain $G'$. We first prove that, in $\Gamma'$, a degree-$2$ neighbor of $c$ and $d$ lies on a level strictly between the level $\ell(c)$ of $c$ and the level $\ell(d)$ of $d$. Note that, since $\Gamma'$ is weakly leveled planar, each level not lower than $\ell(c)$ and not lower than $\ell(d)$ contains at most two vertices in $V_{cd}$; furthermore, there are at most $s+1$ such levels that are at distance at most $s$ from both $\ell(c)$ and $\ell(d)$. Hence, the number of degree-$2$ neighbors of $c$ and $d$ that lie on a level not lower than $\ell(c)$ and not lower than $\ell(d)$ is at most $2s+2$. Likewise, the number of degree-$2$ neighbors of $c$ and $d$ that lie on a level not higher than $\ell(c)$ and not higher than $\ell(d)$ is at most $2s+2$. Since the number of degree-$2$ neighbors of $c$ and $d$ in the reduced graph is $4s+5$, by the pigeonhole principle, a degree-$2$ neighbor $v_{cd}$ of $c$ and $d$ lies on a level strictly between $\ell(c)$ and $\ell(d)$. All the removed degree-$2$ neighbors of $c$ and $d$ can then be placed next to $v_{cd}$, obtaining a weakly leveled planar drawing of $G$ with span $s$.
\end{proof}

We now prove the existence of a kernel whose size is a polynomial function of $k$ and $s$. 
\begin{theorem} \label{thm:kernel-cover-span}
  Let $(G,s)$ be an instance of the {\sc $s$-Span Weakly leveled planarity} problem with a vertex cover~$C$ of size~$k$. Then there exists a kernelization that applied to $(G,s)$ constructs a kernel of size~$O(k \cdot s)$.
\end{theorem}

\begin{proof}
Let~$(G',s)$ with~$G'=(V'\subseteq V,E'\subseteq E)$ be the instance obtained by applying \Cref{rule:vc-deg1,rule:vc-deg2}. By \Cref{lem:vc-rules-safe}, we have that~$(G,s)$ and~$(G',s)$ are equivalent. To bound the size of~$G'$, we count the number of vertices in~$V' \setminus C$. By \Cref{le:degree-3-vc}, the number of vertices with degree at least~three is in $O(k)$. The use of \Cref{rule:vc-deg1} and \Cref{rule:vc-deg2} allows us to bound the number of degree-$1$ vertices in $V'\setminus C$ to $O(k)$, namely $3$ for each vertex in $C$, and the number of degree-$2$ vertices in $V'\setminus C$ to $O(k \cdot s)$, namely $4s+5$ for each pair of vertices in $C$ that are connected to a degree-$2$ vertex in $V'\setminus C$; the number of such pairs is in $O(k)$ by \Cref{le:degree-3-vc}. Finally, the number of vertices in $C$ is in $O(k)$.
\end{proof}

Finally, we get rid of the dependence on $s$ of the kernel size.

\begin{theorem} \label{cor:fpt-vc}
  Let $(G,s)$ be an instance of {\sc $s$-Span Weakly leveled planarity} with a vertex cover~$C$ of size~$k$. There exists a kernelization that applied to $(G,s)$ constructs a kernel of size~$O(k^2)$. Hence, the problem is FPT with respect to the size of a vertex cover.
\end{theorem}

\begin{proof}
The kernel $(G',s)$ is obtained as in \Cref{thm:kernel-cover-span}, by applying \Cref{rule:vc-deg1,rule:vc-deg2}, however, for each pair of vertices $\{c,d\}\in C$, the number of degree-$2$ vertices of $G$ that are neighbors of $c$ and $d$ and that are kept in the kernel is at most $\min\{4s+5,24k+5\}$. Since this number is in $O(k)$ and since the number of pairs of vertices $\{c,d\}\in C$ such that $c$ and $d$ are the neighbors of a degree-$2$ vertex in $V-C$ is also in $O(k)$ by \Cref{le:degree-3-vc}, the bound on the kernel size follows from the one of \Cref{thm:kernel-cover-span}. Since the kernel can be constructed in polynomial time, it follows that {\sc $s$-Span Weakly leveled planarity} is FPT with respect to $k$.

It remains to prove that $(G',s)$ is equivalent to $(G,s)$. The proof distinguishes two cases. If $s\leq 6k$, then the proof follows from \Cref{thm:kernel-cover-span}, as in this case $(G',s)$ is the same kernel as the one computed for that theorem. On the other hand, if $s> 6k$, by \cref{le:drawing-vc} we have that $(G,s)$ is a positive instance, and hence $(G',s)$ is a positive instance, as well. 
\end{proof}

\subsection{Modulator to Bounded Size Components} \label{subse:fpt-mbsc}

In this section, we extend the results of \Cref{subse:fpt-vc} by providing a parameterization for {\sc $s$-Span Weakly leveled planarity} that is stronger than the one by vertex cover. This comes at the expense of an increase in the size of the kernel (namely, this size is now going to be exponential) with respect to the one of \Cref{cor:fpt-vc}. 

Let $(G,s)$ be an instance of {\sc $s$-Span Weakly leveled planarity} and~$b$ be a parameter. A set~$M$ of vertices of $G$ is a \emph{modulator to components of size $b$} (\emph{$b$-modulator} for short) if every connected component of $G-M$ has size at most~$b$.  Note that a $1$-modulator is a vertex cover.  We show that testing whether a graph with a $b$-modulator of size $k$ admits a weakly leveled planar drawing with span $s$ is FPT with respect to $b+k$.

For a component $C$ of $G-M$, its \emph{attachments}, denoted by $\att(C)$, are the vertices of $M$ adjacent to vertices of $C$.  For a component $C$, we denote by~$\bridge(C)$ the subgraph of $G$ formed by $C \cup \att(C)$ together with all edges between vertices of $C$ and vertices of $\att(C)$. Note that edges between vertices of $\att(C)$ are not part of $\bridge(C)$.

As for vertex cover, the first ingredient consists of showing that, if the span is sufficiently large (namely, larger than $(5b+1)bk$), then $(G,s)$ is a positive instance. This allows us to assume that $s$ is bounded with respect to  $b+k$.

We start with three lemmata about the drawings of the connected components of $G-M$.


\begin{lemma}\label{le:uv-incident-same-face}
Suppose that $G$ is connected. Let~$\{u,v\} \subseteq M$ be such that there are two distinct connected components~$C,C'$ of $G-M$ with~$\att(C) = \att(C') = \{u,v\}$.  Then in every planar drawing of $G$ we have that $u$ and $v$ are incident to a common face. 
\end{lemma}
\begin{proof}
    Assume, for a contradiction, that there is a planar drawing $\Gamma$ of $G$ in which $u$ and $v$ are not incident to a common face. Then there exists a cycle $\mathcal O$ in $G$ that has $u$ inside and $v$ outside in $\Gamma$, without loss of generality up to renaming $u$ with $v$. If $\mathcal O$ belongs to $G-C$, then $C$ crosses $\mathcal O$ in $\Gamma$, a contradiction. If $\mathcal O$ belongs to $C$, then $C'$ crosses $\mathcal O$ in $\Gamma$, also a contradiction.
\end{proof}


\begin{lemma}
    \label{lem:1-att-drawing}
    Let~$X$ be an $n_x$-vertex planar graph and $x$ be a vertex of $X$.  Then $X$ has a leveled planar drawing on~$n_x$ levels where $x$ is the unique vertex on the lowest (or highest, as desired) level.
\end{lemma}

\begin{proof}
We show how to construct the required drawing with $x$ on the lowest level, the other case is symmetric. Augment $X$ with extra edges so that it becomes a biconnected planar graph $X'$. Compute then an \emph{$st$-numbering}~\cite{lec-aptg-67} of $X'$ with $s:=x$; this is a bijective mapping from $V(X')$ to the set $\{1,\dots,n_x\}$ such that a special vertex $s$ (for us, this is $x$) gets number~$1$, a special vertex $t$ (for us, this is any vertex sharing a face with $x$ in a planar embedding of $X'$) gets number $n_x$, and every other vertex has both a neighbor with lower number and a neighbor with higher number. Then $X'$ admits a leveled planar drawing such that the vertex with number $i$ lies on level $i$~\cite[Theorem 3.5]{DBLP:journals/tcs/BattistaT88}. Restricting such a drawing to the edges of $X$ provides the desired drawing.
\end{proof}

\begin{lemma}
    \label{lem:2-att-drawing}
    Let~$X$ be an $n_x$-vertex planar graph, and let $x$ and $y$ be two vertices of $X$ such that $X$ admits a planar drawing in which $x$ and $y$ are incident to a common face. Then $X$ admits a leveled planar drawing on $n_x$ levels where $x$ is the unique vertex on the lowest (or highest, as desired) level and $y$ is the unique vertex on the highest (resp.\ lowest) level.
\end{lemma}

\begin{proof}
The proof is similar to the one of \Cref{lem:1-att-drawing}. We again only show the construction to make $x$ the lowest vertex. If the edge $(x,y)$ is not in $X$, then add it to $X$. Because of the assumption that $X$ admits a planar drawing in which $x$ and $y$ are incident to a common face, this preserves the planarity of $X$. Further augment $X$ with extra edges so that it becomes a biconnected planar graph $X'$. Compute then an $st$-numbering of $X'$ with $s:=x$ and $t:=y$. Then $X'$ admits a leveled planar drawing in which the vertex with number $i$ lies on level $i$~\cite[Theorem 3.5]{DBLP:journals/tcs/BattistaT88}. Restricting such a drawing to the edges of $X$ provides the desired drawing.
\end{proof}

We now proceed to construct a \emph{trimmed graph} $\trim(G)$. This is done by removing all components~$C$ of $G-M$ with~$|\att(C)| = 1$.  Further, for each pair~$\{u,v\} \subseteq M$ such that there are at least two components $C$,$C'$ with~$\att(C') = \att(C) = \{u,v\}$, we remove all such components and replace them by a single edge~$uv$.  We have the following.

\begin{lemma} \label{le:bm-trim-size}
We have that $\trim(G)$ is planar and has size at most~$(5b+1)k$.
\end{lemma}
\begin{proof}
We first prove the planarity of $\trim(G)$. Actually, this follows from the planarity of~$G$ almost directly, with one caveat. For each pair~$\{u,v\} \subseteq M$ such that there are at least two components $C$,$C'$ with~$\att(C') = \att(C) = \{u,v\}$, \Cref{le:uv-incident-same-face} ensures that the insertion of the edge $(u,v)$ preserves the planarity of $\trim(G)$. Note that if a pair~$\{u,v\} \subseteq M$ is such that there is a unique component $C$ with~$\att(C) = \{u,v\}$, then a statement analogous to the one of \Cref{le:uv-incident-same-face} does not hold, and that is why we cannot replace $C$ with the edge $(u,v)$.

We now prove the bound on the size of $\trim(G)$. By \Cref{le:degree-3-vc}, we have that $\trim(G)$ contains at most $2k$ components with three or more attachments. Each of such components has at most $b$ vertices, for a total of at most $2bk$ vertices. Also, again by \Cref{le:degree-3-vc}, we have that $\trim(G)$ contains at most $3k$ pairs $\{u,v\} \subseteq M$ of vertices such that there exists a component $C$ with $\att(C) = \{u,v\}$ and, by construction, for each such a pair $\{u,v\}$, we have that $\trim(G)$ contains one component $C$ with $\att(C) = \{u,v\}$. This results in $3bk$ additional vertices in such components. Finally, $\trim(G)$ has at most $k$ vertices in $M$. We are now ready to prove the following.
\end{proof}

%

\begin{lemma} \label{lem:bm-span-bound}
Suppose that $s\geq (5b+1)bk$. Then $(G,s)$ is a positive instance of {\sc $s$-Span Weakly leveled planarity}.
\end{lemma}

\begin{proof}
    Let~$\Gamma_t$ be an arbitrary leveled planar drawing of~$\trim(G)$ without empty levels.  Since the size of $\trim(G)$ is bounded by $(5b+1)k$, it follows that the height and therefore also the span of~$\Gamma$ is at most $(5b+1)k$.  
%
%
    To obtain a leveled planar drawing~$\Gamma_G$ of $G$, we insert $b$ intermediate levels between any pair of consecutive levels of $\Gamma_t$, resulting in a total of $(5b+1)bk$ levels.  We use these levels to re-insert the components of $G$ that we removed to obtain~$\trim(G)$. This is done as follows.  For each component~$C$ of $G-M$ with a single attachment~$u$, by \Cref{lem:1-att-drawing} we obtain a drawing of~$\bridge(C)$ on $b+1$ levels in which~$u$ is the unique vertex on the highest level.  We can thus use the level of $u$ and the $b$ intermediate levels below it to merge such a drawing into~$\Gamma_G$ (by identifying the two copies of $u$).  Similarly, for a component $C$ of $G-M$ with $\att(C)=\{u,v\}$, if $C$ is not contained in~$\trim(G)$, then $\trim(G)$ contains the edge~$(u,v)$.   By \Cref{lem:2-att-drawing}, $\bridge(C)$ has a drawing on $b+2$ levels, where $u$ and $v$ are on the highest and lowest level, respectively.  As $u$ and~$v$ are drawn on distinct levels in~$\Gamma_t$, there are at least $b$ intermediate levels between them in $\Gamma_G$, and we can merge the drawing of~$\bridge(C)$ next to $e$ (by merging the two copies of $u$ and the two copies of $v$, respectively).  Since these are the only components of $G-M$ that are not contained in $\trim(G)$, this yields a drawing~$\Gamma_G$ of~$G$ on $(5b+1)bk$ levels.  Therefore~$(G,s)$ is a positive instance of {\sc $s$-Span Weakly leveled planarity}.
\end{proof}

We now proceed to show how to obtain a kernel for the given instance $(G,s)$ of {\sc $s$-Span Weakly leveled planarity}. We define two components of $G-M$ with~$|\att(C)| \le 2$ as \emph{equivalent} if~$\att(C) = \att(C')$ and moreover~$\bridge(C)$ and~$\bridge(C')$ are isomorphic with an isomorphism that keeps the attachments fixed. The following lemma follows from the fact that the number of planar graphs on $b$ vertices is an exponential function of $b$. 

\begin{lemma} \label{le:computable}
    There is a computable function $f \colon \mathbb N \to \mathbb N$ such that, for any planar graph $G$ and any $b$-modulator of size~$k$, there are at most~$f(b)$ equivalence classes of components of~$G-M$ with at most two specified attachments.
\end{lemma}

In order to prove the correctness of the upcoming kernelization, we are going to need two topological lemmata about leveled planar drawings.

\begin{lemma}
\label{lem:block-above-below}
Let $(G,\ell)$ be a leveled planar graph with a leveled planar drawing~$\Gamma$ and let $v$ be a cutvertex of $G$ that is part of $4t$ blocks.  Then at least one of the following two statements is true in $\Gamma$:
\begin{itemize}
    \item A block incident to $v$ is entirely above or entirely below~$v$, except at $v$ itself.
    \item A block incident to $v$ has height larger than or equal to~$t$.
\end{itemize}
\end{lemma}


\begin{proof}
    Assume that there is no block incident to $v$ whose vertices different from $v$ all lie above or all lie below $v$.  Then each block intersects the level~$\ell(v)$ at a point different from $v$, where an intersection is either a crossing between the level and an edge of the block or is a vertex of the block that is different from $v$ and that lies on $\ell(v)$. For each block $C_i$, let $\pi_i$ be a curve that is part of the drawing of $C_i$ in $\Gamma$, that connects $v$ with a point $p_i$ different from $v$ on~$\ell(v)$, and that does not contain any point on~$\ell(v)$ in its interior. Assume that at least $2t$ of the points $p_i$ are to the right of~$v$, as the case in which $2t$ of the points $p_i$ are to the left of $v$ is symmetric. Let~$C_1,\dots,C_{2t}$ be $2t$ blocks such that the points $p_1,\dots,p_{2t}$ are to the right of~$v$, where $p_1,\dots,p_{2t}$ are in this left-to-right order along $\ell(v)$. For~$i=1,\dots,2t$, since~$p_{i+1}$ lies right of~$p_i$ and the curves~$\pi_1,\dots,\pi_{2t}$ are pairwise non-crossing, it follows that~$\pi_{i+1}$ intersects a level that is strictly above or strictly below all levels that are intersected by~$\pi_1,\dots,\pi_i$.  Therefore~one of the blocks contains a vertex that is either $t$ levels below or $t$ levels above~$v$, i.e., its height is at least~$t$.
\end{proof}

\begin{lemma}
    \label{lem:split-component-between}
    Let $(G,\ell)$ be a leveled planar graph with a leveled planar drawing~$\Gamma$ and let $\{u,v\}$ be a separating pair $G$ with~$\ell(u) \le \ell(v)$ that is part of $8t$ split components.  Then at least one of the following statements is true in $\Gamma$:
  \begin{itemize}
    \item A split component incident to $u$ and $v$ is strictly between~$\ell(u)$ and~$\ell(v)$, except at $u$ and $v$.
    \item A split component incident to $u$ and $v$ has height larger than or equal to~$t$.
\end{itemize}  
\end{lemma}

\begin{proof}
Assume that there is no split component whose vertices different from $u$ and $v$ all lie strictly between~$\ell(u)$ and~$\ell(v)$.  Then each split component intersects the level~$\ell(u)$ at a point different from $u$ or intersects the level~$\ell(v)$ at a point different from $v$. Assume that there are at least $4t$ split components that intersect~$\ell(v)$ at a point different from $v$, as the other case is symmetric. The rest of the proof is the same as in \Cref{lem:block-above-below}. For each split component $C_i$ that intersect~$\ell(v)$ at a point different from $v$, let $\pi_i$ be a curve that is part of the drawing of $C_i$ in $\Gamma$, that connects $v$ with a point $p_i$ different from $v$ on~$\ell(v)$, and that does not contain any point on~$\ell(v)$ in its interior. Assume that at least $2t$ of the points $p_i$ are to the right of $v$, as the other case is symmetric. Let~$C_1,\dots,C_{2t}$ be $2t$ split components such that the points $p_1,\dots,p_{2t}$ are to the right of $v$, where $p_1,\dots,p_{2t}$ are in this left-to-right order along $\ell(v)$. For~$i=1,\dots,2t-1$, since~$p_{i+1}$ lies right of~$p_i$ and the curves~$\pi_1,\dots,\pi_{2t}$ are pairwise non-crossing, it follows that~$\pi_{i+1}$ intersects a level that is strictly above or strictly below all levels that are intersected by~$\pi_1,\dots,\pi_i$.  Therefore~one of the split components contains contain a vertex that is either $t$ levels below or $t$ levels above~$v$, i.e., its height is at least~$t$.
\end{proof}

We are now ready to show the two reduction rules that allow us to get a kernel from the given instance $(G,s)$ with a $b$-modulator $M$ of size $k$.

\begin{reductionrule} \label{rule:bm-deg1}
For every vertex $v \in M$, consider every maximal set~$\mathcal C_{v}$ of equivalent components~$C$ of $G-M$ such that $\att(C) = \{v\}$.  If $|\mathcal C_v| > (4s+4)b$, then remove all but $(4s+4)b$ of these components from $G$.
\end{reductionrule}

\begin{reductionrule} \label{rule:bm-deg2}
    For every pair of vertices $\{u,v\}\in M$, consider every maximal set~$\mathcal C_{uv}$ of equivalent components~$C$ of $G-M$ with~$\att(C) = \{u,v\}$.  If~$|\mathcal C_{uv}| > (8s+8)(b+1)$, then remove all but $(8s+8)(b+1)$ such components from $G$.
\end{reductionrule}

\begin{lemma}
  \label{lem:bm-rules-safe}
The instance obtained by applying \Cref{rule:bm-deg1,rule:bm-deg2} is equivalent to $(G,s)$.
\end{lemma}

\begin{proof}
Let $G'$ be the instance obtained by applying \Cref{rule:bm-deg1,rule:bm-deg2} to $G$. Clearly, if $(G,s)$ is a positive instance of {\sc $s$-Span Weakly leveled planarity}, then $(G',s)$ is a positive instance too, as $G'$ is a subgraph of $G$. In order to show that, if $(G',s)$ is a positive instance, then $(G,s)$ is a positive instance too, we need to prove that the components removed by applying \Cref{rule:bm-deg1,rule:bm-deg2} can be inserted in a weakly leveled planar drawing $\Gamma'$ of $G'$ with $\spn(\Gamma') \le s$ so that in the resulting drawing $\Gamma$ the span remains bounded by $s$. We reinsert one component $C$ at a time; we always call $\Gamma'$ the drawing into which we are inserting $C$, although several components previously removed by \Cref{rule:bm-deg1,rule:bm-deg2} might have been reinserted already. We distinguish two cases based on whether~$|\att(C)| = 1$ or~$|\att(C)|=2$. 


We start with the case~$|\att(C)|=1$.  Let~$\att(C) = \{v\}$.  As $C$ was removed by \Cref{rule:bm-deg1}, we have that $G'-M$ contains $4(s+1)b$ components equivalent to~$C$. If one of them, say~$C'$, is such that $\bridge(C')$ were drawn with height at least $(s+1)b$, then $\bridge(C')$ would contain an edge whose span is at least~$s+1$, given that $\bridge(C')$ is connected and has at most $b+1$ vertices. However, this would contradict $\spn(\Gamma') \le s$. It hence follows from \Cref{lem:block-above-below}, applied with $t=(s+1)b$, that~$G'-M$ contains a component~$C'$ equivalent to~$C$ that is entirely drawn above (or below)~$v$ in~$\Gamma'$; assume the former, as the latter case is symmetric. 

\begin{figure}
    \centering
    \begin{subfigure}[t]{0.48\textwidth}
    \centering
        \includegraphics[page=1]{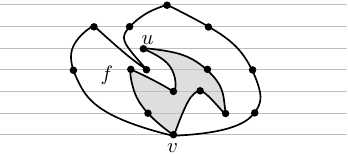}
        \subcaption{}
        \label{fig:ModularoModification-a}
    \end{subfigure}
    \begin{subfigure}[t]{0.48\textwidth}
    \centering
        \includegraphics[page=2]{Figures/ModulatorModification.pdf}
        \subcaption{}
        \label{fig:ModularoModification-b}
    \end{subfigure}
    \caption{(a) The component $C'$ (shaded gray) in $\Gamma'$. Other than $C'$, only the boundary of the face $f$ is shown. (b) Drawing a new copy of $C'$ next to $C'$.}
    \label{fig:ModularoModification}
\end{figure}

Consider the leftmost vertex $u$ of $C'$ on the highest level intersected by $C'$ and let $f$ be the face of $\Gamma'$ that is immediately to the left of $u$; see~\Cref{fig:ModularoModification-a}. We insert another copy of $C'$ in $f$, with the two copies identified at $v$. It might be the case that new copy of $C'$ cannot be drawn monotonically inside $f$ in~$\Gamma'$, however, since all the vertices in the path that is encountered when walking in clockwise direction along the outer face of $C'$ from $u$ to $v$ lie on levels strictly between those of $u$ and $v$, the graph, together with the new copy of $C'$, can be planarly redrawn so that each vertex remains on the same level and each edge is $y$-monotone; see~\Cref{fig:ModularoModification-b}. Since~$C$ is equivalent to~$C'$, we can interpret the new copy of $C'$ as a drawing of~$C$. This completes the reinsertion of~$C$ in $\Gamma'$.




The case~$|\att(C)|=2$ works similarly.  Let~$\att(C) = \{u,v\}$.  Without loss of generality,  assume~$\ell(u) \le \ell(v)$.  As~$C$ was removed by \Cref{rule:bm-deg2}, we have that $G'-M$ contains $8(s+1)(b+1)$ components equivalent to~$C$. If one of such components, say~$C'$, is such that $\bridge(C')$ were drawn with height at least $(s+1)(b+1)$ then, since $\bridge(C')$ is connected and has at most $b+2$ vertices, $\bridge(C')$ would contain an edge whose span is at least~$s+1$, contradicting $\spn(\Gamma') \le s$. It hence follows from \Cref{lem:split-component-between}, applied with $t=(s+1)(b+1)$, that~$G'-M$ contains a component~$C'$ equivalent to~$C$ that is drawn strictly between~$\ell(u)$ and~$\ell(v)$ in~$\Gamma'$. Again, we can insert a new copy of $C'$ immediately to the left of $C'$, where the two copies are identified at $u$ and $v$, and redraw the entire graph with straight-line segments. By interpreting the new copy of $C'$ as $C$, we obtain a drawing of $G'\cup C$ with span at most $s$.
\end{proof}


We are now ready to prove the following.

 \begin{theorem} \label{thm:kernel-modulator-span}
  Let $(G,s)$ be an instance of the {\sc $s$-Span Weakly leveled planarity} problem and a $b$-modulator of size~$k$. Then there exists a kernelization that applied to $(G,s)$ constructs a kernel of size~$O(f(b)\cdot k \cdot b^2 \cdot s)$.
\end{theorem}

\begin{proof}
First, we prove that testing whether two components~$C,C'$ of~$G-M$ with~$\att(C)=\att(C')$ and $|\att(C)| \le 2$ are equivalent can be done in polynomial time. This is done as follows. Recall that each of $C$ and $C'$ has at most $b$ vertices. Assume that $|\att(C)| =2$, as the case $|\att(C)| =1$ is easier to discuss. Let $\att(C)=\att(C')=\{u,v\}$. Augment each of $C$ and $C'$ with a path of length $b+1$ attached to $u$ and with a path of length $b+2$ attached to~$v$. Let $C_{+}$ and $C'_{+}$ be the resulting graphs, respectively. Now any isomorphism between $C_{+}$ and $C'_{+}$ is an isomorphism between $C$ and $C'$ that maps to one another the two copies of $u$ and maps to one another the two copies of $v$. Thus, testing the equivalence of $C$ and $C'$ amounts to testing the isomorphism of $C_{+}$ and $C'_{+}$, which can be done in $O(b)$ time~\cite{DBLP:conf/stoc/HopcroftW74}.  

Since there are $O(n)$ components in~$G-M$ that have to be pairwise tested for equivalence, it follows that the kernelization of~\Cref{rule:bm-deg1,rule:bm-deg2} can be performed in polynomial time, resulting in an instance~$(G',s)$ which, by~\Cref{lem:bm-rules-safe}, is equivalent to~$(G,s)$. 

It remains to bound the size of~$G'$.  By \Cref{le:degree-3-vc}, there are at most~$2k$ components of~$G'-M$ with three or more attachments, which together contribute with~$O(k\cdot b)$ vertices to the size of $G'$. Further, the number of vertices in $M$ is in $O(k)$. We thus only need to bound the number of vertices in the components of~$G'-M$ that have one or two attachments.  
    
For a fixed vertex~$u \in M$, there are at most~$f(b)$ equivalence classes of components of $G'-M$ with sole attachment~$u$. Due to \Cref{rule:bm-deg1}, each such equivalence class contains~$O(b\cdot s)$ components. Since each component has at most $b$ vertices, the components whose attachment is a fixed vertex~$u \in M$ contribute with~$O(f(b)\cdot b^2 \cdot s)$ vertices to $G'$, and thus all components with a single attachment contribute with~$O(f(b)\cdot k\cdot b^2 \cdot s)$ vertices in total.

Similarly, for a fixed pair of vertices~$\{u,v\}\in M$, there are at most~$f(b)$ equivalence classes of components of $G'-M$ with attachments~$\{u,v\}$. Due to \Cref{rule:bm-deg2}, each such equivalence class contains~$O(b\cdot s)$ components. Since each component has at most $b$ vertices, the components whose attachments are a fixed pair of vertices~$\{u,v\}\in M$ contribute with~$O(f(b)\cdot b^2 \cdot s)$ vertices to $G'$, and thus all components with two attachments contribute with~$O(f(b)\cdot k\cdot b^2 \cdot s)$ vertices in total, as there are $O(k)$ such components, by \Cref{le:degree-3-vc}.
\end{proof}


By \Cref{lem:bm-span-bound}, we obtain the following.


 \begin{theorem}\label{cor:kernel-modulator-span}
  Let $(G,s)$ be an instance of {\sc $s$-Span Weakly leveled planarity} with a $b$-modulator of size~$k$. There exists a kernelization that applied to $(G,s)$ constructs a kernel of size~$O(f(b)\cdot k^2 \cdot b^4)$. Hence, the problem is FPT with respect to $k+b$.
\end{theorem}


\begin{proof}
The kernel $(G',s)$ is obtained as in \Cref{thm:kernel-modulator-span}, by applying \Cref{rule:bm-deg1,rule:bm-deg2}, however:
\begin{itemize}
\item for each vertex $u\in M$, the number of components with attachment $\{u\}$ that belong to each equivalence class and that are kept in the kernel is $\min\{(4s+4)b,(4((5b+1)bk)+4)b\}$; and
\item for each pair of vertices $\{u,v\}\in M$, the number of components with attachments $\{u,v\}$ that belong to each equivalence class and that are kept in the kernel is $\min\{(8s+8)(b+1),(8((5b+1)bk)+8)(b+1)\}$. 
\end{itemize}

Since the number of equivalence classes is at most $f(b)$, the number of vertices in each component is at most $b$, and the number of pairs of vertices $\{u,v\}\in M$ that are attachment to some component of $G-M$ is in $O(k)$ by \Cref{le:degree-3-vc}, it follows that the size of the kernel is $O(f(b)\cdot k^2 \cdot b^4)$. Since the kernel can be constructed in polynomial time, as proved in \Cref{thm:kernel-modulator-span}, it follows that {\sc $s$-Span Weakly leveled planarity} is FPT with respect to $k$.

It remains to prove that $(G',s)$ is equivalent to $(G,s)$. The proof distinguishes two cases. If $s\leq (5b+1)\cdot b \cdot k$, then the proof follows from \Cref{thm:kernel-cover-span}, as in this case $(G',s)$ is the same kernel as the one computed for that theorem. On the other hand, if $s> (5b+1)\cdot b \cdot k$, by \Cref{lem:bm-span-bound}, we have that $(G,s)$ is a positive instance, and hence $(G',s)$ is a positive instance, as well. 
\end{proof}

\subsection{Treedepth}\label{sse:treedepth}

We now move to treedepth. A \emph{treedepth decomposition} of a graph $G=(V,E)$ is a tree~$T$ on vertex set~$V$ with the property that every edge of $G$ connects a pair of vertices that have an ancestor-descendant relationship in $T$. The \emph{treedepth} of $G$ is the minimum depth (i.e., maximum number of vertices in any root-to-leaf path) of a treedepth decomposition $T$ of $G$.

Let $td$ be the treedepth of $G$ and $T$ be a treedepth decomposition of $G$ with depth $td$. Let~$r$ be the root of $T$. For a vertex~$u \in V$, we denote by~$T_u$ the subtree of $T$ rooted at~$u$, by~$V_u$ the vertex set of $T_u$, by $d(u)$ the depth of $u$ (where $d(u)=1$ if $u$ is a leaf and $d(u)=td$ if $u=r$), and by~$R(u)$ the set of vertices on the path from~$u$ to~$r$ (end-vertices included). 

As in~\Cref{le:drawing-vc,lem:bm-span-bound}, we can show that the instance $(G,s)$ is positive if $s$ is sufficiently large, namely larger than or equal to $((5td)^{td}+1)^{td}$. Hence, we can assume that $s$ is bounded with respect to $td$. This is proved in the following.


\begin{lemma} \label{le:drawing-treedepth}
Suppose that $s\geq ((5td)^{td}+1)^{td}$. Then $(G,s)$ is a positive instance of {\sc $s$-Span Weakly leveled planarity}.
\end{lemma}



\begin{proof}
We first define a trimmed graph from $G$ (and a corresponding treedepth decomposition from $T$). The trimming operation we define in order to get the trimmed graph is similar to the one defined in \Cref{subse:fpt-mbsc}. However, it is applied $td+1$ times, namely for increasing values of $i$ from $1$ to $td$, it is applied to all the vertices of $G$ that have depth $i$ in~$T$. The main feature of the trimming operation is that, for each processed vertex $v$ of $G$, the outdegree (that is, the number of children) of $v$ in $T$ is bounded by $5td$.

Vertices at depth~$1$ already satisfy this property initially. Let~$v$ be a non-processed vertex at depth~$j$ whose children $u_1,\dots,u_k$ have been already processed. We process $v$ as follows. By the defining property of a treedepth decomposition, each connected component of $G-R(v)$ is such that its vertex set is either part of $V_{u_i}$, for some child $u_i$ of $v$ with $i\in \{1,\dots,k\}$, or is disjoint from $V_v$. Let~$\mathcal C$ denote the set of connected components of~$G-R(v)$ whose vertex sets are contained in the sets $V_{u_1},\dots,V_{u_k}$. We remove all components in $\mathcal C$ with one attachment and, if there are two or more distinct components in~$\mathcal C$ with the same two attachments~$\{s,t\}$, we replace all such components with a single edge~$(s,t)$. By \Cref{le:degree-3-vc}, there are at most~$3td$ pairs $\{s,t\}$ such that $\mathcal C$ contains a component whose attachment is $\{s,t\}$.  Furthermore, again by \Cref{le:degree-3-vc}, there are at most~$2td$ components in $\mathcal C$ with three or more attachments in~$R(v)$. Henceforth, the number of components in $\mathcal C$ that survive the trimming operation is at most $5td$. This is also an upper bound on the size of the outdegree of $v$ in the corresponding treedepth decomposition. 

Let~$(G_0,T_0) := (G,T)$.  For~$j=1,\dots,td$, we define~$(G_j,T_j)$ to be the graph obtained from~$(G_{j-1},T_{j-1})$ by processing all vertices at depth~$j$ as described above.  We claim that in~$(G_j,T_j)$ all vertices at depth~$j$ or less have outdegree at most~$5td$ in $T_j$.  This clearly holds for~$(G_0,T_0)$, since leaves have outdegree~$0$. Furthermore, the described processing of vertices at depth $j$ ensures that the claim holds for $(G_j,T_j)$ given that it holds for $(G_{j-1},T_{j-1})$. 

Note that, by construction, the graph $G_{td}$ consists of a single vertex, given that $R(r)$ consists just of the root $r$ and no component of $G-R(r)$ with one attachment is kept in $G_{td}$ by the trimming operation. Therefore, $G_{td}$ has a leveled planar drawing with height at most~$1$ (in fact, such an height is $0$, but for technical reasons it is more convenient to use the weaker upper bound).  Similarly to \Cref{lem:bm-span-bound}, we show that, given a leveled planar drawing of~$G_i$ with height~$h$, we can construct a leveled planar drawing of~$G_{i-1}$ of height~$h((5td)^{td}+1)$. This then implies a drawing of~$G=G_0$ of height~$((5td)^{td}+1)^{td}$, as desired.

Let~$\Gamma_i$ be a leveled planar drawing of~$G_i$ with height~$h$.  To obtain a leveled planar drawing $\Gamma_{i-1}$ of~$G_{i-1}$, we insert $(5td)^{td}$ levels immediately below every level of~$\Gamma_i$ (thus, the number of levels of $\Gamma_{i-1}$ is $h((5td)^{td}+1)$). We use these levels to re-insert the components of $G_{i-1}$ that were removed to obtain~$G_i$.  This is done as follows.  Let~$C$ be one of such components and let~$v$ be the vertex whose processing (in~$T_{i-1}$) led to the removal of~$C$.  By assumption, all vertices in~$V_v \setminus \{v\}$ have outdegree at most~$5td$ in~$T_{i-1}$ and therefore the number of vertices in $C$ is at most $(5td)^{td}$.  If~$C$ has a single attachment~$x \in R(v)$, by \Cref{lem:1-att-drawing}, we obtain a leveled planar drawing $\Gamma_C$ of~$C$ together with its attachment on~$(5td)^{td}+1$ levels in which $x$ is the unique vertex on the highest level. Note that $x$ belongs to $\Gamma_i$. We can thus use the level of~$x$ and the $(5td)^{td}$ new levels immediately below it to merge $\Gamma_C$ into~$\Gamma_i$ (by merging the two copies of~$x$).  If~$C$ has two attachments~$x,y \in R(v)$, then~$G_i$ contains the edge~$(x,y)$.  By \Cref{lem:2-att-drawing}, we have that $C$ together with~$x,y$ have a leveled planar drawing $\Gamma_C$ on~$(5td)^td+2$ levels, where~$x$ and~$y$ are the unique vertices on the highest and the lowest level, respectively. Note that $x$ and $y$ belong to $\Gamma_i$.  As~$x$ and~$y$ are drawn on distinct levels in~$\Gamma_i$, there are $(5td)^td$ new intermediate levels between them, and we can merge the drawing of~$C \cup \{x,y\}$ into~$\Gamma_i$ next to~$(x,y)$ (by merging the two copies of~$x$ and the two copies of~$y$, respectively).  By reinserting all components in this way, we obtain a leveled planar drawing~$\Gamma_{i-1}$ of~$G_{i-1}$ with height~$h ((5td)^{td}+1)$.
\end{proof}

We now proceed to devise our kernelization for an instance $(G,s)$ of {\sc $s$-Span Weakly leveled planarity} parameterized by the treedepth of $G$. As a first step, the algorithm checks whether $s>((5td)^{td}+1)^{td}$. In the positive case, the kernel consists of any graph that admits a leveled planar drawing with span at most $s$ (say, the graph consists of a single vertex). Indeed, a leveled planar drawing of $G$ with height, and hence span, at most $((5td)^{td}+1)^{td}<s$ can be constructed as in \Cref{le:drawing-treedepth}, hence the instance is positive (and the kernel is a positive instance, too). In the following, we hence assume that $s\leq ((5td)^{td}+1)^{td}$. 

Recall that $f(b)$ denotes the number of equivalence classes of planar graphs with $b+1$ vertices, one of which is fixed,  or with $b+2$ vertices, two of which are fixed (see \Cref{subse:fpt-mbsc}). We also define a computable function~$g \colon \mathbb N \times \mathbb N \to \mathbb N$ via the following recurrence: 
\[g(i,s) = 
\begin{cases}
    1 & \text{if } i=1\\
    \left((12s+12) \cdot (g(i-1,s)+1)\cdot 3td \cdot f(g(i-1,s))  + 2td\right)\cdot g(i-1,s) + 1 & \text{if } i > 1
\end{cases}
\]
Our goal is to process $T$ in a bottom-up fashion, one vertex at a time, so that, at any point of the traversal, the graph $G$ and the tree $T$ have been modified so that the following invariants are satisfied:
\begin{itemize}
\item[(I1)] the current instance $(G,s)$ is equivalent to the original instance $(G,s)$;
\item[(I2)] $T$ is still a treedepth decomposition of $G$; and 
\item[(I3)] for each processed vertex~$v$, we have $|V_v|\leq g(d(v),s)$.  
\end{itemize}
The bottom line is that, when the entire tree has been processed, the above invariants guarantee that we have an instance equivalent to the original one of size at most~$g(td,s)$. Since $s$ is bounded by a function of $td$, this is indeed a kernel with respect to $td$.

For a leaf $v$ of $T$, we do not need to perform any action. This ensures that invariants (I1) and (I2) are trivially satisfied after all the leaves of $T$ have been processed. Note that the size of $T_v$ is $1$, which coincides with $g(d(v),s)$, given that $d(v)=1$ and by the definition of the function $g$. Hence, invariant (I3) is also satisfied. 

Consider now an unprocessed vertex~$v$ of $T$ with children~$u_1,\dots,u_k$ that have been already processed.  The key intuition that we are going to use is that $R(v)$ is a modulator of size at most~$td$ to components of size at most~$g(d(v)-1,s)$ for the subgraph $G'$ of $G$ induced by the vertices in $V_v\cup R(v)$. Hence, we can apply the machinery introduced in \Cref{subse:fpt-mbsc} in order to reduce the size of $G'$ (and, consequently, of $T_v$). We now justify and describe in detail this strategy.

First, by invariant (I3), we have that $T_{u_i}$ has size at most~$g(d(v)-1,s)$, for $i=1,\dots,k$. Also, the set~$R(v)$ has size at most~$td$, as the depth of $T$ is $td$. By the defining property of a treedepth decomposition, each connected component of $G-R(v)$ is such that its vertex set is either part of $V_{u_i}$, for some $i\in \{1,\dots,k\}$, or is disjoint from $V_v$. Let~$\mathcal C$ denote the set of connected components of~$G-R(v)$ whose vertex sets are contained in the sets $V_{u_1},\dots,V_{u_k}$.   By \Cref{le:degree-3-vc}, there are at most~$2td$ such components that have three or more attachments in~$R(v)$. We reduce the number of components in~$\mathcal C$ with one and two attachments in $R(v)$, by applying \Cref{rule:bm-deg1,rule:bm-deg2}. By \Cref{le:degree-3-vc}, there are at most~$3td$ pairs $\{x,y\}$ of vertices in~$R(v)$ such that $\mathcal C$ contains a component whose attachments are $x$ and $y$. By \Cref{le:computable}, there at most~$f(g(d(v)-1,s))$ equivalence classes for the components in~$\mathcal C$ that have at most two attachments in $R(v)$ and for which these attachments are fixed. Hence, the number of components in~$\mathcal C$ with one and two attachments in $R(v)$ that are kept in $G$ is at most~$(4s+4)\cdot g(d(v)-1,s)\cdot td \cdot f(g(d(v)-1,s))$ and $(8s+8) \cdot (g(d(v)-1,s)+1)\cdot 3td\cdot f(g(d(v)-1,s))$, respectively. Recall that each of such components has at most $g(d(v)-1,s)$ vertices, as its vertex set is part of $V_{u_i}$, for some $i\in \{1,\dots,k\}$.
\begin{itemize}
\item By \Cref{lem:bm-rules-safe}, the described reduction yields an equivalent instance\footnote{The correctness of the reduction actually does not directly descend from \Cref{lem:bm-rules-safe}, as in that case \Cref{rule:bm-deg1,rule:bm-deg2} apply to all the components obtained by removing the modulator from the graph. In our case, \Cref{rule:bm-deg1,rule:bm-deg2} only apply to the components in $\mathcal C$, which are just part of the components obtained by removing the modulator $R(v)$ from $G$. However, the proof of \Cref{lem:bm-rules-safe} is not affected by the presence of components that are ignored by \Cref{rule:bm-deg1,rule:bm-deg2}, hence the very same proof of correctness applies to the current reduction, as well.}, hence satisfying Invariant (I1). 
\item The new treedepth decomposition $T$ of the reduced graph $G$ is simply obtained from the previous one by removing all vertices belonging to components in $\mathcal C$ that have been removed from $G$, and by then connecting each vertex in $T_v$ whose parent has been removed from $T$ to the lowest ancestor (in the previous tree) which still belongs to $T$. Since a pair of vertices are in an ancestor-descendant relationship in the new tree if and only if they were in the previous tree, it follows that $T$ is still a treedepth decomposition of the reduced graph $G$, thus  satisfying Invariant (I2). 
\item Finally, after the reduction, the size of~$T_v$ is at most $2td\cdot g(d(v)-1,s)$ (for the components in $C$ with three attachments in $R(v)$) plus $(12s+12)\cdot (g(d(v)-1,s)+1)\cdot 3td\cdot f(g(d(v)-1,s)) \cdot g(d(v)-1,s)$  (for the components in $C$ with one or two attachments in $R(v)$) plus one (for $v$ itself). This coincides with $g(d(v),s)$, hence satisfying Invariant (I3). 
\end{itemize}

Altogether we obtain the following theorem.

\begin{theorem} \label{thm:treedepth-FPT}
  Let $(G,s)$ be an instance of {\sc $s$-Span Weakly leveled planarity} with treedepth~$td$. There exists a kernelization that applied to $(G,s)$ constructs a kernel whose size is a computable function of~$td$. Hence, the problem is FPT with respect to the treedepth.
\end{theorem}
  


\begin{proof}
Consider an instance $(G,s)$ of the {\sc $s$-Span Weakly leveled planarity} problem, where $G$ has treedepth $td$. We distinguish two cases. 

If $s\geq ((5td)^{td}+1)^{td}$, the algorithm directly concludes that $(G,s)$ is a positive instance. Indeed, a leveled planar drawing of $G$ with height, and hence span, at most $((5td)^{td}+1)^{td}\leq s$ can be constructed as in \Cref{le:drawing-treedepth}. 

By Invariants (I1)--(I3), the reduction described before the theorem constructs an instance equivalent to $(G,s)$ whose size is a function of $td$ and $s$ only. If $s< ((5td)^{td}+1)^{td}$, such a size is a function of $td$ only. Hence, we can search by brute force whether the given instance has a solution and find one, if it exists. 
\end{proof}

\section{Upper and Lower Bounds} \label{se:combinatorial}


In this section, we establish upper and lower bounds on the span of weakly leveled planar~drawings of certain graph classes. 

The $n$-vertex {\em nested-triangle graphs}, commonly used to prove lower bounds for planar grid drawings~\cite{dlt-pepg-84,DBLP:conf/gd/FratiP07}, are easily shown to require $\Omega(n)$ span in any weakly leveled planar drawing. In contrast, outerplanar graphs admit $1$-span weakly leveled planar drawings~\cite{DBLP:journals/jgaa/FelsnerLW03}. We extend the investigation by considering two graph classes that include the outerplanar graphs, namely, graphs with outerplanarity 2 and graphs with treewidth 2. \cref{thm:2-outerplanar} proves that graphs in the former class may require linear~span.


%
\begin{figure}
    \centering
        \includegraphics{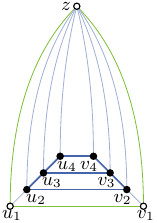}
    \caption{A $n$-vertex $2$-outerplanar graph requiring $\Omega(n)$ span in every weakly leveled planar drawing.}
    \label{fig:lower-2Outerplanar}
\end{figure}

\begin{theorem}\label{thm:2-outerplanar}
There exists an $n$-vertex $2$-outerplanar graph such that every weakly leveled planar drawing of it has span in $\Omega(n)$.
\end{theorem}

\begin{proof}
Suppose first that $n$ is odd. The lower bound is provided by the graph $G_k$ composed of $k:=\lfloor (n-1)/2\rfloor$ ``$1$-fused stacked cycles'' (see~\cref{fig:lower-2Outerplanar}): Start from $k$ cycles $\mathcal C_i:=(u_i,v_i,z_i)$, identify vertices $z_1,\dots,z_k$ into a unique vertex $z$, and insert the edges of the paths  $(u_1,\dots,u_k)$ and $(v_1,\dots,v_k)$. Biedl~\cite{DBLP:journals/dcg/Biedl11} introduced $G_k$ as an example of a $2$-outerplanar graph requiring $\Omega(n)$ width and height in any planar polyline grid drawing. We can use the same inductive proof as hers to show that one of the edges of $\mathcal C_k$ has a span $k\in \Omega(n)$ in any weakly leveled planar drawing $\Gamma_k$ of $G_k$ in which the outer face is delimited by $\mathcal C_k$. 

If $k=1$, the claim is trivial. If $k>1$, consider the subgraph $G_{k-1}$ of $G_k$ induced by $\mathcal C_1,\dots,\mathcal C_{k-1}$ and let $\Gamma_{k-1}$ be the restriction of $\Gamma_k$ to $G_k$. By induction, in $\Gamma_{k-1}$ one of the edges of $\mathcal C_{k-1}$, say $e$, has span at least $k-1$. Since $\mathcal C_{k}$ surrounds $e$ in its interior in $\Gamma_k$, except possibly for the vertex $z$ that might be shared by $e$ and $\mathcal C_{k}$, it follows that $\mathcal C_{k}$ spans at least one more level than $e$, hence the longest edge of $\mathcal C_{k}$ has span at least $k$ in $\Gamma_k$. Two observations conclude the proof. First, in any weakly leveled planar drawing of $G_k$, there is a subgraph with at least $\lceil k/2 \rceil\in \Omega(n)$ $1$-fused stacked cycles that is drawn so that these cycles are one nested into the other (this removes the assumption that the outer face has to be delimited by $\mathcal C_k$). Second, if $n$ is even, we can add to $G_k$ a vertex adjacent to $u_1$ and $v_1$; this maintains the graph $2$-outerplanar.
\end{proof}

There is a well-studied graph family, the Halin graphs, which have outerplanarity 2 and admit $1$-span weakly leveled planar drawings~\cite{DBLP:journals/algorithmica/BannisterDDEW19,digiacomo2023new}. This motivates the study of cycle-trees~\cite{DBLP:conf/isaac/LozzoDEJ17}, a superclass of Halin graphs still having outerplanarity 2. We first consider 3-connected cycle-trees and show a constant tight bound on the span. We then extend the study to general cycle-trees, for which we prove an asymptotically optimal logarithmic bound on the span. The approach for 3-connected cycle-trees relies on removing an edge from the external face so to obtain a graph for which we construct a suitable decomposition tree. We conclude by discussing the span of weakly leveled planar drawings of planar graphs with treewidth 2. We start by introducing some necessary definitions.

\subparagraph{Path-Trees.}
 A \emph{path-tree} is a plane graph $G$ that can be augmented to a cycle-tree~$G'$ by adding an edge $e=(\ell,r)$ in its outer face; see \Cref{fig:path-tree-a}. Without loss of generality, let $\ell$ occur right before $r$ in clockwise order along the outer face of $G'$; then $\ell$ and $r$ are the \emph{leftmost} and \emph{rightmost path-vertex} of $G$, respectively. The external (internal) vertices of $G'$ are \emph{path-vertices} (\emph{tree-vertices}). The tree-vertices induce a tree in $G$. We can select any tree-vertex $\rho$ incident to the unique internal face of $G'$ incident to $e$ as the \emph{root} of $G$. Then $G$ is \emph{almost-$3$-connected} if it becomes $3$-connected by adding the edges $(\rho,\ell)$, $(\rho, r)$, and $(\ell, r)$, if they are not already part of $G$. 
 

\begin{figure}[t]
  \centering
  \begin{subfigure}[b]{.45\linewidth}
    \centering
    \includegraphics[page=1]{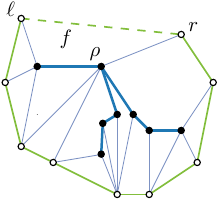}
    \subcaption{}
    \label{fig:path-tree-a}
  \end{subfigure}
  \hfill
  \begin{subfigure}[b]{.45\linewidth}
    \centering
    \includegraphics[page=2]{Figures/path-tree}
    \subcaption{}
    \label{fig:path-tree-b}
  \end{subfigure}
	\caption{(a) An almost-3-connected path-tree $G$, drawn with solid edges. The path-vertices are white and the tree-vertices are black. (b) The path-tree~$G$ with root $\rho$.}
	\label{fig:path-tree}
\end{figure}

\subparagraph{SPQ-decomposition of path-trees.} Let $G$ be an almost-$3$-connected path-tree with root~$\rho$, leftmost path-vertex $\lambda$, and rightmost path-vertex $r$. We present the \emph{SPQ-decomposition} of~$G$, introduced in~\cite{DBLP:conf/isaac/LozzoDEJ17}, which constructs a tree $\mathcal{T}$, called the \emph{SPQ-tree} of $G$. The nodes of $\mathcal T$ are of three types: {\em S-}, {\em P-}, and {\em Q-nodes}. Each node $\mu$ of $\mathcal{T}$ corresponds to a subgraph $G_\mu$ of $G$, called the \emph{pertinent graph} of~$\mu$, which  is an almost-$3$-connected rooted path-tree. We denote by $\rho_\mu$ the root of $G_{\mu}$ (a tree-vertex), by $\lambda_\mu$ the leftmost path-vertex of $G_\mu$, and by $r_\mu$ the rightmost path-vertex of $G_\mu$. 
To handle the base case,
we consider as a path-tree also a graph whose path is the single edge $(\lambda,r)$ and whose tree consists of a single vertex $\rho$, possibly adjacent to only one of $\lambda$ and $r$. Also, we consider as almost-$3$-connected a path-tree such that adding $(\rho, r)$, $(\rho, \lambda)$, and $(\lambda, r)$, if missing, yields a $3$-cycle. 

We now describe~the~decomposition. 

\begin{itemize}
\item{\sc Q-node}: the pertinent graph $G_\mu$ of a \emph{Q-node} $\mu$ is an almost-$3$-connected rooted path-tree which consists of $\rho_\mu$,~$\lambda_\mu$, and~$r_\mu$. The edge $(\lambda_\mu, r_\mu)$ belongs to
$G_\mu$, while the edges $(\rho_\mu,\lambda_\mu)$ and $(\rho_\mu, r_\mu)$ may not exist; see~\cref{fig:nodes}(left).

\item{\sc S-node}: the pertinent graph $G_\mu$ of an \emph{S-node} $\mu$ is an almost-$3$-connected rooted path-tree which consists of $\rho_\mu$ and of an almost-$3$-connected path-tree $G_\nu$, where $\rho_\mu$ is adjacent to $\rho_\nu$ and, possibly, to $\lambda_\nu$ and $r_\nu$. We have that $\mu$ has a unique child in $\mathcal T$, namely a node $\nu$ whose pertinent graph is $G_\nu$. Further, we have $\lambda_\mu=\lambda_\nu$ and $r_\mu=r_\nu$; see~\cref{fig:nodes}(middle).

\item{\sc P-node}: the pertinent graph $G_\mu$ of a \emph{P-node} $\mu$ is an almost-$3$-connected rooted path-tree which consists of almost-$3$-connected rooted path-trees $G_{\nu_1},\dots,G_{\nu_k}$, with $k>1$. This composition is defined as follows. First, we have $\rho_\mu=\rho_{\nu_1}=\dots=\rho_{\nu_k}$. Second, we have  $\lambda_{\nu_i}=r_{\nu_{i-1}}$, for $i=2,\dots,k$. Third, $\mu$ has children $\nu_1,\dots,\nu_k$ (in this left-to-right order) in $\mathcal T$, where $G_{\nu_i}$ is the pertinent graph of $\nu_i$, for $i=1,\dots,k$. Finally, we have $\lambda_{\mu}=\lambda_{\nu_1}$ and $r_{\mu}=r_{\nu_k}$; see~\cref{fig:nodes}(right). 
\end{itemize}

\begin{figure}[t]
	\centering
		\includegraphics[page=2,scale=.9]{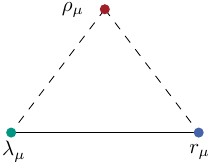}
	\hfil
		\includegraphics[page=3,scale=.9]{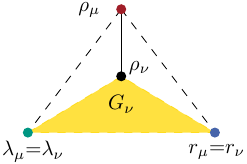}
	\hfil
		\includegraphics[page=4,scale=.9]{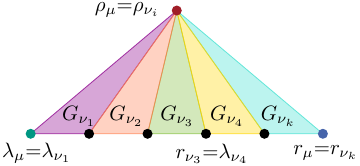}
	\caption{Path-trees corresponding to a Q-node (left), an S-node (middle), and a P-node (right). Dashed edges may or may not belong to $G_{\mu}$. Shaded triangles represent smaller path-trees $G_{\nu_i}$.}
	\label{fig:nodes}
\end{figure}

In the following, all considered SPQ-trees are \emph{canonical}, that is, the child of every P-node is an S- or Q-node. For a given path-tree, a canonical SPQ-tree always exists~\cite{DBLP:journals/algorithmica/ChaplickLGLM24}.

\subsection{3-connected Cycle-Trees} 
%
%
%
%
Let $G$ be a plane graph with three consecutive vertices $u,v,w$ encountered in this order when walking in clockwise direction along the boundary of the outer face of $G$.  A leveling of~$G$ is \emph{single-sink with respect to $(u,v,w)$} if all vertices of $G$ have a neighbor on a higher level, except for exactly one of $\{u,v,w\}$. A single-sink leveling $\ell$ with respect to $(u,v,w)$ is \emph{flat} if $\ell(u) < \ell(v) < \ell(w)$ or $\ell(w) < \ell(v) < \ell(u)$; $\ell$ is a \emph{roof} if $\ell(v) > \ell(u)$ and  $\ell(v) > \ell(w)$.  Note that a single-sink leveling is necessarily either roof or flat.

Given a single-sink leveling $\ell$ of $G$ with respect to $(u,v,w)$, 
%
a \emph{good} weakly leveled planar drawing $\Gamma$ of $(G,\ell)$ is one with the following properties: 
\begin{enumerate}
\item $\Gamma$ respects the planar embedding of $G$;
\item it holds that $x(u) < x(w)$ in  $\Gamma$; and
\item all vertices of $V(G) \setminus \{u,v,w\}$ are contained in the interior of the bounded region~$R_{uvw}$ defined by the path~$(u,v,w)$, by the vertical rays starting at~$u$ and~$w$, and by the horizontal line $y:=\min_{z \in V(G)}y(z)$.
\end{enumerate}

Let $a$ and $b$ be two non-zero integers. A good weakly leveled planar drawing $\Gamma$ of $(G,\ell)$ is an \emph{(a,b)-flat drawing} if $\ell$ is flat, $a = \ell(v) - \ell(u)$, and $b = \ell(w) - \ell(v)$; it is an \emph{(a,b)-roof drawing} if $\ell$ is roof, $a = \ell(v) - \ell(u)$, and  $b = \ell(w) - \ell(v)$. Note that, by definition, in an $(a,b)$-flat drawing we have that $a$ and $b$ are either both positive or both negative, while in an $(a,b)$-roof drawing $a$ is positive and $b$ is negative.
%
%

In order to prove an upper bound on the span of weakly leveled planar drawings of almost-$3$-connected path-trees, we will construct a drawing such that the non-horizontal edges connecting vertices on different levels are $y$-monotone polygonal chains whose consecutive segments meet~at~a~level. We are going to exploit the following.

\begin{lemma}\label{lem:geometric-realizability}
   Let $G$ be a plane graph and $u,v,w$ be three vertices that are consecutive in clockwise order along the outer face of $G$. Let~$\ell$ be a single-sink leveling of $G$ with respect to $(u,v,w)$ such that $(G,\ell)$ admits a weakly leveled planar drawing that respects the embedding of~$G$. Let $\Gamma_0$ be any crossing-free $y$-monotone polyline drawing of the path $(u,v,w)$ such that $y(u)=\ell(u)$, $y(v)=\ell(v)$, $y(w)=\ell(w)$, and $x(u) < x(w)$. 
   
   Then $(G,\ell)$ admits a polyline good weakly leveled planar drawing contained in $R_{uvw}$, in which the path $(u,v,w)$ is represented by $\Gamma_0$.
\end{lemma}

\begin{proof}
    Let~$\Gamma$ be a weakly leveled planar drawing of $(G,\ell)$ that respects the given embedding of $G$.  Let~$\Gamma'$ be obtained by subdividing each non-horizontal edge of $\Gamma$ at each level it crosses and replacing the curves connecting consecutive vertices by straight-line segments. 
    Note that $\Gamma'$ is a weakly leveled planar drawing such that each edge has span either $0$ or $1$ of a graph~$(G',\ell')$, where $G'$ is a subdivision of $G$. 
    By \cref{obs:gioproperty}, it follows that moving the vertices of $G'$ within each level of $\ell'$ neither changes the embedding of $\Gamma'$ nor violates its planarity as long as we keep the order of the vertices within each level.  
    Hence, we can move the vertices of $G'$ in such a way that the drawing of the path~$(u,v,w)$ is~$\Gamma_0$ and all other vertices are in the interior of the region $R_{uvw}$. The obtained drawing corresponds to a good weakly leveled planar drawing of $G$ contained in $R_{uvw}$, in which the path $(u,v,w)$ is represented by $\Gamma_0$.
\end{proof}

\noindent We exploit \cref{lem:geometric-realizability} to prove the following lemma.


\begin{lemma}\label{lem:triconnected-cycle-paths}
Let $G$ be an almost $3$-connected path-tree, let $\mathcal{T}$ be its SPQ-tree, let~$\mu$ be the root of $\mathcal{T}$, and let $u=\lambda_\mu$, $v=\rho_\mu$, and $w = r_\mu$. Then $G$ has a weakly leveled planar drawing with span at most $4$. Also, if the edges $(u,v)$ and $(v,w)$ are in $G$, then $G$ has a leveling~$\ell_\mu$ with $\spn(\ell_\mu) \leq 4$ such that $(G_\mu,\ell_\mu)$ admits a $(1,1)$-flat weakly leveled planar drawing. 
\end{lemma}   
\begin{proof}
Since the removal of edges does not increase the span of a weakly leveled planar drawing, we can make a couple of assumptions on $G$ that will simplify the proof. First, we will assume that the edges $(u,v)$ and $(v,w)$ are part of $G$. Second, since $G$ can always be augmented with edges connecting path vertices with internal vertices to an internally-triangulated almost $3$-connected path-tree having the same outerface as $G$, we will assume that $G$ is internally triangulated. In other words, in the following, we will prove the statement when $G$ is a maximal almost $3$-connected path-tree.
    
The statement of the lemma is implied by the following claim, stating that $G$ admits certain types of $4$-span weakly leveled planar drawings, based on the type of $\mu$. 

\begin{claim}\label{cl:nodes-span}
   Let $\mu$ be a node of $\mathcal{T}$. Then, the following holds:
   \begin{description}
   \item[$\mu$ is a P-node.]\label{claim:p-nodes}
     Then $G=G_\mu$ has levelings $\ell^i_\mu$ for $i=1,\dots,4$, with $\spn(\ell^i_\mu) \leq 4$, such that $(G_\mu,\ell^i_\mu)$ admits a $p_i$-flat weakly leveled planar drawing with $p_1 = (-1,-1)$, $p_2 = (1,1)$, $p_3 = (-1,-3)$, and $p_4 = (3,1)$.
   \item[$\mu$ is a Q- or an S-node.]\label{claim:qs-nodes}
      Then $G_\mu$ has levelings $\ell^i_\mu$ for $i=1,\dots,4$, with $\spn(\ell^i_\mu) \leq 4$, such that $(G_\mu,\ell^i_\mu)$ admits a $q_i$-flat weakly leveled planar drawing with $q_1 = (1,1)$, $q_2 = (-1,-1)$, $q_3 = (3,1)$, and $q_4 = (-1,-3)$, and levelings $\ell^j_\mu$ for $j=5,6$, with $\spn(\ell^j_\mu) \leq 4$, such that $(G_\mu,\ell^j_\mu)$ admits a $q_j$-roof weakly leveled planar drawing
        with $q_5 = (1,-3)$ and $q_6 = (3,-1)$.
    \end{description}
   \end{claim}
   \begin{nestedproof}
    We prove the claim by induction on the height of~$\mathcal{T}$.  

    In the base case, the height of~$\mathcal{T}$ is~$0$, i.e., $\mu$ is the unique node of $\mathcal{T}$. Clearly, $\mu$ is a Q-node. It is immediate to see that $G_\mu$ admits levelings $\ell^i_\mu$, for $i=1,\dots,4$ and $\ell^j_\mu$, for $j=5,6$,  with span at most $4$, as required by the statement.  \cref{fig:Q-nodes} shows corresponding examples of weakly leveled planar drawings of $(G_\mu,\ell^i_\mu)$ and of $(G_\mu,\ell^j_\mu)$. Note that only the $(3,1)$-flat and $(-1,3)$-flat weakly leveled planar drawings contain an edge with span $4$ (namely, the edge $(u,w)$).

    \begin{figure}[t]
    \centering
        \includegraphics[width = \textwidth, page=2]{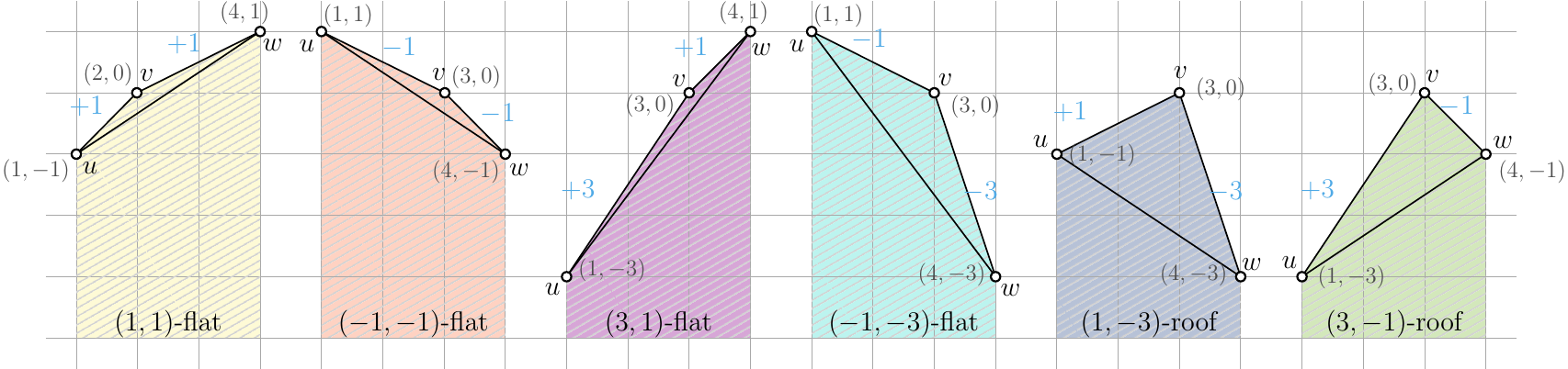}
        \caption{Illustrations for the proof of \cref{claim:qs-nodes}, when $\mu$ is a Q-node and the height of $\mathcal{T}$ is $0$.}
    \label{fig:Q-nodes}
    \end{figure}
    
    Now assume that the height of $\mathcal{T}$ is greater than~$0$, which implies that $\mu$ is either a P-node or an S-node.  Let~$\nu_1,\dots,\nu_k$ denote the children of~$\mu$, with $k \geq 1$ in the left-to-right order in which they appear in $\mathcal{T}$. By inductive hypothesis, we can assume that each graph $G_{\nu_i}$ has levelings admitting weakly leveled planar drawings as required by the statement. We distinguish two cases based on the type of~$\mu$.

    \begin{figure}[t]
    \begin{minipage}{\textwidth}
    \centering
        \includegraphics[width=0.9\textwidth,page=3]{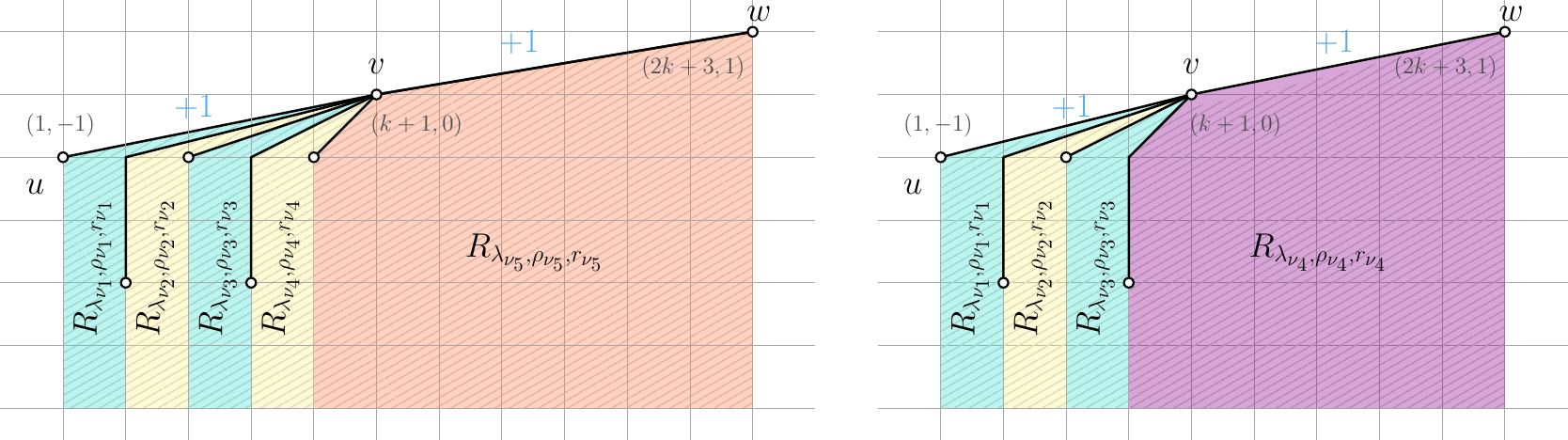}
        \subcaption{$(1,1)$-flat}
    \label{fig:P-node-drawing-11}
    \end{minipage}

    \medskip
    \begin{minipage}{\textwidth}
    \centering
        \includegraphics[width=0.9\textwidth,page=4]{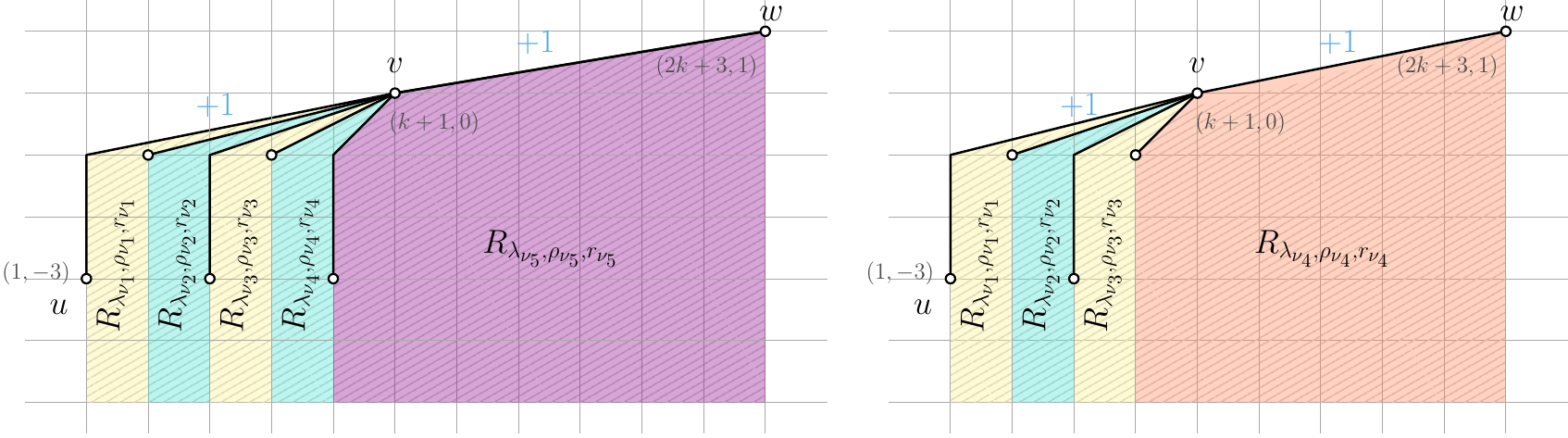}
        \subcaption{$(3,1)$-flat}
    \label{fig:P-node-drawing-31}
    \end{minipage}
    \caption{Illustrations for the construction of a flat weakly leveled planar drawing of the pertinent graph of a P-node, with odd number (left) and even number (right) of children.}
    \label{fig:P-node-drawing}
    \end{figure}

    \noindent {\bf $\mu$ is a P-node.} Since~$\mu$ is a P-node, each $\nu_i$ is an S- or a Q-node ($1 \leq i \leq k$). Note that $u=\lambda_{\mu}=\lambda_{\nu_1}$, $v=\rho_\mu=\rho_{\nu_1}=\ldots=\rho_{\nu_k}$, $w=r_{\mu}=r_{\nu_k}$ and that $r_{\nu_{i-1}} = \lambda_{\nu_i}$ for $1 < i \leq k$. 
    
    We first show how to construct a $(1,1)$-flat weakly leveled planar drawing of $G_\mu$; refer to \cref{fig:P-node-drawing-11}. To this aim, 
    we draw $u$ at point $(1,-1)$, $v$ at point $(k+1,0)$, $w$ at point $(2k+3,1)$, and the edges $(u,v)$ and $(v,w)$ as straight-line segments. For any odd $i\in [k]$, we place $\lambda_{\nu_i}$ at coordinate $(i,-1)$ and draw the edge $(v,\lambda_{\nu_i})$ as a straight-line segment. For any even $i\in [k]$, we place $\lambda_{\nu_i}$ at coordinate $(i,-3)$ and draw the edge $(v,\lambda_{\nu_i})$ with one bend at coordinate $(i,-1)$. Note that $r_{\nu_i}=\lambda_{\nu_{i+1}}$, for $i=1,\dots,k-1$, and $r_{\nu_k}=w$ have been drawn.

   This construction defines, for every child $\nu_i$ of $\mu$, a drawing of the path $(\lambda_{\nu_i},\rho_{\nu_i},r_{\nu_i})$, where $\rho_{\nu_i}$ coincides with $v$. 
    By induction, we have that 
    \begin{enumerate*}[label=(\roman*)]
        \item for every odd $i\in [k-1]$, $G_{\nu_i}$ has a leveling $\ell_i$ that admits a $(1,-3)$-roof weakly leveled planar drawing;
        \item for every even $i\in [k-1]$, $G_{\nu_i}$ has a leveling $\ell_i$ that admits a $(3,-1)$-roof weakly leveled planar drawing;
        \item $G_{\nu_k}$ has a leveling that admits a $(1,1)$-flat weakly leveled planar drawing with span at most $4$ and it has a leveling that admits a $(3,1)$-flat weakly leveled planar drawing  with span at most $4$. 
    \end{enumerate*}  

    By construction, for every $i \in [k]$, we have that $y(\lambda_{\nu_i}) = \ell_i(\lambda_{\nu_i})$, $y(\rho_{\nu_i}) = \ell_i(\rho_{\nu_i}) = y(v)$, and $y(r_{\nu_i}) = \ell_i(r_{\nu_i})$. Moreover, again by construction, we have that $x(\lambda_{\nu_i}) < x(r_{\nu_i})$, for every $i \in [k]$. Therefore, we can apply \cref{lem:geometric-realizability} to obtain a crossing-free $y$-monotone polyline drawing $\Gamma_{\nu_i}$ of $(G_{\nu_i}, \ell_i)$ inside the region $R_{\lambda_{\nu_i}\rho_{\nu_i}r_{\nu_i}}$, for every $i \in [k]$. The union of these drawings is a $(1,1)$-flat weakly leveled planar drawing of $G_\mu$ with span at most $4$. Note that the $(1,1)$-flat weakly leveled planar drawing of $G_{\nu_k}$ is used when $k$ is odd and the $(3,1)$-flat weakly leveled planar drawing is used when $k$ is even.   

    The construction of a $(3,1)$-flat weakly leveled planar drawing of $G_\mu$ is similar (refer to \cref{fig:P-node-drawing-31}), 
    except that we place $u$ at point $(1,-3)$, draw the edge $(u,v)$ as a 1-bend polyline that starts at $u$, goes to point $(1,-1)$, and then to $v$, and we exchange ``odd'' and ``even'' in the construction.     
    The constructions of a $(-1,-1)$-flat and of a $(-3,-1)$-flat weakly leveled planar drawing are symmetric to the constructions of a $(1,1)$-flat and of a $(3,1)$-flat weakly leveled planar drawing, respectively.

    \noindent {\bf $\mu$ is an S-node.}
    In this case, $k=1$ holds, that is, there is a single child $\nu=\nu_1$, which is either an S-node, or a P-node, or a Q-node. Recall that $u=\lambda_{\mu}=\lambda_{\nu}$, $v=\rho_\mu$, $w=r_{\mu}=r_{\nu}$; moreover, $(\rho_\mu,\rho_\nu)$ is an edge of $G_\mu$ not in $G_\nu$, and hence $\rho_\mu \neq \rho_\nu$.    

    \begin{figure}[t]
    \centering
    \begin{minipage}{.32\textwidth}
        \centering
        \includegraphics[width=0.9\textwidth,page=7]{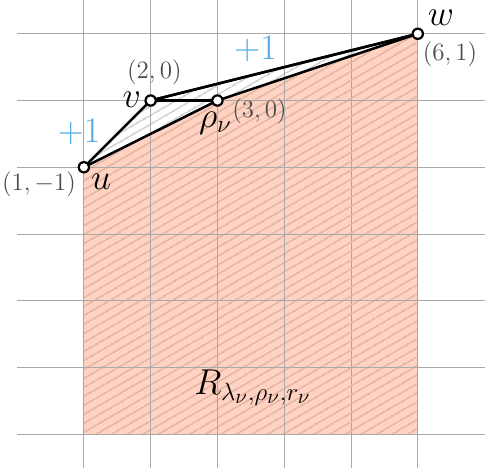}
        \subcaption{$(1,1)$-flat}
        \label{fig:S-node-drawing-11}
    \end{minipage}
    \begin{minipage}{.32\textwidth}
        \centering
        \includegraphics[width=0.9\textwidth,page=10]{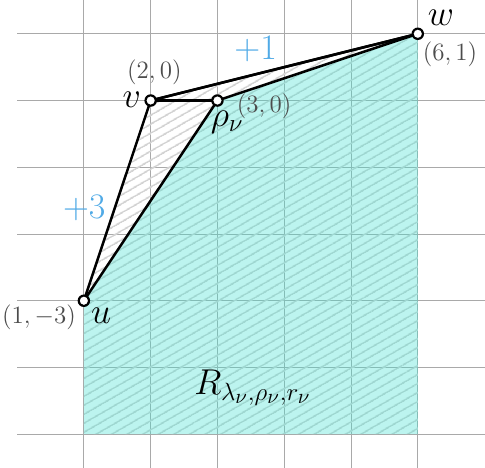}
        \subcaption{$(3,1)$-flat}
        \label{fig:S-node-drawing-31}
    \end{minipage}
    \begin{minipage}{.32\textwidth}
        \centering
        \includegraphics[width=0.9\textwidth,page=8]{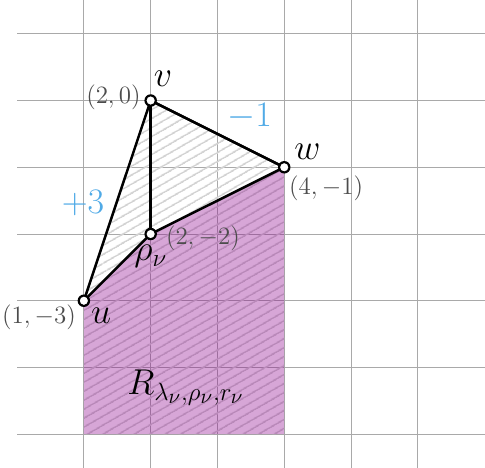}
        \subcaption{$(3,-1)$-roof}
        \label{fig:S-node-drawing-roof}
    \end{minipage}
    \caption{Illustrations for the construction of a weakly leveled planar drawing of the pertinent graph of an S-node.}
    \label{fig:S-node-drawing}
    \end{figure}
    
    We first show how to construct a $(1,1)$-flat weakly leveled planar drawing of $G_\mu$ (see \cref{fig:S-node-drawing-11}); the construction of a $(-1,-1)$-flat weakly leveled planar drawing of $G_\mu$ is symmetric.
    We draw $u$ at point $(1,-1)$, $v$ at point $(2,0)$, $w$ at point $(6,1)$,  $\rho_\nu$ at point $(3,0)$,
    and all edges between $u$, $v$, $w$ and $\rho_\nu$ as straight-line segments. 
    Since $\nu$ is either an S-, P- or Q-node, by induction, we have that $G_\nu$ has a leveling $\ell_\nu$ such that $(G_\nu,\ell_\nu)$ admits a $(1,1)$-flat weakly leveled planar drawing with span at most $4$. By construction, we have that $y(\lambda_\nu)= \ell_\nu(\lambda_\nu) = y(u)$, $y(\rho_\nu) = \ell_\nu(\rho_\nu)$, and $y(r_\nu)=\ell_\nu(r_\nu) = y(w)$. Moreover, again by construction, we have that $x(\lambda_\nu) < x(r_\nu)$. Therefore, we can apply \cref{lem:geometric-realizability} to obtain a crossing-free $y$-monotone polyline drawing $\Gamma_{\nu}$ of $G_\nu$ inside the region $R_{\lambda_{\nu}\rho_{\nu}r_{\nu}}$. The union of~$\Gamma_{\nu}$ and of the drawing of the subgraph of $G_\mu$ induced by $u$, $v= \lambda_\nu$, $w=r_\nu$, and $\rho_\nu$ described above is a $(1,1)$-flat weakly leveled planar drawing of $G_\mu$ with span at most $4$.
    
    We next show how to construct a $(3,1)$-flat weakly leveled planar drawing of $G_\mu$ with span at most $4$; the construction of a $(-1,-3)$-flat weakly leveled planar drawing of $G_\mu$ with span at most $4$ is symmetric. 
    We draw $u$ at point $(1,-3)$, $v$ at point $(2,0)$, $w$ at point $(6,1)$, $\rho_\nu$ at point $(3,0)$,
    and all edges between $u$, $v$, $w$, and $\rho_\nu$ as straight-line segments.
    By induction, $G_\nu$ admits a leveling $\ell_\nu$ such that $(G_\nu,\ell_\nu)$ has a $(3,1)$-flat weakly leveled planar drawing with span at most $4$. Moreover, the coordinates of the path $(\lambda_{\nu}{=}u,\rho_\nu,r_\nu{=}w)$ satisfy the conditions of \cref{lem:geometric-realizability}. Hence, 
    by \cref{lem:geometric-realizability}, we can obtain a crossing-free $y$-monotone polyline drawing of $G_\nu$ inside the region $R_{\lambda_{\nu}\rho_{\nu}r_{\nu}}$ with span at most $4$. This completes the construction of the desired $(3,1)$-flat weakly leveled planar drawing of $G_\mu$ with span at most~$4$.
    
    Finally, we show how to construct a $(3,-1)$-roof weakly leveled planar drawing of $G_\mu$ with span at most $4$ (see \cref{fig:S-node-drawing-roof}); the construction of a  $(1,-3)$-roof weakly leveled planar drawing of $G_\mu$ with span at most $4$ is symmetric.
    We draw $u$ at point $(1,-3)$, $v$ at point $(2,0)$, $w$ at point $(4,-1)$, $\rho_\nu$ at point $(2,-2)$,
    and all edges between $u$, $v$, $w$, and $\rho_\nu$ as straight-line segments.
    By induction, $G_\nu$ admits a leveling $\ell_\nu$ such that $(G_\nu,\ell_\nu)$ has a $(1,1)$-flat weakly leveled planar drawing with span at most $4$. Moreover, the coordinates of the path $(\lambda_{\nu}{=}u,\rho_\nu,r_\nu{=}w)$ satisfy the conditions of \cref{lem:geometric-realizability}. Hence, 
    by \cref{lem:geometric-realizability}, we can obtain a crossing-free $y$-monotone polyline drawing of $G_\nu$ inside the region $R_{\lambda_{\nu}\rho_{\nu}r_{\nu}}$ with span at most $4$. This completes the construction of the desired $(3,-1)$-roof weakly leveled planar drawing of $G_\mu$ with span at most $4$.
    \end{nestedproof}
    As mentioned before \cref{cl:nodes-span} completes the proof of \cref{lem:triconnected-cycle-paths}.
\end{proof}

\noindent In the following, we use \cref{lem:triconnected-cycle-paths} to get the following result on $3$-connected cycle-trees.

\begin{theorem}\label{lem:upper-cycle-tree-3conn}
Every $3$-connected cycle-tree admits a $4$-span weakly leveled planar drawing.
\end{theorem}
\begin{proof}
Let $G$ be a $3$-connected cycle-tree and let $(u,w)$ be an edge on the outer face of~$G$. Since $G$ can be augmented with edges connecting external with internal vertices to an internally-triangulated $3$-connected cycle-tree with the same outer face as $G$, we can assume that $G$ is internally triangulated.
Let $v$ be the internal vertex on the internal face incident to $(u,w)$. 
We obtain an almost-$3$-connected path-tree $G'$ by removing the edge $(u,w)$.

Let $\mathcal T$ be the SPQ-tree of $G'$ with root node $\mu$ such that $u=\lambda_\mu$, $v=\rho_\mu$, and $w=r_\mu$. By \cref{lem:triconnected-cycle-paths}, we have that $G'$ admits a (1,1)-flat 4-span weakly leveled planar drawing $\Gamma'$. We construct a 4-span weakly leveled planar drawing $\Gamma$ of $G$ by reinserting the edge $(u,w)$ with span 2 in $\Gamma'$ by introducing a bend at the level of $v$.
\end{proof}

\begin{figure}
    \centering
    \includegraphics[page=1]{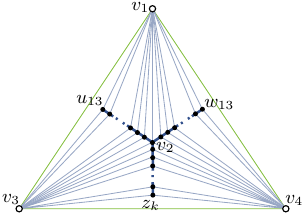}
    \caption{An $n$-vertex $3$-connected cycle-tree requiring span $4$ in every weakly leveled planar drawing.}
    \label{fig:lower-cycle-tree-3conn}
\end{figure}

The approach in the proof of \cref{lem:upper-cycle-tree-3conn} can be implemented in quadratic time. To get linear time, we can maintain only the order of the vertices on their levels and calculate the exact coordinates at the end of the algorithm. Our upper bound is matched by a lower bound, showing that some $3$-connected cycle-trees require span at least $4$; see~\cref{fig:lower-cycle-tree-3conn}.

\begin{theorem} \label{lem:lower-cycle-tree-3conn}
For all $n\geq 43$, there exists an $n$-vertex $3$-connected cycle-tree $G$ such that every weakly leveled planar drawing of $G$ has span greater than or equal to $4$.
\end{theorem}
\begin{proof}
The $n$-vertex cycle-tree $G$ proving the lower bound is constructed as follows; see~\cref{fig:lower-cycle-tree-3conn}. Starting from a complete graph on four vertices $\{v_1, v_2, v_3, v_4\}$, insert three paths $(u_1,\dots,u_{13})$, $(w_1,\dots,w_{13})$, and $(z_1,\dots,z_k)$, where $k=n-30\geq 13$. Connect $u_1$ to $v_1$, $v_2$, and $v_3$; for $i=2,\dots,13$ connect $u_i$ to $v_1$ and $v_3$; connect $w_1$ to $v_1$, $v_2$, and $v_4$; for $i=2,\dots,13$ connect $w_i$ to $v_1$ and $v_4$;  connect $z_1$ to $v_2$, $v_3$, and $v_4$; and for $i=2,\dots,k$ connect $z_i$ to $v_3$ and $v_4$.

Consider any weakly leveled planar drawing $\Gamma$ of $G$ and let $\ell$ be its associated leveling.  Suppose first that one of the edges $v_1v_3$, $v_1v_4$, $v_3v_4$ has span~$0$ or $1$.  Without loss of generality, we may assume $\spn_\ell(v_1v_3) \in \{0, 1\}$ and $\ell(v_1)\leq \ell(v_3)$, as the case $\ell(v_3)<\ell(v_1)$ is analogous. Let $A$ and $B$ be the subsets of vertices among $\{u_1,\dots,u_{13}\}$ whose level is not smaller than $\ell(v_3)$ and not larger than $\ell(v_1)$, respectively. Note that $|A|\geq 7$ or $|B|\geq 7$. Suppose the former, as the other case is analogous. At most two vertices in $A$ lie on the same level, as edges from $v_1$ and $v_3$ to three vertices in $A$ on the same level would define a crossing. It follows that the vertices in $A$ occupy at least four levels, not lower than the level of $v_3$. Then the edge from $v_1$ to the vertex in $A$ on the highest level has span at least $4$. 

We can thus assume that the edges $v_1v_3$, $v_1v_4$, and $v_3v_4$ all have span at least $2$ in $\Gamma$. Let $v_s$ and $v_t$ be the vertices among $v_1$, $v_3$, and $v_4$ with the smallest and largest level, respectively; let $v_m$ be the vertex among $v_1$, $v_3$, and $v_4$ distinct from $v_s$ and $v_t$. Then the span of the edge~$v_sv_t$ in $\Gamma$ is 
$\spn_\ell(v_sv_t) = \spn_\ell(v_sv_m) + \spn_\ell(v_sv_t)$, hence it is at least $4$.
\end{proof}

Similar to~\cite[Lemma~14]{DBLP:journals/dmtcs/DujmovicPW04}, one can prove that $s$-span weakly leveled planar graphs have queue number at most $s+1$ as in~the~following.


\begin{lemma}\label{lem:span-queue}
Weakly leveled planar graphs with span $s$ have queue number at most $s+1$.
\end{lemma}

\begin{proof}
Let $G$ be a weakly leveled planar graph with span at most $s$, and let $\Gamma$ be a corresponding weakly leveled planar drawing of $G$. It is known that every leveled planar graph with span at most $s$ admits a queue layout with $s$ queues
\cite[Lemma~14]{DBLP:journals/dmtcs/DujmovicPW04}, in which vertices of the same level appear consecutively in the queue layout in the same left to right order as in the leveled drawing. This implies that all edges of $\Gamma$ with span greater than zero fit in $s$ queues. The remaining ones can be accommodated in a single additional queue, as they cannot be nested in the underlying linear order, completing the proof.
\end{proof}

\noindent By \cref{lem:upper-cycle-tree-3conn,lem:span-queue}, we obtain the following.

\begin{corollary}
The queue number of $3$-connected cycle-trees is at most $5$. 
\end{corollary}

The \emph{edge-length ratio} of a straight-line graph drawing is the maximum ratio between the Euclidean lengths of $e_1$ and $e_2$, over all edge pairs $(e_1,e_2)$. The \emph{planar edge-length ratio} of a planar graph $G$ is the infimum edge-length ratio of $\Gamma$, over all planar straight-line drawings $\Gamma$ of $G$. 
Constant upper bounds on the planar edge-length ratio are known for outerplanar graphs~\cite{DBLP:journals/tcs/LazardLL19} and for Halin graphs~\cite{digiacomo2023new}.
%
%
We exploit the property that graphs that admit $s$-span weakly leveled planar drawings have planar edge-length ratio at most $2s+1$~\cite[Lemma~4]{digiacomo2023new} to
%
obtain a constant upper bound on the edge-length ratio of 3-connected cycle trees.


\begin{corollary}\label{co:edge-length-cycle-tree-3conn}
The planar edge-length ratio of $3$-connected cycle-trees is at most $9$. 
\end{corollary}


\subsection{General Cycle-Trees} 
We now discuss general cycle-trees, for which we can prove a $\Theta(\log n)$ bound on the span of their weakly leveled planar drawings. We start with the upper bound.

\begin{theorem}\label{lem:upper-general-cycle-trees}
Every $n$-vertex cycle-tree $G$ has an $s$-span weakly leveled planar drawing such that $s \in O(\log n)$.
\end{theorem}
\begin{proof}
We can assume that $G$ is connected, as distinct connected components of $G$ can be laid out independently on the same levels. Then $G$ has a plane embedding $\mathcal E$ in which the outer face is delimited by a walk $W$, so that by removing the vertices of $W$ and their incident edges from $G$ one gets a tree~$T$; see~\cref{fig:cycle-2conn-a}. We now add the maximum number of edges in the outer face of $\mathcal E$, while preserving planarity, simplicity, and the property that every vertex of $W$ is incident to the outer face. Further, we add the maximum number of edges inside the internal faces of $\mathcal E$, connecting vertices of $W$ with vertices of $W$ and with vertices of $T$. This latter augmentation is performed until no edge can be added as described while preserving planarity and simplicity. Clearly, this maintains the property that $G$ is a cycle-tree; see~\cref{fig:cycle-2conn-b}.

\begin{figure}
    \centering
    \begin{subfigure}{0.48\textwidth}
    \centering
        \includegraphics[page=1]{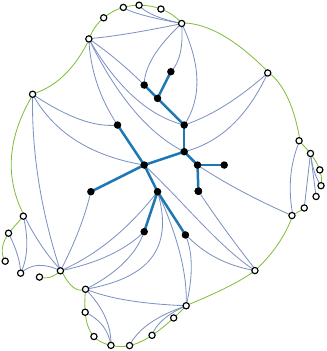}
        \subcaption{}
        \label{fig:cycle-2conn-a}
    \end{subfigure}
    \begin{subfigure}{0.48\textwidth}
    \centering
        \includegraphics[page=2]{Figures/CycleTrees-2Conn.pdf}
        \subcaption{}
        \label{fig:cycle-2conn-b}
    \end{subfigure}
    \hfill
    \caption{(a) An embedding $\mathcal E$ of a cycle-tree $G$ such that removing the vertices of the walk $W$ delimiting the outer face one gets a tree $T$, which is represented by bold lines. (b) The augmentation of $G$; the face $f$ of $\mathcal E_{G_W}$ is shaded gray.}
    \label{fig:cycle-2conn}
\end{figure}

Now the plan is as follows. We are going to remove some parts of $G$, so that $G$ turns into a $3$-connected cycle-tree $H$. We then apply~\cref{lem:upper-cycle-tree-3conn} to construct a weakly leveled planar drawing $\Lambda$ of $H$ with $O(1)$ span. We next insert $O(\log n)$ levels between any two consecutive levels of $\Lambda$. We use such levels in order to re-introduce the parts of $G$ that were previously removed, thus obtaining a  weakly leveled planar drawing $\Gamma$ of $G$ with $O(\log n)$ span. 

First, note that the subgraph $G_W$ of $G$ induced by the vertices of $W$ is outerplanar (and $2$-connected, because of the described augmentation). Let $\mathcal E_{G_W}$ be the restriction of $\mathcal E$ to $G_W$. Since $T$ is connected, it lies inside a single internal face $f$ of $\mathcal E_{G_W}$; let $C$ be the cycle delimiting $f$. We remove from $G$ the vertices of $G_W$ not in $C$ and their incident edges. The removed vertices induced connected components $H^1_W,\dots,H^k_W$ of $G_W$. For $i=1,\dots,k$, we call \emph{component of $G$ outside $C$} the subgraph $G^i_W$ of $G$ induced by the vertices of $H^i_W$ and the two vertices of $C$ that are neighbors of vertices in $H^i_W$. Note that $G^i_W$ is a maximal outerplanar graph, because of the initial augmentation. \cref{fig:cycle-2conn-removed-a} shows the result of the described removal on the graph from~\cref{fig:cycle-2conn-b}.

\begin{figure}[htb]
    \centering
    \begin{subfigure}{0.48\textwidth}
    \centering
        \includegraphics[page=3]{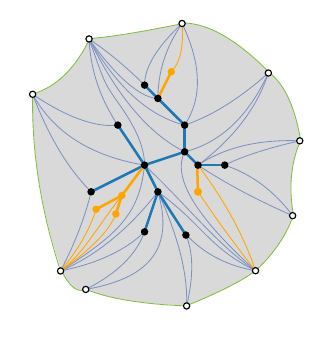}
        \subcaption{}
        \label{fig:cycle-2conn-removed-a}
    \end{subfigure}
    \begin{subfigure}{0.48\textwidth}
    \centering
        \includegraphics[page=4]{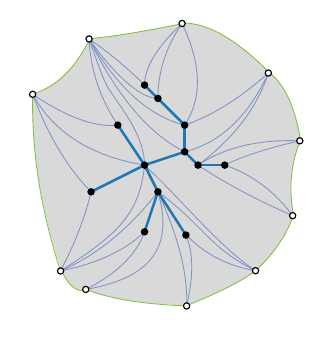}
        \subcaption{}
        \label{fig:cycle-2conn-removed-b}
    \end{subfigure}
    \hfill
    \caption{(a) Removing the components of $G$ outside $C$ (except for the vertices and edges of $C$). (b) Removing the components of $G$ inside $C$.}
    \label{fig:cycle-2conn-removed}
\end{figure}
Second, we remove from $G$ all the vertices of $T$ that have at most one neighbor in $C$, as well as their incident edges. \cref{fig:cycle-2conn-removed-b} shows the result of this removal on the graph from~\cref{fig:cycle-2conn-removed-a}. Let $H$ be resulting graph and let  $\mathcal E_H$ be the restriction of $\mathcal E$ to $H$. Note that the outer face of $\mathcal E_H$ is delimited by $C$. We now prove that $H$ is a cycle-tree, that is, by removing $C$ from $H$ one is left with a tree $Y$. This is equivalent to saying that the removal of the vertices of $T$ actually removed entire subtrees of $T$, which we call \emph{components of $G$ inside $C$}. 

Suppose, for a contradiction, that $Y$ has multiple connected components; refer to~\cref{fig:cycle-is3conn-statements-a}. Consider two vertices $u$ and $v$ that belong to distinct connected components of $Y$ and that are incident to a common face $f_H$ of $\mathcal E_H$. Choose $u$ and $v$ so that their shortest path distance in $T$ is minimum. Since each of $u$ and $v$ has more than one neighbor in $C$, it follows that there exist four edges $(u,w_1)$, $(u,w_2)$, $(v,z_1)$, and $(v,z_2)$ that belong to the boundary of $f_H$; suppose, without loss of generality, that $w_1$, $w_2$, $z_2$, and $z_1$ appear in this clockwise order along $C$. By planarity, the path $P_{uv}$ that connects $u$ and $v$ in $T$ lies inside $f_H$. By the minimality of the shortest path distance between $u$ and $v$, all the internal vertices of~$P_{uv}$ do not belong to~$Y$. Let~$u'$ be the vertex of~$P_{uv}$ adjacent to~$u$. Then~$u'$ is incident to (at least) two faces of $\mathcal E$, one for each ``side of $P_{uv}$''. Formally, consider the bounded region $R_1$ of the plane delimited by the path $(w_1,u)\cup P_{uv}\cup (v,z_1)$ and by the path~$P_1$ in $C$ connecting $z_1$ with~$w_1$ and not passing through $w_2$; then $u'$ is incident to a face $f_1$ of $\mathcal E$ inside $R_1$. Analogously, consider the bounded region $R_2$ of the plane delimited by the path $(w_2,u)\cup P_{uv}\cup (v,z_2)$ and by the path~$P_{2}$ in $C$ connecting $z_2$ with $w_2$ and not passing through $w_1$; then $u'$ is incident to a face $f_2$ of $\mathcal E$ inside $R_2$. Since $T$ is a tree, for $i=1,2$, the boundary of $f_i$ contains at least one vertex $x_i$ of $C$ (possibly $x_i=w_i$ and/or $x_i=z_i$). Then $u'$ is adjacent to $x_i$, as otherwise the edge $(u,x_i)$ could be inserted inside $f_i$, contradicting the maximality of the initial augmentation of $G$. However, this implies that $u'$ is adjacent to two vertices of $C$, namely $x_1$ and $x_2$, thus contradicting the fact that no internal vertex of $P_{uv}$ belongs to $Y$.  
\begin{figure}
    \centering
    \begin{minipage}{0.24\textwidth}
    \centering
        \includegraphics[page=5]{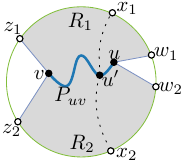}
        \phantomsubcaption{}
        \label{fig:cycle-is3conn-statements-a}
    \end{minipage}
    \hfil
    \begin{minipage}{0.24\textwidth}
    \centering
        \includegraphics[page=6]{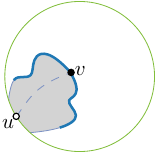}
        \phantomsubcaption{}
        \label{fig:cycle-is3conn-statements-b}
    \end{minipage}
    \hfill
    \caption{(a) Illustration for the proof that $Y$ is a tree; the face $f_H$ of $\mathcal E_H$ is shaded gray. (b) Illustration for the proof that every vertex $u$ of $C$ has a neighbor in $Y$; the face $g_H$ is shaded gray.}
    \label{fig:cycle-is3conn-statements}
\end{figure}
%

We now prove that $H$ is $3$-connected. In order to do that, we need the following two auxiliary statements.
\begin{itemize}
    \item Statement (S1): Every vertex of $C$ has a neighbor in $Y$. For a contradiction, suppose that a vertex $u$ of $C$ has no neighbor in~$Y$; refer to~\cref{fig:cycle-is3conn-statements-b}. Since $C$ is an induced cycle, it follows that $u$ is incident to a single internal face $g_H$ of $\mathcal E_H$ and at least one vertex~$v$ of~$Y$ is incident to $g_H$. Then the edge $(u,v)$ can be inserted inside $g_H$, contradicting the fact that no edge can be added connecting a vertex of $W$ with a vertex of $T$ inside an internal face of $\mathcal E$. 
    \item Statement (S2): Every vertex of $Y$ has at least two neighbors in $C$. This directly comes from the fact that vertices of $T$ with at most one neighbor in $C$ were removed from $G$.
\end{itemize}   

We are now ready to prove that, for any two vertices $u$ and $v$ of $H$, there exist three paths $Q_1$, $Q_2$, and $Q_3$ between $u$ and $v$ that are vertex-disjoint, except at $u$ and $v$. We distinguish three cases. 
\begin{figure}
    \centering
    \begin{minipage}[t]{0.24\textwidth}
    \centering
        \includegraphics[page=1]{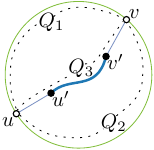}
        \phantomsubcaption{}
        \label{fig:cycle-is3conn-a}
    \end{minipage}
    \begin{minipage}[t]{0.24\textwidth}
    \centering
        \includegraphics[page=2]{Figures/CycleTrees-Is3Conn.pdf}
        \phantomsubcaption{}
        \label{fig:cycle-is3conn-b}
    \end{minipage}
    \begin{minipage}[t]{0.24\textwidth}
    \centering
        \includegraphics[page=3]{Figures/CycleTrees-Is3Conn.pdf}
        \phantomsubcaption{}
        \label{fig:cycle-is3conn-c}
    \end{minipage}
    \begin{minipage}[t]{0.24\textwidth}
    \centering
        \includegraphics[page=4]{Figures/CycleTrees-Is3Conn.pdf}
        \phantomsubcaption{}
        \label{fig:cycle-is3conn-d}
    \end{minipage}
    \caption{Illustration for the proof that $H$ is $3$-connected. In (a) $u$ and $v$ belong to $C$; in (b)  $u$ and $v$ belong to $Y$; in (c) $u$ belongs to $C$, $v$ belongs to $Y$, and $v\neq u'$ or $u\notin \{z_1,z_2\}$; in (d) $u$ belongs to $C$, $v$ belongs to $Y$, $v=u'$, and $z_2=u$.}
    \label{fig:cycle-is3conn}
\end{figure}
\begin{itemize}
    \item Suppose first that $u$ and $v$ belong to $C$, as in~\cref{fig:cycle-is3conn-a}. Then $Q_1$ and $Q_2$ are the paths between $u$ and $v$ composing~$C$. Let $(u,u')$ and $(v,v')$ be two distinct edges such that $u'$ and $v'$ belong to $Y$; these edges exist by statement (S1). Then $Q_3$ is composed of $(u,u')$, of $(v,v')$, and of the path between $u'$ and $v'$ in $Y$.
    \item Suppose next that $u$ and $v$ belong to $Y$, as in~\cref{fig:cycle-is3conn-b}. Then $Q_1$ is the path in~$Y$ between $u$ and $v$. Let $(u,w_1)$ and $(u,w_2)$ be two distinct edges such that $w_1$ and $w_2$ belong to $C$, and let $(v,z_1)$ and $(v,z_2)$ be two distinct edges such that $z_1$ and $z_2$ belong to $C$; these edges exist by statement (S2). We can assume, without loss of generality, up to renaming $w_1$ with $w_2$, that $C$ contains a path $P_1$ from $w_1$ and $z_1$ not passing through $w_2$ and $z_2$ and a path $P_2$ from $w_2$ and $z_2$ not passing through $w_1$ and $z_1$. Then $Q_2$ is the path $(u,w_1)\cup P_1 \cup (v,z_1)$ and $Q_3$ is the path $(u,w_2)\cup P_2 \cup (v,z_2)$. 
    \item Suppose finally that $u$ belongs to $C$ and $v$ to $Y$, as in~\cref{fig:cycle-is3conn-c,fig:cycle-is3conn-d}. Let $(u,u')$ be an edge such that $u'$ belongs to $Y$; this edge exists by statement (S1). Then $Q_1$ consists of $(u,u')$ and of the path between $u'$ and $v$ in $T$. Let  $(v,z_1)$ and $(v,z_2)$ be two distinct edges such that $z_1$ and $z_2$ belong to $C$; these edges exist by statement (S2). Then $C$ contains a path $P_1$ from $z_1$ to $u$ not passing through $z_2$ (however $z_2$ might be an end-vertex of this path if $z_2=u$) and a path $P_2$ from $z_2$ to $u$ not passing through $z_1$ (however $z_1$ might be an end-vertex of this path if $z_1=u$). If $v\neq u'$ or if $u\notin \{z_1,z_2\}$, as in~\cref{fig:cycle-is3conn-c}, then $Q_2$ is the path $(v,z_1) \cup P_1$ and $Q_3$ is the path $(v,z_2) \cup P_2$. Otherwise, we have $v=u'$ and, say, $z_2=u$, as in~\cref{fig:cycle-is3conn-d}. Consider any vertex $w$ of $C$ different from $u$ and $z_1$.  Let $(w,w')$ be an edge such that $w'$ belongs to $Y$; this edge exists by statement (S1). Then $C$ contains a path $P'_1$ from $z_1$ to $u$ not passing through $w$ and a path $P'_2$ from $w$ to $u$ not passing through $z_1$. We have that $Q_2$ is the path $(v,z_1) \cup P'_1$ and $Q_3$ is the path composed of $P'_2 \cup (w,w')$ and of the path in $Y$ between $w'$ and $v$.
\end{itemize}

This concludes the proof that $H$ is $3$-connected. By~\cref{lem:upper-cycle-tree-3conn}, we have that $H$ admits a weakly leveled planar drawing with  span at most $4$. By~\cref{le:weak-nonweak}, we have that $H$ admits a leveled planar drawing $\Lambda$ with span at most $9$. We now insert $2+2\lceil \log_2 n\rceil$ levels between any two consecutive levels of $\Lambda$, thus turning $\Lambda$ into a  leveled planar drawing with span $O(\log n)$ such that every edge has span greater than or equal to $3+2\lceil \log_2 n\rceil$. We next show how to insert the vertices of the components of $G$ outside and inside $C$ in the levels of $\Lambda$, thus turning $\Lambda$ into a leveled planar drawing $\Gamma$ of $G$ with span $O(\log n)$.

We start with the components of $G$ outside $C$. Consider each of such components, $G^i_W$, and recall that $G^i_W$ is a maximal outerplanar graph such that an edge $(u_i,v_i)$ incident to the outer face of the outerplane embedding of $G^i_W$ is already drawn in $\Lambda$, while all the other vertices and edges of $G^i_W$ need to be added to $\Lambda$. Let $h_i\geq 3+2\lceil \log_2 n\rceil$ be the span of $(u_i,v_i)$ in $\Lambda$. Assume that $y(u_i)<y(v_i)$ and that the edge $(u_i,v_i)$ has the outer face of $\Lambda$ on the right, as the other cases are analogous. It follows from results by Biedl~\cite{DBLP:journals/dcg/Biedl11,DBLP:conf/gd/Biedl14} that $G^i_W$ admits a planar $y$-monotone poly-line grid drawing $\Gamma_i$ such that: 
\begin{itemize}
\item the height of $\Gamma_i$ (i.e., the number of grid rows intersected by $\Gamma_i$) is $h_i+1$; 
\item the $y$-coordinate of $u_i$ (resp.\ of $v_i$) is strictly smaller (resp.\ larger) than the $y$-coordinate of every other vertex of $G^i_W$; and 
\item the edge $(u_i,v_i)$ has the rest of $\Gamma_i$ to the right. 
\end{itemize}
Then interpreting the grid rows of $\Gamma_i$ as levels, we can plug $\Gamma_i$ into $\Lambda$, without crossing the rest of  $\Lambda$.

We next deal with the components of $G$ inside $C$. As proved previously, the vertices of $T$ with two neighbors in~$C$ induce a tree~$Y$, hence the components of~$G$ inside~$C$ are subtrees of~$T$. We are going to insert such components in $\Lambda$ one at a time, in any order. We show how to insert in $\Lambda$ any component $X$ of $G$ inside $C$. By construction, all the vertices of $X$ have at most one neighbor in $C$. This property can be refined as in the following two claims.
\begin{claim} \label{cl:one-to-one}
Every vertex of $X$ is adjacent to exactly one vertex of $C$.   
\end{claim}
\begin{nestedproof}
By construction, every vertex of $X$ is adjacent to at most one vertex of $C$. Suppose, for a contradiction, that a vertex $v$ of $X$ is not adjacent to any vertex of $C$. Let $f_X$ be any face of $\mathcal E$ vertex $v$ is incident to. Since $T$ is a tree, at least one vertex $v'$ of $C$ is incident to $f_X$. Then the edge $(v,v')$ can be inserted inside $f_X$, contradicting the fact that no edge can be added connecting a vertex of $W$ with a vertex of $T$ inside an internal face of $\mathcal E$. 
\end{nestedproof}
\begin{claim} \label{cl:all-the-same}
All the vertices of $X$ are adjacent to the same vertex $u$ of $C$.   
\end{claim}
\begin{figure}
    \centering
    \begin{subfigure}{0.48\textwidth}
    \centering
        \includegraphics[page=1]{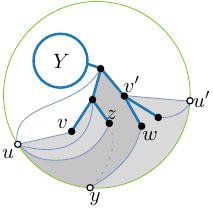}
        \subcaption{}
        \label{fig:insertion-a}
    \end{subfigure}
    \begin{subfigure}{0.48\textwidth}
    \centering
        \includegraphics[page=2]{Figures/Reinsertion.pdf}
        \subcaption{}
        \label{fig:insertion-b}
    \end{subfigure}
    \hfill
    \caption{(a) Illustration for \cref{cl:all-the-same}. The tree $X$ is represented by thick lines, the region $R$ is shaded gray, the face $f_X$ is shaded dark gray. (b) If the edge between the reference vertex $x$ and the base vertex $u$ of a component $X$ did not belong to $G$, then it could be added inside an internal face of $\mathcal E$. }
    \label{fig:insertion}
\end{figure}
\begin{nestedproof}
By \cref{cl:one-to-one}, each vertex of $X$ is adjacent to exactly one vertex of $C$. Suppose, for a contradiction, that two vertices $v$ and $v'$ of $X$ are neighbors of distinct vertices $u$ and $u'$ of $C$, respectively. Refer to \cref{fig:insertion-a}. Let $P$ be the path composed of the edges $(u,v)$, $(u',v')$, and of the path in $X$ between $v$ and $v'$. Consider the two bounded regions of the plane delimited by $C\cup P$. One of them, say $R$, does not contain vertices of $Y$. Let $(u,y)$ be the edge of $C$ on the boundary of $R$ and let $f$ be the face of $\mathcal E$ that is incident to $(u,y)$ and lies inside $R$. Since no vertex of $X$ is neighbor of more than one vertex of $C$, it follows that $R$ has at least two vertices of $X$ on its boundary, say $w$ and $z$. Assume, without loss of generality, that $u$, $y$, $w$ and $z$ appear in this circular order along the boundary of $f$. By the planarity of $\mathcal E$, at least one of the edges $(u,w)$ and $(y,z)$ does not belong to $G$, contradicting the maximality of the initial augmentation of $G$.
\end{nestedproof}
The vertex $u$ of \cref{cl:all-the-same} is called \emph{base vertex} of the component $X$. The next claim is about the relationship of $X$ with the rest of $T$.
\begin{claim} \label{cl:subtree-to-rest}
There exists a unique vertex $x$ in $T$ which does not belong to $X$ and that is a neighbor of a vertex in $X$. Moreover, $x$ belongs to $Y$.
\end{claim}   
\begin{nestedproof}
First, $T$ does not coincide with $X$ since the tree $Y$ contains vertices, by statement~(S1). Since $X$ is a subtree of $T$, there exists a unique vertex $x$ of $T$ that does not belong to~$X$ and that is a neighbor of a vertex in~$X$. If $x$ had at most one neighbor in $C$, it would belong to~$T$, hence it has at least two neighbors in $C$, and so it does belong to $Y$.
\end{nestedproof}
The vertex $x$ of \cref{cl:subtree-to-rest} is called \emph{reference vertex} of the component $X$. Further, the neighbor of $x$ in $X$ is called the \emph{attachment vertex} of $X$. Recall that $u$ is the base vertex of $X$. We have the following.
\begin{claim} \label{cl:connecting-edge}
The edge $(x,u)$ belongs to $G$.
\end{claim}   
\begin{nestedproof}
If the edge $(x,u)$ did not belong to $G$, then it could be added to $G$, right outside the subgraph of $G$ induced by $V(X)\cup \{x,u\}$; see \cref{fig:insertion-b}. However, this would contradict the fact that no edge can be added connecting a vertex of $W$ with a vertex of $T$ inside an internal face of $\mathcal E$. 
\end{nestedproof}



We are now ready to insert $X$ in $\Lambda$. Refer to \cref{fig:comp-x-a}. We root $X$ at its attachment vertex $a$. By \cref{cl:connecting-edge}, the edge $(u,x)$ is in $G$, hence it already belongs to $\Lambda$. Let $\ell$ be the leveling associated with $\Lambda$ and assume that $\ell(u)<\ell(x)$, as the other case is symmetric.  Recall that the span of $(u,x)$ is at least $3+2\lceil \log_2 n\rceil$, hence there are more than $1+\lceil \log_2 n\rceil$ levels larger than $\ell(u)$ and smaller than $\ell(x)$. We label $\ell_1,\dots,\ell_k$ such levels from the highest (the level $\ell(x)-1$) to the lowest (the level $\ell(u)+1$). We draw a $y$-monotone curve $\gamma_{xu}$ from $x$ to $u$ immediately to the right of the edge $(u,x)$. That is, the intersections of $\gamma_{xu}$ with $\ell_1,\dots,\ell_k$ are immediately to the right of the intersections of the edge $(u,x)$ with such lines; further, $\gamma_{xu}$ enters $u$ and $x$ immediately to the right of the edge $(u,x)$. Let $\mathcal R_{xu}$ be the region delimited by the edge $(u,x)$ and by $\gamma_{xu}$. The drawing of $X$ and of the edges connecting the vertices of $X$ with $u$ and $x$ will be contained in the interior of $\mathcal R_{xu}$. This, together with the planarity of such a drawing, implies the planarity of $\Lambda$ after the insertion of $X$. 

We insert $X$ in $\Lambda$ by means of an algorithm derived from a well-known algorithm for constructing planar straight-line grid drawings of trees with height at most $1+\lfloor \log_2 n\rfloor$ \cite{DBLP:journals/dcg/Chan20,DBLP:journals/comgeo/CrescenziBP92,s-lpag-76}.
\begin{figure}
    \centering
    \begin{subfigure}{0.48\textwidth}
    \centering
        \includegraphics[scale=.7,page=3]{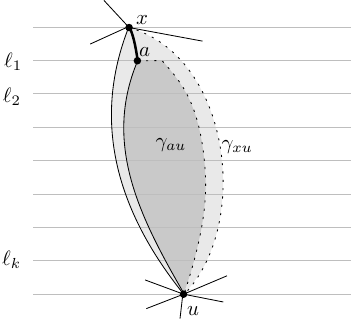}
        \subcaption{}
        \label{fig:comp-x-a}
    \end{subfigure}
    \begin{subfigure}{0.48\textwidth}
    \centering
        \includegraphics[scale=.7,page=4]{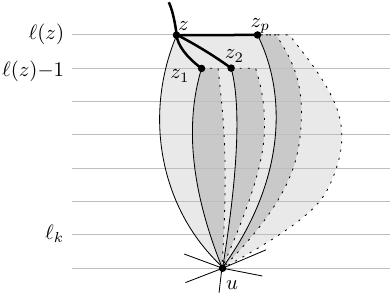}
        \subcaption{}
        \label{fig:comp-x-b}
    \end{subfigure}
    \hfill
    \caption{(a) Start of the construction of the drawing of $X$. Region $\mathcal R_{xu}$ is shaded light and dark gray, region $\mathcal R_{au}$ is shaded dark gray only. (b) Construction of the drawing of $Z$. Region $\mathcal R_{zu}$ is shaded light and dark gray, regions $\mathcal R_{z_iu}$ are shaded dark gray only.}
    \label{fig:comp-x}
\end{figure}
We start by placing $a$ on $\ell_1$ in the interior of $\mathcal R_{xu}$. We draw the edges $(a,x)$ and  $(a,u)$ monotonically in the interior of $\mathcal R_{xu}$. We also draw a curve $\gamma_{au}$ from $a$ to $u$ as follows. The curve $\gamma_{au}$ leaves $a$ horizontally towards the right, up to a point still inside $\mathcal R_{xu}$. Then it proceeds monotonically to $u$ inside $\mathcal R_{xu}$ to the right of the edge $(a,u)$. Let $\mathcal R_{au}$ be the region delimited by the edge $(a,u)$ and by $\gamma_{au}$. The drawing of $X$ and of the edges connecting the vertices of $X$ with $u$ will be contained in (the interior or on the boundary of) $\mathcal R_{au}$. This ensures that the drawing of $X$ and of the edges connecting the vertices of $X$ with $u$ and $x$ is contained in the interior of $\mathcal R_{xu}$.

The algorithm now works recursively. Assume that some tree $Z$ has to be drawn in $\Lambda$, so that the root $z$ of $Z$ is already drawn in $\Lambda$, together with the edge $(z,u)$. Assume also that a region $\mathcal R_{zu}$ has been defined that is delimited by the edge $(z,u)$ on the left and whose remaining boundary consists of a horizontal segment $s_z$ that leaves $z$ towards the right and of a $y$-monotone curve between $u$ and the endpoint of $s_z$ different from $z$. These hypotheses are initially satisfied with $Z=X$, $z=a$, and $\mathcal R_{zu}=\mathcal R_{au}$. We show how to draw $Z$ together with the edges connecting the vertices of $Z$ with $u$ inside or on the boundary of $\mathcal R_{zu}$.

If $z$ is a leaf, then there is nothing to be drawn. Otherwise, let $z_1,\dots,z_p$ be children of $z$ and let $Z_1,\dots,Z_p$ be the subtrees of $Z$ rooted at $z_1,\dots,z_p$, respectively. Assume, without loss of generality, that the number of vertices in $Z_p$ is not smaller than the number of vertices in $Z_i$, for $i=1,\dots,p-1$. We proceed as follows; refer to \cref{fig:comp-x-b}. 

We place $z_1$ on the level $\ell(z)-1$ in the interior of $\mathcal R_{zu}$. Then we draw the edges $(z,z_1)$ 
and $(z_1,u)$ monotonically. We also draw a curve $\gamma_{z_1u}$ from $z_1$ to $u$ as follows. The curve~$\gamma_{z_1u}$ leaves $z_1$ horizontally towards the right, up to a point still inside $\mathcal R_{zu}$. Then it proceeds monotonically to $u$ inside $\mathcal R_{zu}$ to the right of the edge $(z_1,u)$. Let $\mathcal R_{z_1u}$ be the region delimited by the edge $(z_1,u)$ and by $\gamma_{z_1u}$. A drawing of $Z_1$ and of the edges connecting the vertices of $Z_1$ with $u$ is recursively constructed in $\mathcal R_{z_1u}$. This ensures that such a drawing is contained in the interior of $\mathcal R_{zu}$. 
 
For $i=2,\dots,p-1$, we place $z_i$ on the level $\ell(z)-1$, in the interior of $\mathcal R_{zu}$, to the right of~$\gamma_{z_{i-1}u}$. Then we draw the edge $(z,z_i)$ monotonically, in the interior of $\mathcal R_{zu}$, and to the right of the edge $(z,z_{i-1})$. We also draw the edge $(z_i,u)$ monotonically, in the interior of $\mathcal R_{zu}$, and to the right of $\gamma_{z_{i-1}u}$. We then draw a curve $\gamma_{z_iu}$ from $z_i$ to $u$, similarly as described previously for $\gamma_{z_1u}$. Let $\mathcal R_{z_iu}$ be the region delimited by the edge $(z_i,u)$ and by $\gamma_{z_iu}$. A drawing of $Z_i$ and of the edges connecting the vertices of $Z_i$ with $u$ is recursively constructed in $\mathcal R_{z_iu}$. This ensures that such a drawing is contained in the interior of $\mathcal R_{zu}$. 

Finally, we place $z_p$ on the level $\ell(z)$, in the interior of the segment $s_z$. We draw the edge $(z,z_p)$ as a horizontal segment that is part of $s_z$. We also draw the edge $(z_p,u)$ monotonically, in the interior of $\mathcal R_{zu}$, and to the right of $\gamma_{z_{p-1}u}$ (the last condition is vacuous if $p=1$), similarly as done for the edges $(z_i,u)$ with $i=1,\dots,p-1$. We also draw a curve $\gamma_{z_pu}$ from~$z_p$ to $u$ as follows. The curve $\gamma_{z_pu}$ leaves~$z_p$ horizontally towards the right, up to a point still in the interior of $s_z$. Then it proceeds monotonically to $u$ inside $\mathcal R_{zu}$ to the right of the edge $(z_p,u)$. Let $\mathcal R_{z_pu}$ be the region delimited by the edge $(z_p,u)$ and by $\gamma_{z_pu}$. A drawing of~$Z_p$ and of the edges connecting the vertices of~$Z_p$ with $u$ is recursively constructed in the interior or on the boundary of $\mathcal R_{z_pu}$. This ensures that such a drawing is contained in the interior of $\mathcal R_{zu}$. 

Notice that the regions $\mathcal R_{z_1u},\dots,\mathcal R_{z_pu}$ are disjoint, except at $u$, by construction, hence drawings of distinct subtrees of $z$ do not cross one another. It remains to observe that the levels $\ell_1,\dots\ell_k$ are sufficiently many for the recursion to work. This comes from an analysis similar to the one of the height of the cited planar straight-line grid drawings of trees \cite{DBLP:journals/dcg/Chan20,DBLP:journals/comgeo/CrescenziBP92,s-lpag-76}. Indeed, adopting the previous notation, the number of levels required for drawing $Z$ is the maximum between the number of levels required for drawing $Z_p$, which contains at least one vertex less than $Z$, and $1$ plus the number of levels required for drawing any of $Z_1,\dots,Z_{p-1}$. However, because $Z_p$ contains at least as many vertices as any of $Z_1,\dots,Z_{p-1}$, each of such subtrees has less than half of the vertices of $Z$. Hence the number $l(n)$ of levels required for drawing an $n$-vertex tree satisfies the recurrence relationship $l(n)=\max\{l(n-1),1+l(\frac{n-1}{2})\}$, which solves to $l(n)\leq 1+\log_2(n)$. The proof is concluded by recalling that there are more than $1+\lceil \log_2 n\rceil$ levels larger than $\ell(u)$ and smaller than~$\ell(x)$.
\end{proof}

\begin{figure}
    \centering
     \includegraphics[]{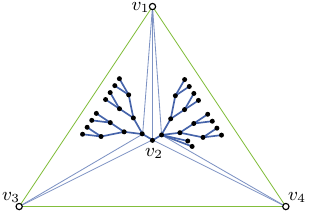}
    \caption{An $n$-vertex cycle-tree requiring $\Omega(\log n)$ span in every weakly leveled planar drawing. The graph resulting from the removal of the vertices incident to the outer face is drawn bold.}
    \label{fig:lower-cycle-tree}
\end{figure}

\begin{theorem} \label{lem:lower-cycle-tree-general}
There exists an $n$-vertex cycle-tree $G$ such that every weakly leveled planar drawing of $G$ has span in $\Omega(\log n)$.
\end{theorem}
\begin{proof}
The $n$-vertex cycle-tree $G$ proving the lower bound is constructed as follows; see~\cref{fig:lower-cycle-tree}. Starting from a complete graph on four vertices $\{v_1, v_2, v_3, v_4\}$, insert two $k$-vertex complete binary trees $T$ and $T'$, where $k+1$ is the largest power of two such that $k\leq \frac{n-4}{2}$. Add $n-2k-4$ further leaves adjacent to the root of $T'$, so that $G$ has $n$ vertices. Finally, connect the root of $T$ to $v_1$, $v_2$, and $v_3$, and the root of $T'$ to $v_1$, $v_2$, and $v_4$. 

In any planar drawing of $G$, we have that $T$ lies inside the cycle $(v_1,v_2,v_3)$, or $T'$ lies inside the cycle $(v_1,v_2,v_4)$, or both. Hence, it suffices to prove that any weakly leveled planar drawing of $T$ (and $T'$) has {\em height} $\Omega(\log n)$, in order to prove that any weakly leveled planar drawing of $G$ has {\em span} $\Omega(\log n)$; indeed, if, say, $T$ lies inside the cycle $(v_1,v_2,v_3)$, then at least two edges of this cycle have span $\Omega(\log n)$.

Suppose that $T$ has a weakly leveled planar drawing with height $h$, for some $h>0$. Then it admits a leveled planar drawing with height at most $2h+1$, by \cref{le:weak-nonweak}. From this, it follows that $T$ has pathwith at most $2h+1$~\cite{DBLP:journals/algorithmica/DujmovicFKLMNRRWW08}. However, the pathwidth of $T$ is in $\Omega(\log n)$, since $T$ has $k\in \Omega(n)$ nodes and a complete binary tree with $k$ nodes has pathwidth $\Omega(\log k)$; see~\cite{DBLP:journals/tcs/Bodlaender98,scheffler}. It follows that $h\in \Omega(\log n)$, as desired. 
\end{proof}

    \subsection{Planar Graphs with Treewidth 2}
In this section, we show that sub-linear span can be achieved for planar graphs with treewidth~2. 
Note that this is not possible for planar graphs of larger treewidth, as the graph in \cref{fig:lower-2Outerplanar} has treewidth three and requires span $\Omega(n)$.

\begin{theorem}\label{lem:series-parallel-upper-bound}
Let $G$ be an $n$-vertex planar graph with treewidth 2. Then $G$ has an $s$-span weakly leveled planar drawing such that $s \in O(\sqrt n)$.
\end{theorem}

\begin{proof}
Biedl~\cite{DBLP:journals/dcg/Biedl11} proved that every $n$-vertex planar graph $G$ with treewidth 2 admits a planar $y$-monotone grid drawing $\Gamma$ with $O(\sqrt n)$ {\em height}, that is, the drawing touches $O(\sqrt n)$ horizontal grid lines. Interpreting the placement of the vertices along these lines as a leveling shows that $G$ admits a leveled planar drawing $\Gamma$ with height, and hence span, $O(\sqrt n)$.
\end{proof}

Since graphs that admit $s$-span weakly leveled planar drawings have planar edge-length ratio at most $2s+1$~\cite[Lemma~4]{digiacomo2023new}, we obtain the following result a corollary of~\cref{lem:series-parallel-upper-bound}, improving upon a previous $O(n^{0.695})$ bound by Borrazzo and Frati~\cite{DBLP:journals/jocg/BorrazzoF20}

\begin{corollary} \label{cor:series-edgelength}
Treewidth-2 graphs with $n$ vertices have planar edge-length ratio $O(\sqrt n)$.  
\end{corollary}

The following lower bound uses a construction by Frati~\cite{DBLP:conf/gd/Frati10,DBLP:journals/dmtcs/Frati10}. Note that $2^{\Omega(\sqrt{\log n})}$ is larger than any poly-logarithmic function of $n$, but smaller than any polynomial function of~$n$.

\begin{theorem} \label{lem:series-parallel-lower-bound}
There exists an $n$-vertex planar graph $G$ with treewidth 2 such that every weakly leveled planar drawing of $G$ has span in $2^{\Omega(\sqrt{\log n})}$.
\end{theorem}

\begin{proof}
Frati~\cite{DBLP:conf/gd/Frati10,DBLP:journals/dmtcs/Frati10} proved that there is an $m$-vertex planar graph $H_m$ with treewidth 2 requiring $2^{\Omega(\sqrt{\log m})}$ height in every planar $y$-monotone (in fact, in every poly-line) grid drawing. Our lower bound graph $G$ is constructed from a complete bipartite graph $K_{2,3}$ with vertex families $\{u_1,u_2\}$ and $\{v_1,v_2,v_3\}$, by inserting three copies $H^1_k$, $H^2_k$, and $H^3_k$ of $H_k$, with $k=\lfloor \frac{n-2}{3} \rfloor$, so that a vertex of $H^i_k$ is identified with $v_i$, for $i=1,2,3$. If $n\not\equiv 2 \bmod 3$, one or two further vertices can be made adjacent to any vertex of $G$, so that it has $n$ vertices.

In any weakly leveled planar drawing $\Gamma$ of $G$, for some $i,j,\ell$ such that $\{i,j,\ell\}=\{1,2,3\}$, we have that $H^i_k$ lies inside the $4$-cycle $(u_1,v_j,u_2,v_\ell)$. Hence, it suffices to prove that the drawing $\Gamma^i$ of $H^i_k$ in $\Gamma$ has {\em height} $2^{\Omega(\sqrt{\log n})}$, in order to prove that $\Gamma$ has {\em span} $2^{\Omega(\sqrt{\log n})}$, as the former implies that one of the edges of the $4$-cycle $(u_1,v_j,u_2,v_\ell)$ has span $2^{\Omega(\sqrt{\log n})}$. 

Suppose that $\Gamma^i$ has height $h$. By~\cref{le:weak-nonweak}, there exists a leveled planar drawing of $H^i_k$ with height at most $2h$. By the results of Biedl~\cite{DBLP:conf/gd/Biedl14,DBLP:journals/dcg/Biedl11}, we have that $H^i_k$ admits a planar $y$-monotone grid drawing with height at most $2h$. Frati's lower bound~\cite{DBLP:journals/dmtcs/Frati10} then implies that $h\in 2^{\Omega(\sqrt{\log k})}$. The proof is completed by observing that $k\in \Omega(n)$.   
\end{proof}

\section {Conclusions and Open Problems}\label{se:conclusion}
We studied $s$-span weakly leveled planar drawings 
both 
from an algorithmic and a combinatorial perspective. 
We proved the para-NP-hardness of the \wlpF{s} problem with respect to the parameter $s$ and considered its parameterized complexity with respect to structural graph parameters, as the vertex cover number and the treedepth. 
On the way of proving our result for the treedepth, we looked at problem instances with a bounded-size modulator to bounded-size components. 
Since outerplanar and Halin graphs are known to have $1$-span weakly leveled planar drawings, we also considered graph families that have outerplanarity $2$ or treewidth $2$. We proved a linear lower bound for $2$-outerplanar graphs and tight bounds for cycle-trees. 
We also gave upper and lower bounds on the span of weakly leveled planar drawings of graphs with treewidth at most $2$.
%
We conclude by listing natural open problems arising from~our~research: 
\begin{itemize}
    \item 
~Does \wlpF{s} have a kernel of polynomial size when parameterized by the treedepth? Is the problem FPT with respect to the treewidth?
\item \cref{lem:series-parallel-upper-bound,lem:series-parallel-lower-bound} show a gap between the upper and lower bounds in the span for the family of 2-trees. It would be interesting to reduce and possibly close this gap.
\item It would also be interesting to close the gap between the lower bound of $\Omega(\log n)$~\cite{DBLP:conf/gd/Blazej0L20,DBLP:journals/ijcga/BlazjFL21} and the upper bound of $O(\sqrt{n})$~of \cref{cor:series-edgelength} on the edge-length ratio of 2-trees.
\end{itemize}





\bibliography{references}

\begin{thebibliography}{10}

\bibitem{BalkoCG00V022}
Martin Balko, Steven Chaplick, Robert Ganian, Siddharth Gupta, Michael
  Hoffmann, Pavel Valtr, and Alexander Wolff.
\newblock Bounding and computing obstacle numbers of graphs.
\newblock In Shiri Chechik, Gonzalo Navarro, Eva Rotenberg, and Grzegorz
  Herman, editors, {\em 30th Annual European Symposium on Algorithms, {ESA}
  2022, September 5-9, 2022, Berlin/Potsdam, Germany}, volume 244 of {\em
  LIPIcs}, pages 11:1--11:13. Schloss Dagstuhl - Leibniz-Zentrum f{\"{u}}r
  Informatik, 2022.
\newblock \href {https://doi.org/10.4230/LIPICS.ESA.2022.11}
  {\path{doi:10.4230/LIPICS.ESA.2022.11}}.

\bibitem{DBLP:journals/algorithmica/BannisterDDEW19}
Michael~J. Bannister, William~E. Devanny, Vida Dujmović, David Eppstein, and
  David~R. Wood.
\newblock Track layouts, layered path decompositions, and leveled planarity.
\newblock {\em Algorithmica}, 81(4):1561--1583, 2019.
\newblock \href {https://doi.org/10.1007/s00453-018-0487-5}
  {\path{doi:10.1007/s00453-018-0487-5}}.

\bibitem{DBLP:conf/dagstuhl/BastertM99}
Oliver Bastert and Christian Matuszewski.
\newblock Layered drawings of digraphs.
\newblock In Michael Kaufmann and Dorothea Wagner, editors, {\em Drawing
  Graphs, Methods and Models}, volume 2025 of {\em Lecture Notes in Computer
  Science}, pages 87--120. Springer, 1999.
\newblock \href {https://doi.org/10.1007/3-540-44969-8_5}
  {\path{doi:10.1007/3-540-44969-8_5}}.

\bibitem{DBLP:conf/gd/BhoreLMN21}
Sujoy Bhore, Giordano {Da Lozzo}, Fabrizio Montecchiani, and Martin
  N{\"{o}}llenburg.
\newblock On the upward book thickness problem: Combinatorial and complexity
  results.
\newblock In Helen~C. Purchase and Ignaz Rutter, editors, {\em Graph Drawing
  and Network Visualization - 29th International Symposium, {GD} 2021,
  T{\"{u}}bingen, Germany, September 14-17, 2021, Revised Selected Papers},
  volume 12868 of {\em Lecture Notes in Computer Science}, pages 242--256.
  Springer, 2021.
\newblock \href {https://doi.org/10.1007/978-3-030-92931-2_18}
  {\path{doi:10.1007/978-3-030-92931-2_18}}.

\bibitem{BhoreLMN23}
Sujoy Bhore, Giordano {Da Lozzo}, Fabrizio Montecchiani, and Martin
  N{\"{o}}llenburg.
\newblock On the upward book thickness problem: Combinatorial and complexity
  results.
\newblock {\em Eur. J. Comb.}, 110:103662, 2023.
\newblock \href {https://doi.org/10.1016/J.EJC.2022.103662}
  {\path{doi:10.1016/J.EJC.2022.103662}}.

\bibitem{DBLP:conf/gd/BhoreGMN19}
Sujoy Bhore, Robert Ganian, Fabrizio Montecchiani, and Martin N{\"{o}}llenburg.
\newblock Parameterized algorithms for book embedding problems.
\newblock In Daniel Archambault and Csaba~D. T{\'{o}}th, editors, {\em Graph
  Drawing and Network Visualization - 27th International Symposium, {GD} 2019,
  Prague, Czech Republic, September 17-20, 2019, Proceedings}, volume 11904 of
  {\em Lecture Notes in Computer Science}, pages 365--378. Springer, 2019.
\newblock \href {https://doi.org/10.1007/978-3-030-35802-0_28}
  {\path{doi:10.1007/978-3-030-35802-0_28}}.

\bibitem{BhoreGMN20}
Sujoy Bhore, Robert Ganian, Fabrizio Montecchiani, and Martin N{\"{o}}llenburg.
\newblock Parameterized algorithms for book embedding problems.
\newblock {\em J. Graph Algorithms Appl.}, 24(4):603--620, 2020.
\newblock \href {https://doi.org/10.7155/JGAA.00526}
  {\path{doi:10.7155/JGAA.00526}}.

\bibitem{DBLP:conf/gd/BhoreGMN20}
Sujoy Bhore, Robert Ganian, Fabrizio Montecchiani, and Martin N{\"{o}}llenburg.
\newblock Parameterized algorithms for queue layouts.
\newblock In David Auber and Pavel Valtr, editors, {\em Graph Drawing and
  Network Visualization - 28th International Symposium, {GD} 2020, Vancouver,
  BC, Canada, September 16-18, 2020, Revised Selected Papers}, volume 12590 of
  {\em Lecture Notes in Computer Science}, pages 40--54. Springer, 2020.
\newblock \href {https://doi.org/10.1007/978-3-030-68766-3_4}
  {\path{doi:10.1007/978-3-030-68766-3_4}}.

\bibitem{BhoreGMN22}
Sujoy Bhore, Robert Ganian, Fabrizio Montecchiani, and Martin N{\"{o}}llenburg.
\newblock Parameterized algorithms for queue layouts.
\newblock {\em J. Graph Algorithms Appl.}, 26(3):335--352, 2022.
\newblock \href {https://doi.org/10.7155/JGAA.00597}
  {\path{doi:10.7155/JGAA.00597}}.

\bibitem{DBLP:journals/dcg/Biedl11}
Therese Biedl.
\newblock Small drawings of outerplanar graphs, series-parallel graphs, and
  other planar graphs.
\newblock {\em Discret. Comput. Geom.}, 45(1):141--160, 2011.
\newblock \href {https://doi.org/10.1007/s00454-010-9310-z}
  {\path{doi:10.1007/s00454-010-9310-z}}.

\bibitem{DBLP:conf/gd/Biedl14}
Therese Biedl.
\newblock Height-preserving transformations of planar graph drawings.
\newblock In Christian~A. Duncan and Antonios Symvonis, editors, {\em 22nd
  International Symposium on Graph Drawing, {GD} 2014, September 24-26, 2014,
  W{\"{u}}rzburg, Germany}, volume 8871 of {\em Lecture Notes in Computer
  Science}, pages 380--391. Springer, 2014.
\newblock \href {https://doi.org/10.1007/978-3-662-45803-7_32}
  {\path{doi:10.1007/978-3-662-45803-7_32}}.

\bibitem{DBLP:conf/gd/Blazej0L20}
V{\'{a}}clav Blazej, Jir{\'{\i}} Fiala, and Giuseppe Liotta.
\newblock On the edge-length ratio of 2-trees.
\newblock In David Auber and Pavel Valtr, editors, {\em Graph Drawing and
  Network Visualization - 28th International Symposium, {GD} 2020, Vancouver,
  BC, Canada, September 16-18, 2020, Revised Selected Papers}, volume 12590 of
  {\em Lecture Notes in Computer Science}, pages 85--98. Springer, 2020.
\newblock \href {https://doi.org/10.1007/978-3-030-68766-3_7}
  {\path{doi:10.1007/978-3-030-68766-3_7}}.

\bibitem{DBLP:journals/ijcga/BlazjFL21}
V{\'{a}}clav Blazj, Jir{\'{\i}} Fiala, and Giuseppe Liotta.
\newblock On edge-length ratios of partial 2-trees.
\newblock {\em Int. J. Comput. Geom. Appl.}, 31(2-3):141--162, 2021.
\newblock \href {https://doi.org/10.1142/S0218195921500072}
  {\path{doi:10.1142/S0218195921500072}}.

\bibitem{DBLP:journals/tcs/Bodlaender98}
Hans~L. Bodlaender.
\newblock A partial $k$-arboretum of graphs with bounded treewidth.
\newblock {\em Theor. Comput. Sci.}, 209(1-2):1--45, 1998.
\newblock \href {https://doi.org/10.1016/S0304-3975(97)00228-4}
  {\path{doi:10.1016/S0304-3975(97)00228-4}}.

\bibitem{DBLP:conf/gd/BonichonFM04}
Nicolas Bonichon, Stefan Felsner, and Mohamed Mosbah.
\newblock Convex drawings of 3-connected plane graphs.
\newblock In J{\'{a}}nos Pach, editor, {\em Graph Drawing, 12th International
  Symposium, {GD} 2004, New York, NY, USA, September 29 - October 2, 2004,
  Revised Selected Papers}, volume 3383 of {\em Lecture Notes in Computer
  Science}, pages 60--70. Springer, 2004.
\newblock \href {https://doi.org/10.1007/978-3-540-31843-9_8}
  {\path{doi:10.1007/978-3-540-31843-9_8}}.

\bibitem{DBLP:journals/algorithmica/BonichonFM07}
Nicolas Bonichon, Stefan Felsner, and Mohamed Mosbah.
\newblock Convex drawings of 3-connected plane graphs.
\newblock {\em Algorithmica}, 47(4):399--420, 2007.
\newblock \href {https://doi.org/10.1007/s00453-006-0177-6}
  {\path{doi:10.1007/s00453-006-0177-6}}.

\bibitem{DBLP:conf/gd/BorrazzoF19}
Manuel Borrazzo and Fabrizio Frati.
\newblock On the edge-length ratio of planar graphs.
\newblock In Daniel Archambault and Csaba~D. T{\'{o}}th, editors, {\em Graph
  Drawing and Network Visualization - 27th International Symposium, {GD} 2019,
  Prague, Czech Republic, September 17-20, 2019, Proceedings}, volume 11904 of
  {\em Lecture Notes in Computer Science}, pages 165--178. Springer, 2019.
\newblock \href {https://doi.org/10.1007/978-3-030-35802-0_13}
  {\path{doi:10.1007/978-3-030-35802-0_13}}.

\bibitem{DBLP:journals/jocg/BorrazzoF20}
Manuel Borrazzo and Fabrizio Frati.
\newblock On the planar edge-length ratio of planar graphs.
\newblock {\em J. Comput. Geom.}, 11(1):137--155, 2020.
\newblock \href {https://doi.org/10.20382/jocg.v11i1a6}
  {\path{doi:10.20382/jocg.v11i1a6}}.

\bibitem{DBLP:conf/gd/BrucknerKM19}
Guido Br{\"{u}}ckner, Nadine~Davina Krisam, and Tamara Mchedlidze.
\newblock Level-planar drawings with few slopes.
\newblock In Daniel Archambault and Csaba~D. T{\'{o}}th, editors, {\em Graph
  Drawing and Network Visualization - 27th International Symposium, {GD} 2019,
  Prague, Czech Republic, September 17-20, 2019, Proceedings}, volume 11904 of
  {\em Lecture Notes in Computer Science}, pages 559--572. Springer, 2019.
\newblock \href {https://doi.org/10.1007/978-3-030-35802-0_42}
  {\path{doi:10.1007/978-3-030-35802-0_42}}.

\bibitem{DBLP:journals/algorithmica/BrucknerKM22}
Guido Br{\"{u}}ckner, Nadine~Davina Krisam, and Tamara Mchedlidze.
\newblock Level-planar drawings with few slopes.
\newblock {\em Algorithmica}, 84(1):176--196, 2022.
\newblock URL: \url{https://doi.org/10.1007/s00453-021-00884-x}, \href
  {https://doi.org/10.1007/S00453-021-00884-X}
  {\path{doi:10.1007/S00453-021-00884-X}}.

\bibitem{DBLP:journals/dcg/Chan20}
Timothy~M. Chan.
\newblock Tree drawings revisited.
\newblock {\em Discret. Comput. Geom.}, 63(4):799--820, 2020.
\newblock \href {https://doi.org/10.1007/S00454-019-00106-W}
  {\path{doi:10.1007/S00454-019-00106-W}}.

\bibitem{DBLP:journals/algorithmica/ChaplickLGLM24}
Steven Chaplick, Giordano {Da Lozzo}, Emilio {Di Giacomo}, Giuseppe Liotta, and
  Fabrizio Montecchiani.
\newblock Planar drawings with few slopes of {H}alin graphs and nested
  pseudotrees.
\newblock {\em Algorithmica}, 86(8):2413--2447, 2024.
\newblock URL: \url{https://doi.org/10.1007/s00453-024-01230-7}, \href
  {https://doi.org/10.1007/S00453-024-01230-7}
  {\path{doi:10.1007/S00453-024-01230-7}}.

\bibitem{DBLP:conf/compgeom/ChaplickGFGRS22}
Steven Chaplick, Emilio~Di Giacomo, Fabrizio Frati, Robert Ganian,
  Chrysanthi~N. Raftopoulou, and Kirill Simonov.
\newblock Parameterized algorithms for upward planarity.
\newblock In Xavier Goaoc and Michael Kerber, editors, {\em 38th International
  Symposium on Computational Geometry, SoCG 2022, June 7-10, 2022, Berlin,
  Germany}, volume 224 of {\em LIPIcs}, pages 26:1--26:16. Schloss Dagstuhl -
  Leibniz-Zentrum f{\"{u}}r Informatik, 2022.
\newblock \href {https://doi.org/10.4230/LIPICS.SOCG.2022.26}
  {\path{doi:10.4230/LIPICS.SOCG.2022.26}}.

\bibitem{DBLP:conf/wads/ChaplickLGLM21}
Steven Chaplick, Giordano~Da Lozzo, Emilio~Di Giacomo, Giuseppe Liotta, and
  Fabrizio Montecchiani.
\newblock Planar drawings with few slopes of {H}alin graphs and nested
  pseudotrees.
\newblock In Anna Lubiw and Mohammad~R. Salavatipour, editors, {\em 17th
  Algorithms and Data Structures Symposium, {WADS} 2021, August 9-11, 2021,
  Halifax, Nova Scotia, Canada}, volume 12808 of {\em Lecture Notes in Computer
  Science}, pages 271--285. Springer, 2021.
\newblock \href {https://doi.org/10.1007/978-3-030-83508-8_20}
  {\path{doi:10.1007/978-3-030-83508-8_20}}.

\bibitem{DBLP:journals/ijcga/ChrobakK97}
Marek Chrobak and Goos Kant.
\newblock Convex grid drawings of 3-connected planar graphs.
\newblock {\em Int. J. Comput. Geom. Appl.}, 7(3):211--223, 1997.
\newblock \href {https://doi.org/10.1142/S0218195997000144}
  {\path{doi:10.1142/S0218195997000144}}.

\bibitem{DBLP:journals/comgeo/CrescenziBP92}
Pierluigi Crescenzi, Giuseppe~Di Battista, and Adolfo Piperno.
\newblock A note on optimal area algorithms for upward drawings of binary
  trees.
\newblock {\em Comput. Geom.}, 2:187--200, 1992.
\newblock \href {https://doi.org/10.1016/0925-7721(92)90021-J}
  {\path{doi:10.1016/0925-7721(92)90021-J}}.

\bibitem{DBLP:conf/isaac/LozzoDEJ17}
Giordano {Da Lozzo}, William~E. Devanny, David Eppstein, and Timothy Johnson.
\newblock Square-contact representations of partial 2-trees and triconnected
  simply-nested graphs.
\newblock In Yoshio Okamoto and Takeshi Tokuyama, editors, {\em 28th
  International Symposium on Algorithms and Computation, ISAAC 2017, December
  9-12, 2017, Phuket, Thailand}, volume~92 of {\em LIPIcs}, pages 24:1--24:14.
  Schloss Dagstuhl - Leibniz-Zentrum f{\"{u}}r Informatik, 2017.
\newblock \href {https://doi.org/10.4230/LIPIcs.ISAAC.2017.24}
  {\path{doi:10.4230/LIPIcs.ISAAC.2017.24}}.

\bibitem{DBLP:journals/combinatorica/FraysseixPP90}
Hubert de~Fraysseix, J{\'{a}}nos Pach, and Richard Pollack.
\newblock How to draw a planar graph on a grid.
\newblock {\em Combinatorica}, 10(1):41--51, 1990.
\newblock \href {https://doi.org/10.1007/BF02122694}
  {\path{doi:10.1007/BF02122694}}.

\bibitem{BattistaETT99}
Giuseppe {Di Battista}, Peter Eades, Roberto Tamassia, and Ioannis~G. Tollis.
\newblock {\em Graph Drawing: Algorithms for the Visualization of Graphs}.
\newblock Prentice-Hall, 1999.

\bibitem{BattistaN88}
Giuseppe {Di Battista} and Enrico Nardelli.
\newblock Hierarchies and planarity theory.
\newblock {\em {IEEE} Trans. Syst. Man Cybern.}, 18(6):1035--1046, 1988.
\newblock \href {https://doi.org/10.1109/21.23105}
  {\path{doi:10.1109/21.23105}}.

\bibitem{DBLP:journals/tcs/BattistaT88}
Giuseppe {Di Battista} and Roberto Tamassia.
\newblock Algorithms for plane representations of acyclic digraphs.
\newblock {\em Theor. Comput. Sci.}, 61:175--198, 1988.
\newblock \href {https://doi.org/10.1016/0304-3975(88)90123-5}
  {\path{doi:10.1016/0304-3975(88)90123-5}}.

\bibitem{digiacomo2023new}
Emilio {Di Giacomo}, Walter Didimo, Giuseppe Liotta, Henk Meijer, Fabrizio
  Montecchiani, and Stephen Wismath.
\newblock New bounds on the local and global edge-length ratio of planar
  graphs, 2023.
\newblock {\em Proceedings of the 40th European Workshop on Computational
  Geometry, EuroCG 2024}.
\newblock URL: \url{https://eurocg2024.math.uoi.gr/data/uploads/paper\_18.pdf}.

\bibitem{Diestelbook}
Reinhard Diestel.
\newblock {\em Graph Theory, 4th Edition}, volume 173 of {\em Graduate texts in
  mathematics}.
\newblock Springer, 2012.

\bibitem{dlt-pepg-84}
D.~Dolev, F.~T. Leighton, and H.~Trickey.
\newblock Planar embedding of planar graphs.
\newblock {\em Advances in Computing Research - Volume 2: VLSI Theory}, 1984.

\bibitem{DBLP:journals/algorithmica/DujmovicFKLMNRRWW08}
Vida Dujmović, Michael~R. Fellows, Matthew Kitching, Giuseppe Liotta,
  Catherine McCartin, Naomi Nishimura, Prabhakar Ragde, Frances~A. Rosamond,
  Sue Whitesides, and David~R. Wood.
\newblock On the parameterized complexity of layered graph drawing.
\newblock {\em Algorithmica}, 52(2):267--292, 2008.
\newblock \href {https://doi.org/10.1007/s00453-007-9151-1}
  {\path{doi:10.1007/s00453-007-9151-1}}.

\bibitem{DBLP:journals/jacm/DujmovicJMMUW20}
Vida Dujmović, Gwena{\"{e}}l Joret, Piotr Micek, Pat Morin, Torsten Ueckerdt,
  and David~R. Wood.
\newblock Planar graphs have bounded queue-number.
\newblock {\em J. {ACM}}, 67(4):22:1--22:38, 2020.
\newblock \href {https://doi.org/10.1145/3385731} {\path{doi:10.1145/3385731}}.

\bibitem{DBLP:journals/dmtcs/DujmovicPW04}
Vida Dujmović, Attila P{\'{o}}r, and David~R. Wood.
\newblock Track layouts of graphs.
\newblock {\em Discret. Math. Theor. Comput. Sci.}, 6(2):497--522, 2004.
\newblock \href {https://doi.org/10.46298/DMTCS.315}
  {\path{doi:10.46298/DMTCS.315}}.

\bibitem{DBLP:conf/acsc/EadesMW86}
Peter Eades, Brendan~D. McKay, and Nicholas~C. Wormald.
\newblock On an edge crossing problem.
\newblock In {\em 9th Australian Computer Science Conference, {ACSC} 1986,
  January, 1986, Canberra, Australia}, pages 327--334, 1986.

\bibitem{DBLP:conf/gd/FelsnerLW01}
Stefan Felsner, Giuseppe Liotta, and Stephen~K. Wismath.
\newblock Straight-line drawings on restricted integer grids in two and three
  dimensions.
\newblock In Petra Mutzel, Michael J{\"{u}}nger, and Sebastian Leipert,
  editors, {\em Graph Drawing, 9th International Symposium, {GD} 2001 Vienna,
  Austria, September 23-26, 2001, Revised Papers}, volume 2265 of {\em Lecture
  Notes in Computer Science}, pages 328--342. Springer, 2001.
\newblock \href {https://doi.org/10.1007/3-540-45848-4_26}
  {\path{doi:10.1007/3-540-45848-4_26}}.

\bibitem{DBLP:journals/jgaa/FelsnerLW03}
Stefan Felsner, Giuseppe Liotta, and Stephen~K. Wismath.
\newblock Straight-line drawings on restricted integer grids in two and three
  dimensions.
\newblock {\em J. Graph Algorithms Appl.}, 7(4):363--398, 2003.
\newblock \href {https://doi.org/10.7155/jgaa.00075}
  {\path{doi:10.7155/jgaa.00075}}.

\bibitem{flsz-ktpp-19}
Fedor~V. Fomin, Daniel Lokshtanov, Saket Saurabh, and Meirav Zehavi.
\newblock {\em Kernelization: Theory of Parameterized Preprocessing}.
\newblock Cambridge University Press, 2019.
\newblock \href {https://doi.org/10.1017/9781107415157}
  {\path{doi:10.1017/9781107415157}}.

\bibitem{DBLP:journals/siamcomp/FormannHHKLSWW93}
Michael Formann, Torben Hagerup, James Haralambides, Michael Kaufmann,
  Frank~Thomson Leighton, Antonios Symvonis, Emo Welzl, and Gerhard~J.
  Woeginger.
\newblock Drawing graphs in the plane with high resolution.
\newblock {\em {SIAM} J. Comput.}, 22(5):1035--1052, 1993.
\newblock \href {https://doi.org/10.1137/0222063} {\path{doi:10.1137/0222063}}.

\bibitem{DBLP:conf/gd/Frati10}
Fabrizio Frati.
\newblock Improved lower bounds on the area requirements of series-parallel
  graphs.
\newblock In Ulrik Brandes and Sabine Cornelsen, editors, {\em Graph Drawing -
  18th International Symposium, {GD} 2010, Konstanz, Germany, September 21-24,
  2010. Revised Selected Papers}, volume 6502 of {\em Lecture Notes in Computer
  Science}, pages 220--225. Springer, 2010.
\newblock \href {https://doi.org/10.1007/978-3-642-18469-7_20}
  {\path{doi:10.1007/978-3-642-18469-7_20}}.

\bibitem{DBLP:journals/dmtcs/Frati10}
Fabrizio Frati.
\newblock Lower bounds on the area requirements of series-parallel graphs.
\newblock {\em Discret. Math. Theor. Comput. Sci.}, 12(5):139--174, 2010.
\newblock \href {https://doi.org/10.46298/dmtcs.500}
  {\path{doi:10.46298/dmtcs.500}}.

\bibitem{DBLP:conf/gd/FratiP07}
Fabrizio Frati and Maurizio Patrignani.
\newblock A note on minimum-area straight-line drawings of planar graphs.
\newblock In Seok{-}Hee Hong, Takao Nishizeki, and Wu~Quan, editors, {\em Graph
  Drawing, 15th International Symposium, {GD} 2007, September 24-26, 2007,
  Sydney, Australia}, volume 4875 of {\em Lecture Notes in Computer Science},
  pages 339--344. Springer, 2007.
\newblock \href {https://doi.org/10.1007/978-3-540-77537-9_33}
  {\path{doi:10.1007/978-3-540-77537-9_33}}.

\bibitem{Far48}
István Fáry.
\newblock On straight lines representation of planar graphs.
\newblock {\em Acta Sci. Math. (Szeged)}, 11:229--–233, 1948.

\bibitem{DBLP:journals/siamcomp/HeathR92}
Lenwood~S. Heath and Arnold~L. Rosenberg.
\newblock Laying out graphs using queues.
\newblock {\em {SIAM} J. Comput.}, 21(5):927--958, 1992.
\newblock \href {https://doi.org/10.1137/0221055} {\path{doi:10.1137/0221055}}.

\bibitem{HlinenyS19}
Petr Hlinen{\'{y}} and Abhisekh Sankaran.
\newblock Exact crossing number parameterized by vertex cover.
\newblock In Daniel Archambault and Csaba~D. T{\'{o}}th, editors, {\em Graph
  Drawing and Network Visualization - 27th International Symposium, {GD} 2019,
  September 17-20, 2019, Prague, Czech Republic}, volume 11904 of {\em Lecture
  Notes in Computer Science}, pages 307--319. Springer, 2019.
\newblock \href {https://doi.org/10.1007/978-3-030-35802-0_24}
  {\path{doi:10.1007/978-3-030-35802-0_24}}.

\bibitem{DBLP:conf/cccg/HoffmannKKR14}
Michael Hoffmann, Marc~J. van Kreveld, Vincent Kusters, and G{\"{u}}nter Rote.
\newblock Quality ratios of measures for graph drawing styles.
\newblock In {\em 26th Canadian Conference on Computational Geometry, {CCCG}
  2014, August 11-13, 2014, Halifax, Nova Scotia, Canada}. Carleton University,
  Ottawa, Canada, 2014.
\newblock URL: \url{http://www.cccg.ca/proceedings/2014/papers/paper05.pdf}.

\bibitem{DBLP:conf/stoc/HopcroftW74}
John~E. Hopcroft and J.~K. Wong.
\newblock Linear time algorithm for isomorphism of planar graphs (preliminary
  report).
\newblock In Robert~L. Constable, Robert~W. Ritchie, Jack~W. Carlyle, and
  Michael~A. Harrison, editors, {\em 6th Annual {ACM} Symposium on Theory of
  Computing, STOC 1974, April 30 - May 2, 1974, Seattle, Washington, {USA}},
  pages 172--184. {ACM}, 1974.
\newblock \href {https://doi.org/10.1145/800119.803896}
  {\path{doi:10.1145/800119.803896}}.

\bibitem{JungerLM98}
Michael J{\"{u}}nger, Sebastian Leipert, and Petra Mutzel.
\newblock Level planarity testing in linear time.
\newblock In Sue Whitesides, editor, {\em 6th International Symposium on Graph
  Drawing, GD 1998, August 13-15, 1998, Montr{\'{e}}al, Canada}, volume 1547 of
  {\em Lecture Notes in Computer Science}, pages 224--237. Springer, 1998.
\newblock \href {https://doi.org/10.1007/3-540-37623-2_17}
  {\path{doi:10.1007/3-540-37623-2_17}}.

\bibitem{DBLP:journals/algorithmica/Kant96}
Goos Kant.
\newblock Drawing planar graphs using the canonical ordering.
\newblock {\em Algorithmica}, 16(1):4--32, 1996.
\newblock \href {https://doi.org/10.1007/BF02086606}
  {\path{doi:10.1007/BF02086606}}.

\bibitem{DBLP:conf/gd/LazardLL17}
Sylvain Lazard, William~J. Lenhart, and Giuseppe Liotta.
\newblock On the edge-length ratio of outerplanar graphs.
\newblock In Fabrizio Frati and Kwan{-}Liu Ma, editors, {\em Graph Drawing and
  Network Visualization - 25th International Symposium, {GD} 2017, Boston, MA,
  USA, September 25-27, 2017, Revised Selected Papers}, volume 10692 of {\em
  Lecture Notes in Computer Science}, pages 17--23. Springer, 2017.
\newblock \href {https://doi.org/10.1007/978-3-319-73915-1_2}
  {\path{doi:10.1007/978-3-319-73915-1_2}}.

\bibitem{DBLP:journals/tcs/LazardLL19}
Sylvain Lazard, William~J. Lenhart, and Giuseppe Liotta.
\newblock On the edge-length ratio of outerplanar graphs.
\newblock {\em Theor. Comput. Sci.}, 770:88--94, 2019.
\newblock \href {https://doi.org/10.1016/J.TCS.2018.10.002}
  {\path{doi:10.1016/J.TCS.2018.10.002}}.

\bibitem{lec-aptg-67}
Abraham Lempel, Shimon Even, and Israel Cederbaum.
\newblock An algorithm for planarity testing of graphs.
\newblock In {\em International Symposium on Graph Theory and its Applications,
  Rome, 1966}, pages 215--232. Gordon \& Breach, 1967.

\bibitem{DBLP:journals/siamdm/MalitzP94}
Seth~M. Malitz and Achilleas Papakostas.
\newblock On the angular resolution of planar graphs.
\newblock {\em {SIAM} J. Discret. Math.}, 7(2):172--183, 1994.
\newblock \href {https://doi.org/10.1137/S0895480193242931}
  {\path{doi:10.1137/S0895480193242931}}.

\bibitem{scheffler}
Petra Scheffler.
\newblock {\em Die Baumweite von Graphen als ein Ma\ss\ f\"ur die
  Kompliziertheit algorithmischer Probleme}.
\newblock PhD thesis, Akademie der Wissenschafien der DDR, Berlin, 1989.

\bibitem{s-lpag-76}
Yossi Shiloach.
\newblock {\em Linear and Planar Arrangements of Graphs}.
\newblock PhD thesis, Weizmann Institute of Science, 1976.

\bibitem{DBLP:journals/tsmc/SugiyamaTT81}
Kozo Sugiyama, Shojiro Tagawa, and Mitsuhiko Toda.
\newblock Methods for visual understanding of hierarchical system structures.
\newblock {\em {IEEE} Trans. Syst. Man Cybern.}, 11(2):109--125, 1981.
\newblock \href {https://doi.org/10.1109/TSMC.1981.4308636}
  {\path{doi:10.1109/TSMC.1981.4308636}}.

\bibitem{DBLP:journals/siamcomp/Tamassia87}
Roberto Tamassia.
\newblock On embedding a graph in the grid with the minimum number of bends.
\newblock {\em {SIAM} J. Comput.}, 16(3):421--444, 1987.
\newblock \href {https://doi.org/10.1137/0216030} {\path{doi:10.1137/0216030}}.

\bibitem{Tamassia:13}
Roberto Tamassia, editor.
\newblock {\em Handbook on Graph Drawing and Visualization}.
\newblock Chapman and Hall/CRC, 2013.
\newblock URL: \url{https://cs.brown.edu/people/rtamassi/gdhandbook}.

\bibitem{tutte1963draw}
William~Thomas Tutte.
\newblock How to draw a graph.
\newblock {\em Pro. London Math. Society}, 3(1):743--767, 1963.

\bibitem{DBLP:journals/csr/Zehavi22}
Meirav Zehavi.
\newblock Parameterized analysis and crossing minimization problems.
\newblock {\em Comput. Sci. Rev.}, 45:100490, 2022.
\newblock \href {https://doi.org/10.1016/J.COSREV.2022.100490}
  {\path{doi:10.1016/J.COSREV.2022.100490}}.

\end{thebibliography}


\end{document}